\documentclass[a4paper,11pt, twoside]{article}

\usepackage[none]{hyphenat}
\usepackage{relsize}
\usepackage[normalem]{ulem}
\usepackage{cancel}
\usepackage{xparse}
\usepackage{relsize}
\usepackage{amsfonts}
\usepackage{graphicx}
\usepackage{centernot}
\usepackage{stackengine}
\usepackage{tensor}
\usepackage{graphics}
\usepackage{color, soul}
\usepackage{amsmath, amsthm, slashed}
\usepackage{amssymb}
\usepackage{rotating}
\usepackage{epsfig}
\usepackage{verbatim}
\usepackage[affil-it]{authblk}
\usepackage{rotating}
\usepackage{graphicx}
\usepackage{accents}
\usepackage{enumerate}
\usepackage{mathrsfs}
\usepackage{tikz}
\usepackage[margin=0.5in]{geometry}
\usetikzlibrary{decorations.pathmorphing,calc,intersections}
\usepackage{mathdesign}
\DeclareFontFamily{U}{MnSymbolD}{}
\DeclareSymbolFont{MnSyD}{U}{MnSymbolD}{m}{n}
\SetSymbolFont{MnSyD}{bold}{U}{MnSymbolD}{b}{n}
\DeclareFontShape{U}{MnSymbolD}{m}{n}{
	<-6>  MnSymbolD5
	<6-7>  MnSymbolD6
	<7-8>  MnSymbolD7
	<8-9>  MnSymbolD8
	<9-10> MnSymbolD9
	<10-12> MnSymbolD10
	<12->   MnSymbolD12
}{}
\DeclareFontShape{U}{MnSymbolD}{b}{n}{
	<-6>  MnSymbolD-Bold5
	<6-7>  MnSymbolD-Bold6
	<7-8>  MnSymbolD-Bold7
	<8-9>  MnSymbolD-Bold8
	<9-10> MnSymbolD-Bold9
	<10-12> MnSymbolD-Bold10
	<12->   MnSymbolD-Bold12
}{}
\DeclareMathSymbol{\insmall}{\mathrel}{MnSyD}{"3E}

\DeclareDocumentCommand\lin{m o o} {\accentset{\scalebox{.6}{\mbox{\tiny\textnormal{ (1)}}}}{#1}{_{\IfNoValueF{#2}{#2}}^{\IfNoValueF{#3}{#3}}}}

\newcommand{\mcald}{\mathcal{D}}

\newcommand{\mfM}{\mathfrak{m}}

\newcommand{\mfa}{\mathfrak{a}}

\newcommand{\mscrH}{\mathscr{H}}

\newcommand{\mscrT}{\mathscr{T}}

\newcommand{\mcala}{\slashed{\mathcal{A}}}
\newcommand{\mcalh}{\mathcal{H}}
\newcommand{\mcali}{\mathcal{I}}

\newcommand{\mcall}{\mathcal{L}}
\newcommand{\mcalm}{\mathcal{M}}

\newcommand{\mcalq}{\mathcal{Q}}

\newcommand{\qmm}{\mcalq_{\mathsmaller{\mcalm}}}
\newcommand{\smm}{S^2_{\mathsmaller{\mcalm}}}
\newcommand{\sexd}{\slashed{\ud}}

\newcommand{\mcals}{\mathcal{S}}

\theoremstyle{plain}
\newtheorem{theorem}{Theorem}[subsection] 

\theoremstyle{definition}
\newtheorem{definition}[theorem]{Definition} 

\theoremstyle{proposition}
\newtheorem{proposition}[theorem]{Proposition} 

\theoremstyle{Remark}
\newtheorem{remark}[theorem]{Remark} 

\theoremstyle{Lemma}
\newtheorem{lemma}[theorem]{Lemma} 

\theoremstyle{Corollary}
\newtheorem{corollary}[theorem]{Corollary} 

\newtheorem*{conjecture*}{Conjecture}

\newtheorem*{theorem*}{Theorem}
\newtheorem*{proposition*}{Proposition}
\newtheorem*{corollary*}{Corollary}
\newtheorem*{lemma*}{Lemma*}
\newtheorem*{definition*}{Definition}
\newtheorem*{remark*}{Remark}
\newtheorem*{notation*}{Notation}
\newtheorem*{comment*}{Comment}
\newtheorem{innercustomgeneric}{\customgenericname}
\providecommand{\customgenericname}{}
\newcommand{\newcustomtheorem}[2]{%
	\newenvironment{#1}[1]
	{%
		\renewcommand\customgenericname{#2}%
		\renewcommand\theinnercustomgeneric{##1}%
		\innercustomgeneric
	}
	{\endinnercustomgeneric}
}
\newcustomtheorem{customthm}{Theorem}
\newcustomtheorem{customlemma}{Lemma}

\newcommand{\So}{\mathscr{S}}

\newcommand{\Sos}{\mathscr{D}}
\newcommand{\Soa}{\mathscr{A}}
\newcommand{\flin}{\lin{f}}

\newcommand{\Solution}{\So=\bigg(\gabhatlin, \trgablin, \mling, \sghatlin, \trsglin\bigg)}

\newcommand{\Seed}{\mathscr{D}=\bigg(\oPhilin, \uPhilin, \oPsilin, \uPsilin, \fvulin, \fvblin, \fvslin, \fwulin, \fwblin, \fwslin, \mfM, \mathfrak{a}\bigg)}

\newcommand{\uL}{\underline{L}}
\newcommand{\iplus}{\mathcal{I}^+}

\newcommand{\ubar}[1]{\underaccent{\bar}{#1}}

\newcommand{\uPhilin}{\lin{\ubar{\Phi}}}
\newcommand{\uPsilin}{\lin{\ubar{\Psi}}}

\newcommand{\oPhilin}{\lin{\bar{\Phi}}}
\newcommand{\oPsilin}{\lin{\bar{\Psi}}}

\newcommand{\uomega}{\underline{\omega}}

\newcommand{\alin}{\lin{\mathfrak{a}}}
\newcommand{\mfmlin}{\lin{\mathfrak{m}}}
\newcommand{\mnorm}[1]{|{#1}|_{\mathsmaller{\mcalm}}}

\newcommand{\SiMa}{\mathscr{S}_{\mathsmaller{\mfM, \mathfrak{a}}}}

\newcommand{\SosMa}{\mathscr{D}_{\mathsmaller{\mfM, \mathfrak{a}}}}

\newcommand{\bg}{\boldsymbol{g}}
\newcommand{\bog}{\boldsymbol{\overline{g}}}

\newcommand{\Sf}{\mathring{{\mathscr{S}}}}
\newcommand{\SfMa}{\mathring{{\mathscr{S}}}_{\mathsmaller{\mfM, \mathfrak{a}}}}
\newcommand{\GfMa}{\mathring{{\mathscr{G}}}_{\mathsmaller{\mfM, \mathfrak{a}}}}

\newcommand{\Ke}{\mathscr{K}}

\newcommand{\Ga}{\mathscr{G}}

\newcommand{\Gf}{\mathring{\mathscr{G}}}

\newcommand{\Mgs}{\big(\mcalm, g_M\big)}

\newcommand{\hplus}{\mathcal{H}^+}

\newcommand{\LM}{\Lambda(\mcalm)}

\newcommand{\taus}{\tau^\star}

\newcommand{\pv}{\partial_v}

\newcommand{\n}{\nabla}
\newcommand{\on}{\overline{\nabla}}
\newcommand{\ud}{\, \mathrm{d}}

\newcommand{\mfp}{\mathfrak{p}}

\newcommand{\opmu}{(1+\mu)}
\newcommand{\ommu}{(1-\mu)}
\newcommand{\mcaltf}{\slashed{\mcalt}_f}
\newcommand{\mcaltv}{\slashed{\mcalt}_\xi}
\newcommand{\mcaltt}{\slashed{\mcalt}_\theta}

\newcommand{\trglin}{\textnormal{tr}_g\glin}

\newcommand{\sigmalin}{\lin{\sigma}}
\newcommand{\Philin}{\lin{\Phi}}
\newcommand{\Psilin}{\lin{\Psi}}

\newcommand{\varphilin}{\lin{\varphi}}
\newcommand{\philin}{\lin{\phi}}

\newcommand{\owlin}{\lin{\overline{w}}}

\newcommand{\glin}{\lin{g}}

\newcommand{\qm}{\mcalq}
\newcommand{\qg}{\tilde{g}}

\newcommand{\pt}{\partial_{t^*}}
\newcommand{\pr}{\partial_r}
\newcommand{\qbox}{\widetilde{\Box}}
\newcommand{\qn}{\widetilde{\nabla}}

\newcommand{\qtr}{\textnormal{tr}_{\qg}}

\newcommand{\qhd}{\tilde{\star}}
\newcommand{\qexd}{\tilde{\textnormal{d}}}

\newcommand{\qd}{\tilde{\delta}}

\newcommand{\qepsilon}{\tilde{\epsilon}}

\newcommand{\psilin}{\lin{\psi}}

\newcommand{\trgab}{\textnormal{tr}_{\qg}\tilde{{g}}}

\newcommand{\qV}{\widetilde{V}}

\newcommand{\qf}{\tilde{f}}
\newcommand{\qv}{\tilde{v}}

\newcommand{\isnorm}[1]{\|{#1}\|^2_{\mathsmaller{\Sigma},r}}
\newcommand{\nisnorm}[1]{\|{#1}\|^2_{\mathsmaller{\iplus},\taus}}

\newcommand{\trgablin}{\qtr\lin{\tilde{g}}}

\newcommand{\trtaulin}{\qtr\lin{\tilde{\tau}}}

\newcommand{\etalin}{\lin{\tilde{\eta}}}
\newcommand{\zetalin}{\lin{\tilde{\zeta}}}

\newcommand{\tflin}{\lin{\tilde{f}}}

\newcommand{\hlin}{\lin{\tilde{h}}}
\newcommand{\mcalc}{\mathcal{C}}
\newcommand{\mcalt}{\slashed{\mathcal{T}}}

\newcommand{\gabhat}{\hat{\tilde{g}}}

\newcommand{\qglin}{\lin{\qg}}
\newcommand{\gabhatlin}{\lin{\hat{\tilde{g}}}}

\newcommand{\tauhatlin}{\lin{\hat{\tilde{\tau}}}}

\newcommand{\qmsm}{\mcalq\otimes S}
\newcommand{\mg}{\stkout{g}}
\newcommand{\sghat}{\hat{\slashed{g}}}
\newcommand{\strg}{\slashed{\textnormal{tr}}\slashed{g}}

\newcommand{\stkout}[1]{\ifmmode\text{\sout{\ensuremath{#1}}}\else\sout{#1}\fi}

\newcommand{\momega}{\text{\sout{\ensuremath{\omega}}}}

\newcommand{\mling}{\text{\sout{\ensuremath{\lin{g}}}}}

\DeclareDocumentCommand\adaptable{m o o o} {{#1}_{\mathsmaller{\mcali}^{\IfNoValueF{#3}{[#3]}},\mathsmaller{\szeta}^{\mathsmaller{{\IfNoValueF{#4}{[#4]}}}}}^{\mathsmaller{\IfNoValueF{#2}{[\tilde{#2}]}}}}
\DeclareDocumentCommand\angadaptable{m o o o o} {{#1}_{\mathsmaller{\mcali}^{\IfNoValueF{#3}{[#3]}},\mathsmaller{\szeta}^{\mathsmaller{{\IfNoValueF{#4}{[#4]}}}}}^{\mathsmaller{\IfNoValueF{#2}{[\tilde{#2}]}},\mathsmaller{\IfNoValueF{#5}{[\slashed{#5}]}}}}
\newcommand{\qD}{\tilde{D}}
\DeclareDocumentCommand\linadaptable{m o o o} {\lin{{#1}}_{\mathsmaller{\mcali}^{\IfNoValueF{#3}{[#3]}},\mathsmaller{\szeta}^{\mathsmaller{{\IfNoValueF{#4}{[#4]}}}}}^{\mathsmaller{\IfNoValueF{#2}{[\tilde{#2}]}}}}
\DeclareDocumentCommand\linangadaptable{m o o o o} {\lin{{#1}}_{\mathsmaller{\mcali}^{\IfNoValueF{#3}{[#3]}},\mathsmaller{\szeta}^{\mathsmaller{{\IfNoValueF{#4}{[#4]}}}}}^{\mathsmaller{\IfNoValueF{#2}{[\tilde{#2}]}},\mathsmaller{\IfNoValueF{#5}{[\slashed{#5}]}}}}
\DeclareDocumentCommand\linadaptablepair{m m o o o} {\lin{{#1}}_{\mathsmaller{\mcali}^{\IfNoValueF{#4}{[#4]}},\mathsmaller{\szeta}^{\mathsmaller{{\IfNoValueF{#5}{[#5]}}}}}^{\mathsmaller{\IfNoValueF{#3}{[\tilde{#3}]}}}\btimes\lin{{#2}}_{\mathsmaller{\mcali}^{\IfNoValueF{#4}{[#4]}},\mathsmaller{\szeta}^{\mathsmaller{{\IfNoValueF{#5}{[#5]}}}}}^{\mathsmaller{\IfNoValueF{#3}{[\tilde{#3}]}}}}

\newcommand{\qotimeshat}{\hat{\tilde{\otimes}}}

\newcommand{\sotimeshat}{\hat{\slashed{\otimes}}}

\newcommand{\uwlin}{\lin{\underline{w}}}
\newcommand{\sm}{S}
\newcommand{\sg}{\slashed{g}}

\newcommand{\bC}{\boldsymbol{C}}
\newcommand{\sn}{\slashed{\nabla}}
\newcommand{\slap}{\slashed{\Delta}}

\newcommand{\zslap}{\slashed{\Delta}_\zeta}
\newcommand{\sdo}{\slashed{\mathcal{D}}_1}
\newcommand{\sdt}{\slashed{\mathcal{D}}_2}
\newcommand{\sdso}{\slashed{\mathcal{D}}_1^{\star}}
\newcommand{\sdst}{\slashed{\mathcal{D}}_2^{\star}}
\newcommand{\sdiv}{\slashed{\textnormal{div}}}
\newcommand{\otimeshat}{\hat{\otimes}}
\newcommand{\scurl}{\slashed{\textnormal{curl}}}
\newcommand{\sepsilon}{\slashed{\epsilon}}
\newcommand{\shd}{\slashed{\star}}
\newcommand{\str}{\slashed{\textnormal{tr}}}

\newcommand{\sps}{\slashed{\Pi}_f}
\newcommand{\spv}{\slashed{\Pi}_\xi}
\newcommand{\spt}{\slashed{\Pi}_\theta}
\newcommand{\smcA}[1]{\slashed{\mathcal{A}}_f^{[#1]}}
\newcommand{\vmcA}[1]{\slashed{\mathcal{A}}_\xi^{[#1]}}
\newcommand{\tmcA}[1]{\slashed{\mathcal{A}}_\theta^{[#1]}}

\newcommand{\mfglin}{\lin{\mathfrak{g}}}
\newcommand{\qdL}{\qexd^{\mathsmaller{\mathsmaller{\mathsmaller{\mcali}}}}}
\newcommand{\szetap}[1]{\szeta^{[#1]}}

\newcommand{\szeta}{\slashed{\zeta}}

\newcommand{\mfqlin}{\lin{\mathfrak{q}}}

\newcommand{\sV}{\slashed{V}}

\newcommand{\slf}{\slashed{f}}

\newcommand{\sv}{\slashed{v}}

\newcommand{\trsglin}{\str{\lin{\slashed{g}}}}

\newcommand{\spi}{\slashed{\Pi}}

\newcommand{\sflin}{\lin{\slashed{f}}}

\newcommand{\swlin}{\lin{\slashed{w}}}
\newcommand{\Id}{\textnormal{Id}}

\newcommand{\btimes}{\boldsymbol{\times}}
\newcommand{\sglin}{\lin{\slashed{g}}}
\newcommand{\sghatlin}{\lin{\hat{\slashed{g}}}}

\newcommand{\smi}{S_\nu}
\newcommand{\hplusi}{\hplus\cap\Sigma}

\newcommand{\LS}{\Lambda(\Sigma)}

\newcommand{\sD}{\slashed{D}}
\newcommand{\snnu}{\sn_{\nu}}

\newcommand{\mfT}{\mathfrak{T}}
\newcommand{\mft}{\mathfrak{t}}

\newcommand{\bh}{\bar{h}}
\newcommand{\bk}{\bar{k}}

\newcommand{\fvu}{\ubar{\mathfrak{v}}}
\newcommand{\fwu}{\ubar{\mathfrak{w}}}

\newcommand{\fvs}{\slashed{\mathfrak{v}}}
\newcommand{\fws}{\slashed{\mathfrak{w}}}
\newcommand{\fvulin}{\lin{\ubar{\mathfrak{v}}}}
\newcommand{\fwulin}{\lin{\ubar{\mathfrak{w}}}}
\newcommand{\fvblin}{\lin{\bar{\mathfrak{v}}}}
\newcommand{\fwblin}{\lin{\bar{\mathfrak{w}}}}
\newcommand{\fvslin}{\lin{\slashed{\mathfrak{v}}}}
\newcommand{\fwslin}{\lin{\slashed{\mathfrak{w}}}}
\newcommand{\bmcalm}{\boldsymbol{\mcalm}}

\newcommand{\Nlin}{\lin{N}}
\newcommand{\bblin}{\lin{\bar{b}}}
\newcommand{\bhlin}{\lin{\bar{h}}}
\newcommand{\strhlin}{\str\lin{\slashed{h}}}
\newcommand{\bklin}{\lin{\bar{k}}}

\newcommand{\bslin}{\lin{\bar{s}}}

\newcommand{\llin}{\lin{N}}

\newcommand{\dpart}[1]{{#1}_{l=0,1}}
\newcommand{\mpart}[1]{{#1}_{l=0,1}}

\newcommand{\mhlin}{\lin{\stkout{h}}}
\newcommand{\mklin}{\lin{\stkout{k}}}
\newcommand{\mblin}{\lin{\stkout{b}}}

\newcommand{\mslin}{\lin{\stkout{s}}}

\newcommand{\chihat}{\hat{\chi}}

\newcommand{\qwlin}{\lin{\tilde{w}}}
\newcommand{\opsilin}{\lin{\bar{\psi}}}
\newcommand{\upsilin}{\lin{\ubar{\psi}}}
\newcommand{\gablin}{\lin{\tilde{g}}}

\newcommand{\fvo}{\bar{\mathfrak{v}}}
\newcommand{\fwo}{\bar{\mathfrak{w}}}
\newcommand{\shhatlin}{\lin{\hat{\slashed{h}}}}
\newcommand{\skhatlin}{\lin{\hat{\slashed{k}}}}

\newcommand{\qdot}{\,\tilde{\cdot}\,}
\newcommand{\sdot}{\,\slashed{\cdot}\,}
\newcommand{\trk}{\textnormal{tr}k}
\newcommand{\trklin}{\lin{\big(\textnormal{tr}k}\big)}

\DeclareDocumentCommand\xitau{o} {\Xi_{\tau^*_{\IfNoValueF{#1}{#1}}}}

\DeclareDocumentCommand\eh{o o} {\hplus_{\taus_{\IfNoValueF{#1}{#1}}, \taus_ {\IfNoValueF{#2}{#2}}}}

\DeclareDocumentCommand\ni{o o} {{\mcali}^+_{\taus_{\IfNoValueF{#1}{#1}}, \taus_ {\IfNoValueF{#2}{#2}}}}

\DeclareDocumentCommand\sip{m o} {\big\langle{#1}\big\rangle_{{\IfNoValueF{#2}{#2}}}}

\DeclareDocumentCommand\snorm{m o o} {||{#1}||^2_{L^2_{{\IfNoValueF{#2}{#2}}, {\IfNoValueF{#3}{#3}}}}}

\newcommand{\bvlin}{\lin{\bar{v}}}
\newcommand{\uvlin}{\lin{\ubar{v}}}

\newcommand{\trchi}{\slashed{\textnormal{tr}}\chi}

\DeclareDocumentCommand\bold{m o o} {\boldsymbol{\tensor{#1}{_{\IfNoValueF{#2}{#2}}^{\IfNoValueF{#3}{#3}}}}}

\DeclareDocumentCommand\unbold{m o o} {{\tensor{#1}{_{\IfNoValueF{#2}{#2}}^{\IfNoValueF{#3}{#3}}}}}

\newcommand{\bglin}{\lin{\bar{g}}}
\newcommand{\uglin}{\lin{\ubar{g}}}
\newcommand{\bgprimelin}{\lin{\bar{g}}'}
\newcommand{\ugprimelin}{\lin{\ubar{g}}'}

\DeclareDocumentCommand\tintegral{m o o} {\int_{\tau^*_{\IfNoValueF{#2}{#2}}}^{\IfNoValueF{#3}{#3}}{#1}\ud \overline{\tau}^*}

\DeclareDocumentCommand\fluxintegral{m o} {\int_{\IfNoValueF{#2}{#2}}{#1}\qepsilon|_n\wedge\sepsilon}

\DeclareDocumentCommand\bulkintegral{m o o} {\int_{\tau^*_{\IfNoValueF{#2}{#2}}}^{\tau^*_{\IfNoValueF{#3}{#3}}}\int_{{\Xi_{\overline{\tau}^*}\backslash\sm}}{#1}\qepsilon|_{n_{\Xi_{\overline{\tau}^*}}}\ud\overline{\tau}^*}

\DeclareDocumentCommand\energy{m o o o} {\mathbb{E}_{\IfNoValueF{#3}{#3}}^{\IfNoValueF{#4}{#4}}[#1]{\IfNoValueF{#2}{(\tau^*_{#2})}}}

\DeclareDocumentCommand\initenergy{m o} {\mathbb{E}_{0}^{\IfNoValueF{#2}{#2}}[#1]}

\DeclareDocumentCommand\linnorm{m m o o o} {\lin{\mathbb{#1}}_{\IfNoValueF{#4}{\mathsmaller{#4}}}^{\IfNoValueF{#3}{#3}}[#2]{\IfNoValueF{#5}{({#5})}}}

\DeclareDocumentCommand\linrnorm{m m o o o o} {\lin{\mathbb{#1}}_{\mathsmaller{\mcali^{\IfNoValueF{#4}{#4}}_{\IfNoValueF{#5}{#5}}}}^{\IfNoValueF{#3}{#3}}[#2]{\IfNoValueF{#6}{({#6})}}}

\DeclareDocumentCommand\linrnormdata{m m o o o o} {\lin{\mathbb{#1}}_{\circ^{\mathsmaller{{\IfNoValueF{#4}{#4}}}}_{\mathsmaller{\IfNoValueF{#5}{#5}}}}^{\IfNoValueF{#3}{#3}}[#2]{\IfNoValueF{#6}{({#6})}}}

\DeclareDocumentCommand\norm{m m o o o} {{\mathbb{#1}}_{\IfNoValueF{#4}{\mathsmaller{#4}}}^{\IfNoValueF{#3}{#3}}[#2]{\IfNoValueF{#5}{({#5})}}}

\DeclareDocumentCommand\rnorm{m m o o o o} {{\mathbb{#1}}_{\mathsmaller{\mcali^{\IfNoValueF{#4}{#4}}_{\IfNoValueF{#5}{#5}}}}^{\IfNoValueF{#3}{#3}}[#2]{\IfNoValueF{#6}{({#6})}}}

\DeclareDocumentCommand\rnormdata{m m o o o o} {{\mathbb{#1}}_{\circ^{\mathsmaller{{\IfNoValueF{#4}{#4}}}}_{\mathsmaller{\IfNoValueF{#5}{#5}}}}^{\IfNoValueF{#3}{#3}}[#2]{\IfNoValueF{#6}{({#6})}}}

\DeclareDocumentCommand\linflux{m m o o o} {\lin{\mathbb{#1}}_{\IfNoValueF{#4}{#4}}^{\IfNoValueF{#3}{#3}}[#2]{\IfNoValueF{#5}{({#5})}}}

\DeclareDocumentCommand\linrflux{m m o o o} {\lin{\mathbb{#1}}_{\mathsmaller{\mcali}}^{\IfNoValueF{#3}{#3}}[#2]{\IfNoValueF{#4}{({#4})}}}

\DeclareDocumentCommand\linrrflux{m m o o} {\lin{\mathbb{#1}}_{\mathsmaller{\mcali^2}}^{\IfNoValueF{#3}{#3}}[#2]{\IfNoValueF{#4}{({#4})}}}

\DeclareDocumentCommand\linFlux{m m o o o} {\lin{\textnormal{#1}}_{\IfNoValueF{#4}{#4}}^{\IfNoValueF{#3}{#3}}[#2]{\IfNoValueF{#5}{({#5})}}}

\DeclareDocumentCommand\linrFlux{m m o o o} {\lin{\textnormal{#1}}_{\mathsmaller{\mcali}}^{\IfNoValueF{#3}{#3}}[#2]{\IfNoValueF{#4}{({#4})}}}

\DeclareDocumentCommand\linrrFlux{m m o} {\lin{{\textnormal{#1}}}_{\mathsmaller{\mcali^2}}^{\IfNoValueF{#3}{#3}}[#2]}

\DeclareDocumentCommand\flux{m m o o o} {\mathbb{#1}_{\IfNoValueF{#4}{#4}}^{\IfNoValueF{#3}{#3}}[#2]{\IfNoValueF{#5}{({#5})}}}

\DeclareDocumentCommand\rflux{m m o o o} {\mathbb{#1}_{\mathsmaller{\mcali}}^{\IfNoValueF{#3}{#3}}[#2]{\IfNoValueF{#4}{({#4})}}}

\DeclareDocumentCommand\rrflux{m m o o o} {\mathbb{#1}_{\mathsmaller{\mcali^2}}^{\IfNoValueF{#3}{#3}}[#2]{\IfNoValueF{#4}{({#4})}}}

\DeclareDocumentCommand\Flux{m m o o o} {\textnormal{#1}_{\IfNoValueF{#4}{#4}}^{\IfNoValueF{#3}{#3}}[#2]{\IfNoValueF{#5}{({#5})}}}

\DeclareDocumentCommand\rFlux{m m o o o} {\textnormal{#1}_{\mathsmaller{\mcali}}^{\IfNoValueF{#3}{#3}}[#2]{\IfNoValueF{#4}{({#4})}}}

\DeclareDocumentCommand\rrFlux{m m o o o} {\textnormal{#1}_{\mathsmaller{\mcali^2}}^{\IfNoValueF{#3}{#3}}[#2]{\IfNoValueF{#4}{({#4})}}}

\DeclareDocumentCommand\ied{m o o o} {\mathbb{M}_{\IfNoValueF{#3}{#3}}^{\IfNoValueF{#4}{#4}}[#1]{\IfNoValueF{#2}{(\tau^*_{#2})}}}

\DeclareDocumentCommand\bulk{m o o} {\mathbb{I}^{\IfNoValueF{#3}{#3}}[#1]{\IfNoValueF{#2}{(\tau^*_{#2})}}}

\DeclareDocumentCommand\ID{m o o o} {\mathbb{D}_{\IfNoValueF{#2}{#2}}^{\IfNoValueF{#3}{#3}}[#1]{\IfNoValueF{#4}{(\tau^*_{#4})}}}

\def\QEQ{{%
		\setbox0\hbox{$T$}%
		\rlap{\hbox to \wd0{\hss$\times$\hss}}\box0
}}

\title{The linear stability of the Schwarzschild solution to gravitational perturbations in the generalised wave gauge}
\author{Thomas Johnson\footnote{Email: twj25@cam.ac.uk}}
\affil{Imperial College London, London, UK}

\begin{document}
	
\maketitle

\begin{abstract}
	
In a recent seminal paper \cite{D--H--R} of Dafermos, Holzegel and Rodnianski the linear stability of the Schwarzschild  family of black hole solutions to the Einstein vacuum equations was established by imposing a double null gauge. In this paper we shall prove that the Schwarzschild family is linearly stable as solutions to the Einstein vacuum equations by imposing instead a generalised wave gauge: all sufficiently regular solutions to the system of equations that result from linearising the Einstein vacuum equations, as expressed in a generalised wave gauge, about a fixed Schwarzschild solution remain uniformly bounded on the Schwarzschild exterior region and in fact decay to a member of the linearised Kerr family. The dispersion is at an inverse polynomial rate and therefore in principle sufficient for future nonlinear applications. The result thus fits into the wider goal of establishing the full nonlinear stability of the exterior Kerr family as solutions to the Einstein vacuum equations by employing a generalised wave gauge and therefore complements \cite{D--H--R} in a similar vein as the pioneering work \cite{L--R} of Lindblad and Rodnianski complemented the monumental achievement of Christodoulou and Klainerman in \cite{C--K} whereby the global nonlinear stability of the Minkowski space was established. 

\end{abstract}

\tableofcontents

\section{Introduction}

The celebrated Kerr family \cite{K} of spacetimes, discovered in 1963, comprise a 2-parameter family of solutions to the Einstein vacuum equations of general relativity,
\begin{align}\label{introEE}
\text{Ric}[g]=0.
\end{align}
 They putatively describe isolated gravitating systems that contain a rotating \emph{black hole} -- a region of spacetime which cannot communicate with distant observers.
The parameter $M>0$ thus determines the mass of the black hole whilst the parameter $a$, with\footnote{Black holes with $|a|=M$ are known as \emph{extremal} Kerr black holes and they possess properties that are in stark contrast to their slower rotating cousins. See, for instance, the pioneering \cite{Aretakis}.} $0\leq |a|< M$, measures its angular momentum. The region of spacetime that lies outside the black hole is described by the Kerr exterior family.

The physical reality of such objects, as opposed to being mere mathematical fiction, requires at the very least a positive resolution to the conjectured stability of their exterior\footnote{For the interior region, the situation is drastically different. See the recent remarkable \cite{D--L}.} regions as solutions to the Einstein vacuum equations:
\begin{conjecture*}
The Kerr exterior family is stable as a family of solutions to \eqref{introEE}.
\end{conjecture*}
A precise mathematical formulation of this conjecture can be found\footnote{A discussion on this formulation, along with other weaker formulations of the problem, can be found in the introduction of \cite{D--H--R}.} in \cite{D--H--Rscattering} where it is posed, analogously to the monumental work \cite{C--K} of Christodoulou--Klainerman, in the context of general relativity as a \emph{hyperbolic Cauchy initial value problem}, a correspondence first identified in the pioneering \cite{C-B} of Choquet-Bruhat. Observe in particular  that it is the exterior Kerr family itself which is posited to be stable, as opposed to a single member of this family, for it is expected that a general perturbation contributes both mass and angular momentum to the final state of the black hole.

This conjecture remains open. Indeed, to even attempt to give a positive resolution one must first address the issue of \emph{gauge}. For in the theory of general relativity one can only distinguish between spacetimes up to an equivalence class that is determined by diffeomorphisms. In fact, the Einstein vacuum equations as formulated in \eqref{introEE}  impose an underdetermined system on the spacetime metric $g$, with this degenerancy arising from the diffeomorphism invariance of the theory. It follows that any attempt to resolve the conjecture must necesarily \emph{specify a gauge}. 

In regards to questions pertaining to the issue of stability, particular success has been achieved via the specification of a \emph{wave gauge} in which the hyperbolicity of the Einstein equations is made manifest by reducing \eqref{introEE} to a system of quasilinear wave equations. Indeed, such a gauge was originally employed by Choquet-Bruhat in \cite{C-B} to demonstrate \emph{local} well-posedness of the Einstein equations. That this gauge could in fact be utilised to understand the \emph{global} dynamics of solutions to the Einstein equations was exhibited by Lindblad and Rodnianski in their pioneering\footnote{This result was surprising as there existed convincing heuristics in the literature that such a gauge should become singular globally.} \cite{L--R}, in which they established the nonlinear stability of the trivial solution Minkowski space. Various authors have since extended this by considering either different matter models (see \cite{L--R}, \cite{Sp}, \cite{LeF--M}, \cite{L--T}, \cite{F--J--S}) or different asymptotics (see \cite{Huneau}). See also \cite{L} and \cite{H--Vmink}.
The novelty of this approach due to Lindblad and Rodnianski, as compared to the previously mentioned result \cite{C--K} of Christodoulou and Klainerman, is not only that the proof is dramatically simpler (although one has the caveat of obtaining less detailed asymptotics on the spacetimes constructed) but moreover that it suceeds despite the fact that the Einstein equations in a wave gauge satisfy only the so-called `weak null condition'.

Motivated by this success, it is the intention of this paper to lend credence to the notion that the specification of a \emph{generalised wave gauge} (see \cite{Friedrich} and \cite{C-Bbook} for a precise definition) will be sufficient to resolve the above conjecture in the affirmative. Indeed, one advantage of a generalised wave gauge over the wave gauge is that one can view this gauge as a natural generalisation of the latter to the situation where one is perturbing about a spacetime with non-trivial curvature. This is manifested by the observation that the linearisation of the Einstein equations, as expressed in a generalised wave gauge, \emph{exhibits a particularly amenable structure upon linearisation about a non-trivial background solution}.  Naturally, this structure does not persist if one instead imposes a wave gauge, unless the background one linearises about is the trivial Minkowski spacetime. A more precise definition of the generalised wave gauge can be found in section \ref{OVTheEinsteinequationsinageneralisedwavegauge} of our detailed overview.

In this paper we shall demonstrate the \emph{linear stability} of the Schwarzschild exterior subfamily of the Kerr exterior family with $a=0$ as solutions to \eqref{introEE} under the imposition of a judicious generalised wave gauge:
\begin{theorem*}
All sufficiently regular solutions to the equations of linearised gravity around Schwarzschild i.e. the system of equations that result from linearising the Einstein vacuum equations \eqref{introEE}, as expressed in a (particular and explicit) generalised wave gauge, about a fixed member of the Schwarzschild exterior family
\begin{enumerate}[i)]
	\item remain uniformly bounded on the Schwarzschild exterior
	\item decay to a member of the linearised Kerr family.
\end{enumerate}
\end{theorem*}
Note that the dynamic convergence to a member of the linearised Kerr family is to be understood within the wider context of the stability of the Kerr exterior family. In terms of resolving the conjectured stability of the Kerr family, the generalised wave gauge thus passes the first test put to it by the (less elaborate) Schwarzschild exterior subfamily.

A more comprehensive version of the Theorem is Theorem 1 to be found in section \ref{OVThemaintheoremandoutlineoftheproof} of the overview. However, already at this stage it is proper to discuss the issue surrounding \emph{residual gauge freedom}. This freedom arises from the fact that imposing the generalised wave gauge on a spacetime does not fully specify the gauge. This is directly manifested in the linear theory by the existence of a \emph{residual} class of infinitesimal diffeomorphisms on $\Mgs$ which preserve the generalised wave gauge, thus generating an explicit 
class of solutions to the equations of linearised gravity known as \emph{pure gauge solutions}. The existence of such solutions, along with the presence of the linearised Kerr family, implies that one can only prove a decay statement for solutions to the equations of linearised gravity \emph{up to} the addition of some pure gauge solution and some member of the linearised Kerr family. Indeed, even part i) of our Theorem requires a quantitative gauge-normalisation of initial data. However, this `initial-data-normalisation' is in fact sufficient to obtain part ii) of our Theorem.

The original version of the above Theorem where decay is obtained \emph{only} after the addition of a dynamically determined residual pure gauge solution and in which the Einstein equations are expressed in a \emph{double null gauge} as opposed to a generalised wave gauge was proven by Dafermos, Holzegel and Rodnianski in \cite{D--H--R}. More specifically, their analysis focused on the null decomposed linearised Bianchi equations for the Weyl curvature, coupled to the linearised null structure equations. This approach is therefore in keeping with that of Christodoulou and Klainerman in \cite{C--K}. In particular, the body of work presented here complements that of Dafermos, Holzegel and Rodnianski in a similar vein as to how the result \cite{L--R} of Lindblad and Rodnianski complemented that of Christodoulou and Klainerman.  

One caveat of our Theorem is that the gauge in which we obtain decay is not asymptotically flat. In fact, by modifying our choice of generalised wave gauge it is possible to obtain a Theorem statement as above in addition to which asymptotic flatness is preserved. However, this modification is slightly cumbersome and is relegated to our upcoming \cite{J}. Indeed, the main purpose of this paper is to exhibit in the simplest manner possible the fact that by making a judicious choice of generalised wave gauge one can establish a statement of linear stability with relative ease.

The proof of the theorem relies crucially on the fact that one can extract two fully decoupled scalar wave equations from the equations of linearised gravity. That this is possible is well-known in the literature and corresponds to the remarkable discovery by Regge--Wheeler \cite{RW} and Zerilli \cite{Z} that certain \emph{gauge-invariant} quantities decouple from the full system of linearised Einstein equations into the celebrated Regge--Wheeler and Zerilli equations. Indeed, by combining this decoupling with a sagacious choice of generalised wave gauge one can gauge-normalise initial data in such a way as to cause all linearised metric quantities in this gauge to be fully determined by those two quantities that satisfy the Regge--Wheeler and Zerilli equations respectively. This has the effect of essentially reducing the Theorem to a boundedness and decay statement for solutions to said equations. We remark that this `initial-data-gauge' is nothing but the well-known Regge--Wheeler gauge (see \cite{RW}) adapted to the linearised Einstein equations as they are expressed in a generalised wave gauge. Consequently, \emph{realising this Regge--Wheeler gauge within the framework of a well-posed formulation of linearised gravity comprises one of the key aspects of our work}.

A decay statement for solutions to the Regge--Wheeler equation was established by Holzegel in \cite{Holz}, with earlier results of \cite{B--S} due to Blue and Soffer, whereas a decay statement for solutions to the Zerilli was obtained independently by the author \cite{Me} and Hung--Keller--Wang in \cite{H--K--W}. We also note \cite{A--B--W}. These results regarding solutions to the Regge--Wheeler and Zerilli equations will in fact be applied freely in this paper, although we shall reprove them in a setting more in keeping with the presentation of the main body of this work in our upcoming \cite{Me}. We note that in doing so we rely heavily upon the fundamental techniques developed by Dafermos--Rodnianski in \cite{D--R} and \cite{newmethod} by which one establishes a quantitative rate of dispersion for solutions to the scalar wave equation on the Schwarzschild exterior.

We now discuss other results that are related to our work. We first note that the recent remarkable result \cite{H--V} of Hintz and Vasy whereby the global nonlinear stability of the Kerr--De Sitter family\footnote{These are a family of solutions to $\text{Ric}[g]=\Lambda g$ with $\Lambda>0.$} (\cite{K}, \cite{C}) of black holes was established, for small rotation parameter $a$, proceeded by employing a generalised wave gauge. This was later extended by Hintz in \cite{H} to the Kerr--Newman--De Sitter family (\cite{K}, \cite{C}, \cite{K--N--DS}). A particular feature of this problem however is the presence of a positive cosmological constant which has the effect of making it fundamentally different in nature than one which concerns the nonlinear stability of the Kerr family. In any case, it would be interesting to compare the `initial-data-gauge' we employ in this paper with the (residual) gauge choices that are utilised in the aforementioned works.

Moreover, for a result towards establishing the linear stability of the Schwarzschild exterior which employs a so-called Chandrasekhar gauge in the setting of a ``metric perturbations'' approach, see \cite{H--K--W}. Finally, for further references pertaining to the linearised Einstein equations about the Schwarzschild exterior family, see \cite{Holzconv}, \cite{Moncrief2} and \cite{S--T}--\cite{T}.

We end the introduction with a brief discussion towards nonlinear applications and future work. Indeed, in view of the fact that one must in effect linearise about the solution one expects to approach, providing a positive resolution to the conjectured stability of the Kerr exterior family by utilising a generalised wave gauge, even for the Schwarzschild exterior subfamily, would require upgrading the linear theory established here to the full Kerr exterior family. Nevertheless, in \cite{D--H--R} Dafermos, Holzegel and Rodnianski formulated a restricted nonlinear stability conjecture regarding the Schwarzschild exterior family for which the improved rate of dispersion embodied in part ii) of our Theorem is in principle sufficient, when coupled with the understanding gained by Lindblad--Rodnianski in \cite{L--R}, to treat the nonlinearities present in the system of equations that results from expressing \eqref{introEE} in a generalised wave gauge, thus paving the way for a resolution of this conjecture by means of a generalised wave gauge. A precise formulation of the conjecture can be found in section \ref{OVArestrictednonlinearstabilityconjecture} of the overview. Remarkably, a proof of said conjecture in the symmetry class of axially symmetric and polarised perturations has very recently been announced by Klainerman and Szeftel over a series of three papers, the first of which can be found in \cite{K--S}. 

\section{Overview}\label{Overview}

\noindent We shall now give a complete overview of the paper. The overview is to be divided into five segments.

In the first part, section \ref{OVTheequationsoflinearisedgravityaroundSchwarzschild}, we describe the process behind which one arrives at the equations of linearised gravity around Schwarzschild. In addition, two special classes of solutions to the equations of linearised gravity are discussed, namely the linearised Kerr and pure gauge solutions. 

In the second part, section \ref{OVAneffectivescalarisationoftheequationsoflinearisedgravity}, we discuss how one extracts two fully decoupled scalar wave equations from the equations of linearised gravity. This decoupling will prove vital in establishing the linear stability of the Schwarzschild exterior solution.

In the third part, section \ref{OVTheCauchyproblemfortheequationsoflinearisedgravity}, the well-posedness of the equations of linearised gravity as a Cauchy problem is established. This foundational statement is complicated by the existence of \emph{constraints}.

In the fourth part, section \ref{OVGaugenormalisedsolutionstotheequationsoflinearisedgravity}, we discuss particular solutions to the equations of linearised gravity which differ from a general solution by the addition of certain pure gauge solutions. These \emph{gauge-normalised} solutions will ultimately serve to address the complications arising from the existence of the linearised Kerr and pure gauge solutions.

In the fifth part, section \ref{OVThemaintheoremandoutlineoftheproof}, we discuss our main theorem regarding quantitative boundedness and decay statements for these gauge-normalised solutions. The proof of the theorem, an outline of which shall be given, relies crucially on the decoupling discussed in part two of the overview.

Finally, in section \ref{OVArestrictednonlinearstabilityconjecture}, we state the restricted nonlinear stability conjecture of Dafermos, Holzegel and Rodnianski regarding the Schwarzschild exterior solution for which the results of this paper are in principle sufficient to attempt to resolve by means of a generalised wave gauge.\newline

Before we begin, the author would like to acknowledge that the organisational structure of the overview, and indeed the paper, stays very close to that found in \cite{D--H--R}. This of course facilitates a comparison between the two approaches although needless to say it additionally alleviates some of the structural burden placed upon the author in preparing this document.

\subsection{The equations of linearised gravity around Schwarzschild}\label{OVTheequationsoflinearisedgravityaroundSchwarzschild}

We commence the overview with a discussion as to how one linearises the Einstein equations, as expressed in a generalised wave gauge, about the Schwarzschild exterior spacetime, thus ultimately arriving at the equations of linearised gravity around Schwarzschild. 

It is by deriving quantitative decay bounds on solutions to these equations of linearised gravity that one establishes the linear stability of the Schwarzschild exterior family as solutions to the Einstein vacuum equations in a generalised wave gauge. 
 
 This part of the overview corresponds to sections \ref{ThevacuumEinsteinequationsinageneralisedwavegauge}-\ref{TheequationsoflinearisedgravityaroundSchwarzschild} in the main body of the paper.
 
\subsubsection{The Einstein equations in a generalised wave gauge}\label{OVTheEinsteinequationsinageneralisedwavegauge}

We begin by describing the notion of a generalised wave gauge for an abstract Lorentzian manifold, giving then a description of the vacuum Einstein equations when expressed in this gauge. 

This section of the overview corresponds to section \ref{ThevacuumEinsteinequationsinageneralisedwavegauge} of the main body of the paper.\newline

Let $\big(\bold{\mcalm}, \bold{g}\big)$ and $\big(\bold{\mcalm}, \bold{\overline{g}}\big)$ be $3+1$ globally hyperbolic Lorentzian manifolds with\footnote{Here we recall the notation $T^k(\mcalm)$ for the space of $k$-covariant tensor fields on $\mcalm$.} $\boldsymbol{f}:T^2(\bmcalm)\times T^2(\bmcalm)\rightarrow T\bmcalm$ a smooth map. 

Then following \cite{C-Bbook}, $\boldsymbol{g}$ is said to be in a generalised $\boldsymbol{f}$-wave gauge with respect to $\boldsymbol{\overline{g}}$ iff the identity map
\begin{align*}
\textbf{\textnormal{Id}}:\big(\boldsymbol{\mcalm}, \boldsymbol{g}\big)\rightarrow\big(\boldsymbol{\mcalm}, \boldsymbol{\overline{g}}\big)
\end{align*}
is an $\boldsymbol{f}(\bg, \bog)$-wave map. Denoting by $\bC_{\bg, \bog}$ the connection tensor between $\bg$ and $\bog$, this amounts to
\begin{align}\label{fm}
\boldsymbol{g^{-1}}\cdot\bC_{\bg, \bog}=\boldsymbol{f}(\bg, \bog).
\end{align}

We note that, since the condition \eqref{fm} is equivalent to solving a system of semilinear wave equations, under sufficient regularity one can always\ find an open set $\boldsymbol{U}\subset\boldsymbol{\mcalm}$ such that
\begin{align*}
\textbf{\textnormal{Id}}:\big(\boldsymbol{U}, \boldsymbol{g}\big|_{\boldsymbol{U}}\big)\rightarrow\big(\boldsymbol{U}, \boldsymbol{\overline{g}}\big|_{\boldsymbol{U}}\big)
\end{align*}
is an $\boldsymbol{f}$-wave map. The generalised wave gauge is thus \emph{locally well-posed}. Moreover, observe that if one sets
\begin{align*}
\boldsymbol{f}=0,\qquad\bmcalm=\mathbb{R}^4\qquad\text{and}\qquad\bog=\eta,
\end{align*} 
with $\eta$ the Minkowski metric, then choosing a globally inertial system of coordinates on $\mathbb{R}^4$ one recovers the wave gauge employed so successfully by Lindblad and Rodnianski in \cite{L--R}.

If one assumes that $\bg$ is a generalised $\boldsymbol{f}$-wave gauge with respect to $\bog$ then the Einstein vacuum equations for the metric $\boldsymbol{g}$,
\begin{align*}
{\textbf{Ric}}[\bg]=0,
\end{align*}
reduce to a \emph{quasilinear tensorial wave equation} on $\bg$. A schematic description\footnote{The wave-like nature of this system is a consequence of the Lorentzian character of $\boldsymbol{g}$.} is as follows:
\begin{align}
\Big(\boldsymbol{g^{-1}}\cdot\boldsymbol{\on}^2\Big)\boldsymbol{g}+\bC_{\bg, \bog}\cdot\bC_{\bg, \bog}+\overline{\textbf{Riem}}\cdot\boldsymbol{g}&=\boldsymbol{\mcall}_{\boldsymbol{f}(\bg, \bog)}\boldsymbol{g},\label{m}\\
\boldsymbol{g^{-1}}\cdot\bC_{\bg, \bog}&=\boldsymbol{f}(\bg, \bog)\label{b}.
\end{align}
Here, $\boldsymbol{\overline{\textbf{Riem}}}$ and $\boldsymbol{\on}$ are the Riemann tensor and Levi-Civita connection of $\boldsymbol{\overline{g}}$ respectively.

We note that since the expression $\boldsymbol{f}(\bg, \bog)$ is at the level of the metric $\bg$, then under sufficient regularity the system of equations given by \eqref{m} coupled with \eqref{b}, which corresponding to the \emph{Einstein vacuum equations as expressed in a generalised $\boldsymbol{f}$-wave gauge with respect to $\bog$}, are always locally well-posed as a hyperbolic initial value problem with constrained\footnote{The constraints arise as a consequence of the Gauss--Codazzi equations.} initial data. The generalised wave gauge thus captures the essential hyperbolicity of the Einstein equations. See the book \cite{C-Bbook} of Choquet-Bruhat for details.
\newline

In the introduction, we discussed the conjecture relating to the stability of the Kerr exterior family and how one might seek to resolve it in the affirmative by utilising a generalised wave gauge. One can now make this aim slightly more precise with the statement that one wishes to establish an appropriate notion of stability for solutions to \eqref{m} and \eqref{b} \textbf{for which $\big(\boldsymbol{\mcalm}, \boldsymbol{\overline{g}}\big)$ is set to be any fixed member of the subextremal exterior Kerr family.} This strategy was employed successfully by Hint--Vasy and Vasy in \cite{H--V} and \cite{H} respectively for the case of the Kerr--De Sitter and Kerr--Newman--De Sitter family of black holes with small angular momentum parameter (see also \cite{H--Vmink}).

In this paper, we shall concern ourselves with the instance for which this member resides within the Schwarzschild exterior subfamily. Of course, that still leaves the freedom in specifying the map $\boldsymbol{f}$!\newline\newline

For a non-schematic description of the Einstein vacuum equations as expressed in a generalised wave gauge, see section \ref{TheEinsteinequations}.

\subsubsection{The exterior Schwarzschild spacetime}\label{OVTheexteriorSchwarzschildspacetime}

The Schwarzschild family $\big(\mcals, g_M\big)$, with $M\in\mathbb{R}^+$, constitute the unique family of spherically symmetric solutions to the Einstein vacuum equations. A particularly relevant local description of this family is motivated by the notion that the Einstein equations, as expressed in a generalised wave gauge, are most naturally formulated in terms of a \emph{Cauchy} problem\footnote{In particular, note that the generalised wave gauge condition \eqref{fm} is a condition on the first order derivatives of the metric.}. This suggests describing $g_M$ in a system of coordinates on $\mcals$ which adequately captures the fact that one can foliate $\big(\mcals, g_M\big)$ by Cauchy hypersurfaces. A natural candidate which also describes the event horizon in a regular fashion are the so-called Schwarzschild-star coordinates.

In this coordinate system the metric takes the form
\begin{align}\label{PB}
g_M=-\bigg(1-\frac{2M}{r}\bigg)\ud {t^*}^2+\frac{4M}{r}\ud t^*\ud r+\bigg(1+\frac{2M}{r}\bigg)\ud r^2+r^2\big(\ud\theta^2+\sin^2\theta\ud\varphi^2\big)
\end{align}
where the coordinates take values
\begin{align*}
(t^*, r)\in (-\infty, \infty)\times (0, \infty),\qquad (\theta, \varphi)\in S^2.
\end{align*}
Restricting now the coordinate $r$ to the range $r\in[2M,\infty),$ with the null hypersurface
\begin{align*}
\mcalh^+:=\big\{(t^*, 2M, \theta, \varphi)|(t^*, \theta, \varphi)\in \mathbb{R}\times S^2\big\}
\end{align*}
describing the so-called future event horizon, one thus has the Schwarzschild-star coordinate system on the Schwarzschild exterior spacetime viewed now as a submanifold with boundary $\Mgs$ of $\big(\mcals, g_M\big)$. Since moreover $t^*$ is now a globally regular function on $\mcalm$ whose gradient is everywhere time-like, it follows that the hypersurfaces of constant $t^*$, which we denote by $\Sigma_{t^*}$, describe a foliation of $\big(\mcalm, g_M\big)$ by Cauchy hypersurfaces. In addition, observe that the causal vector field 
\begin{align*}
T=\pt
\end{align*}
is manifestly Killing. The Schwarzschild exterior spacetime is thus \emph{static}. Moreover, this same vector field determines a global time orientation on $\Mgs$. 

A Penrose diagram of the spacetime $\Mgs$ is given in Figure 1.\newline\newline

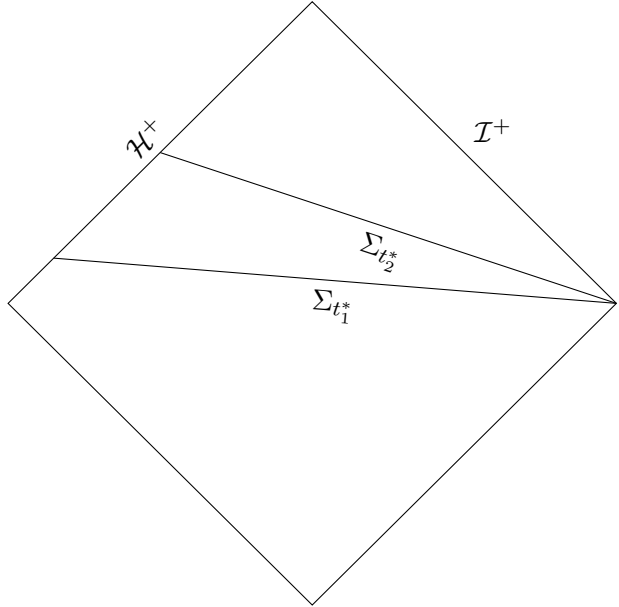
\begin{figure}[!h]
	\centering
	\begin{tikzpicture}
	\node (I)    at ( 4,0)   {};

	\path 
	(I) +(90:4)  coordinate (Itop)
	+(-90:4) coordinate (Ibot)
	+(180:4) coordinate (Ileft)
	+(0:4)   coordinate (Iright)
	;
	\draw  (Ileft) -- (Itop)node[midway, above, sloped] {$\hplus$} -- (Iright) node[midway, above right]    {$\cal{I}^+$} -- (Ibot)-- (Ileft)-- cycle;
	
	\draw 
	(0:8) -- node[midway, below, sloped]{$\Sigma _{t_1^*}$}  ($(Itop)!.85!(Ileft)$);
	
	\draw 
	(0:8) -- node[midway, below, sloped]{$\Sigma _{t_2^*}$}  ($(Itop)!.5!(Ileft)$);
	
	\end{tikzpicture}
	\caption{A Penrose diagram of $\Mgs$ depicting the Cauchy hypersurfaces $\Sigma_{t_1^*}$ and $\Sigma_{t_2^*}$.}
\end{figure}
In the main body of the paper, we will actually use the above Schwarzschild-star coordinates to \emph{define} the Schwarzschild exterior spacetime, modulo the standard degeneration of the coordinates on $S^2$, without reference to the ambient spacetime $\big(\mcals, g_M\big)$. See section \ref{ThedifferentialstructureandmetricoftheSchwarzschildexteriorspacetime} for details.

Moreover, the limit along future-directed outgoing null cones as constructed in the spacetime $\Mgs$ will informally be referred to as \emph{future null infinity}, depicted in Figure 1 as $\iplus$. See section \ref{AfoliationofMthatpenetratesbothhplusandiplus} for details.

\subsubsection{The equations of linearised gravity}\label{OVTheequationsoflinearisedgravity}

The equations of interest in this paper are those that result from linearising the Einstein vacuum equations, as expressed in a generalised wave gauge defined with respect to a fixed Schwarzschild exterior solution, about the very same Schwarzschild exterior solution. To that end, we now describe the linearisation process and present, following \cite{D--H--R}, the so-called `equations of linearised gravity' which result from making a particular choice of the map $\boldsymbol{f}$, at least to linear order. 

This section of the overview corresponds to section \ref{TheequationsoflinearisedgravityaroundSchwarzschild} of the main body of the paper.\newline

One first identifies the abstract Lorentzian manifold $\big(\bmcalm, \bog\big)$ of section \ref{OVTheEinsteinequationsinageneralisedwavegauge} with that of a \emph{fixed} member of the Schwarzschild exterior family $\Mgs$. The schematic description of the Einstein vaccum equations as expressed in \emph{this} generalised wave gauge, the system \eqref{m} and \eqref{b}, then translates to to the following:
\begin{align}
\Big(\boldsymbol{g^{-1}}\cdot\n_M^2\Big)\boldsymbol{g}+\bC_{\bg, g_M}\cdot\bC_{\bg, g_M}+\text{Riem}_M\cdot\boldsymbol{g}&=\mcall_{\boldsymbol{f}(\bg, g_M)}\boldsymbol{g},\label{m2}\\
\boldsymbol{g^{-1}}\cdot\bC_{\bg, g_M}&=\boldsymbol{f}(\bg, g_M)\label{b2}.
\end{align}
In particular, $\nabla_M$ is now the Levi-Civita connection of $g_M$ with $\boldsymbol{g}$ assumed to be a Lorentzian metric on $\mcalm$ and $\boldsymbol{f}:T^2(\mcalm)\times T^2(\mcalm)\rightarrow T\mcalm$ a smooth map.

In order to formally linearise, we consider a smooth 1-parameter family of solutions $\boldsymbol{g}(\epsilon)$ to the Einstein vacuum equations, defined on $\mcalm$, with $\boldsymbol{g}(0)=g_M$. We will also demand that each solution $\boldsymbol{g}(\epsilon)$ is in a generalised $\boldsymbol{f}$-wave gauge with respect to $g_M$ where $\boldsymbol{f}$ is a smooth map such that $\boldsymbol{f}(g_M,g_M)=0$.

Observing thus that $g_M$ is indeed a solution to the Einstein vacuum equations as expressed in a generalised $\boldsymbol{f}$-wave gauge with respect to $g_M$, to linearise we consider a formal power series expansion of $\boldsymbol{g}(\epsilon)$ in terms of $\epsilon$:
\begin{align*}
\boldsymbol{g}(\epsilon)&=g_M+\epsilon\cdot\glin+o\big(\epsilon^2\big).
\end{align*}
Here, $\glin$ is a symmetric 2-covariant tensor field on $\mcalm$ denoting the linearised metric. Thus, in keeping with \cite{D--H--R}, linearised quantities are denoted by a superscript $(1)$. We also write
\begin{align*}
\boldsymbol{f}(\boldsymbol{g}(\epsilon), g_M)=\epsilon\cdot Df\big|_{g_M}(\glin)+o\big(\epsilon^2\big)
\end{align*}
where $Df\big|_{g_M}:T^2(\mcalm)\rightarrow T\mcalm$ is a smooth linear map denoting the linearisation of the map $\boldsymbol{f}(\cdot, g_M)$ at $g_M$.

One then arrives at the linearised equations by inserting this formal power series expansion into the system of equations defined by Einstein vacuum equations, as expressed in the generalised $\boldsymbol{f}$-wave gauge with respect to $g_M$, and discarding those terms that appear to higher than linear order in $\epsilon$. Proceeding in this manner leads to the following\footnote{Note we have dropped the subscript $M$ notation.} system of equations:
\begin{align}
\Box\glin-2\text{Riem}\cdot\glin&=\mcall_{\flin}g_M,\label{OVeqnlingrav1}\\
\text{div}\glin-\frac{1}{2}\ud\text{tr}\glin&=\flin\label{OVeqnlingrav2}.
\end{align}
Here, $\text{div}$ and $\Box$ are the spacetime divergence and Laplacian of $g_M$ respectively, with $\ud$ the exterior derivative on $\mcalm$. Moreover, we have defined the vector field $\flin$ according to
\begin{align*}
\flin:=Df\big|_{g_M}(\glin).
\end{align*}

The linearisation of the Einstein vacuum equations, as expressed in a generalised wave gauge, around Schwarzschild thus comprise of the \emph{tensorial system of linear wave equations \eqref{OVeqnlingrav1} coupled with the divergence relation \eqref{OVeqnlingrav2}} on the Schwarzschild exterior spacetime. We remark that these are nothing but the linearised Einstein equations in a \emph{generalised} Lorentz gauge. See the book of Wald \cite{Wald}.

The structure of the system of equations \eqref{OVeqnlingrav1}-\eqref{OVeqnlingrav2} is made more transparent by employing the so called $2+2$ formalism. In utilising this $2+2$ formalism, a formalism which was first adopted in \cite{G--S}, one exploits the fact that the topology of $\mcalm$ has the product structure $\mcalq\times S^2$ where $\mcalq$ is 2-dimensional manifold with boundary. Informally, this allows a decomposition of objects on $\mcalm$ into their parts `tangent' to $\mcalq$ and $S^2$. In particular, one can decompose tensor fields on $\mcalm$ into what we term as $\mcalq$-tensor fields and $S$-tensor fields respectively (see section \ref{The2+2decompositionoftensorfieldsonM} in the bulk of the paper for a precise definition). 

Indeed, applying this decomposition to the linearised metric $\glin$ yields
\begin{align*}
\glin\rightarrow\qglin, \mling, \sglin
\end{align*}
where $\qglin$ is a symmetric 2-covariant $\mcalq$-tensor, $\sglin$ is a symmetric 2-covariant $S$-tensor and $\mling$ is a $\mcalq\otimes S$ 1-form. Conversely, the decomposition applied to the 1-form $\flin$ on $\mcalm$ returns
\begin{align*}
\flin\rightarrow\tflin, \sflin
\end{align*}
where $\tflin$ is $\mcalq$ 1-form and $\sflin$ is an $S$ 1-form.

Decomposed quantities in hand, the system of equations \eqref{OVeqnlingrav1}-\eqref{OVeqnlingrav2} for the pair
\begin{align*}
\Big(\glin, \flin\Big)
\end{align*}
will now reduce to a system of coupled wave equations for the tuple
\begin{align*}
\bigg(\qglin, \mling, \sglin, \tflin, \sflin\bigg).
\end{align*}
Before we present a sample of the reduced equations however, it will behoove us to perform a further decomposition upon $\qglin$ and $\sglin$ into a symmetric, traceless 2-covariant $\mcalq$-tensor field and a symmetric, traceless 2-covariant $S$-tensor field respectively:
\begin{align*}
\qglin=\gabhatlin+\frac{1}{2}\qg_M\cdot\trgablin, \qquad\sglin=\sghatlin+\frac{1}{2}\sg_M\cdot\trsglin
\end{align*}
with $\qg_M$ and $\sg_M$ the Lorentzian $\mcalq$-metric and the Riemannian $S$-metric resulting from applying the $2+2$ decomposition to $g_M$ accordingly. In particular, the scalar functions $\trgablin$ and $\trsglin$ on $\mcalm$ correspond to the respective traces of $\qglin$ and $\sglin$ with $\qg_M$ and $\sg_M$.

One thus reduces the system of equations \eqref{OVeqnlingrav1}-\eqref{OVeqnlingrav2} to the system of coupled tensorial wave equations satisfied by the collection
\begin{align*}
\bigg(\gabhatlin, \trgablin, \mling, \sghatlin, \trsglin, \tflin, \sflin\bigg).
\end{align*}
A sample of this \emph{system of gravitational perturbations} is presented below.\newline

\noindent For the $\mcalq\otimes S$ 1-form $\mling$,
\begin{align}\label{ovpoo}
\qbox\mling+\slap\mling-\frac{1}{r^2}\ommu\mling=\frac{2}{r}\ud r\,{\otimes}\bigg(\frac{2}{r}\mling_P+\sdiv\sghatlin+\frac{1}{2}\sn\trsglin-\sflin\bigg)-\frac{2}{r}\sn{\otimes}\gabhatlin_P-\frac{1}{r}\ud r\,\sn{\otimes}\trgablin+\qn{\otimes}\sflin+\sn{\otimes}\tflin.
\end{align}
For the symmetric, traceless 2-covariant $\sm$-tensor $\sghatlin$,
\begin{align}\label{ovpoo1}
\qbox\sghatlin+\slap\sghatlin-\frac{2}{r}\qn_{P}\sghatlin-\frac{4}{r}\frac{\mu}{r}\sghatlin&=-\frac{2}{r}\sn\otimeshat\mling_P-2\sdst\sflin.
\end{align}
Finally, the $S$-component of the wave gauge condition 
\begin{align}\label{ovpoo2}
-\qd\mling-\frac{1}{2}\sn\trgablin+\sdiv\sghatlin+\frac{2}{r}\mling_P&=\sflin.\newline
\end{align}

Here, all geometric objects with a tilde above them relate to the $\mcalq$-metric $\qg_M$ whereas all geometric objects with a slashed through them relate to the $S$-metric $\sg_M$. In particular, $-\qd$ is the divergence operator associated to $\qg_M$, $\sn\otimeshat \xi$ is the trace-free part of the Lie derivative of $\sg$ with respect to the $\sm$-vector field $\xi$ and we emphasize that $\qbox$ and $\slap$ are the wave operator and Laplacian associated to $\qg_M$ and $\sg_M$ respectively. Moreover $\sdst$ is an appropriate $L^2$ adjoint of $\sdiv$. Lastly, $\mu=\frac{2M}{r}$ is the mass aspect function, $P$ is the vector field associated to the 1-form $\ud r$ under the musical isomorphism on $\Mgs$ and we use the notation
\begin{align*}
\omega_X:=\omega(X)
\end{align*}
for $\omega$ and $X$ a 1-form and vector field on $\mcalm$ respectively.

The greater transparency afforded by this $2+2$ formalism now motivates the choice of the linear map $Df\big|_{g_M}$ we are to make in this paper. Indeed, we define the map\footnote{Here we recall the notation $\mscrT^k(\mcalm)$ for the space of smooth $k$-covariant tensor fields on $\mcalm$.} $Df\big|_{g_M}:\mscrT^2(\mcalm)\rightarrow\mscrT(\mcalm)$ according to
\begin{align}\label{chin}
Df\big|_{g_M}(X):=\frac{2}{r}\tilde{X}_P-\frac{1}{r}\big(\hat{\tilde{X}}_{l=0}\big)_P+\frac{2}{r}\stkout{X}_P-\frac{1}{r}\ud r\,\str\slashed{X}.
\end{align}
Here, $\tilde{X}$, $\stkout{X}$ and $\slashed{X}$ are the projections of the smooth 2-covariant tensor $X$ onto a smooth 2-covariant $\qm$-tensor, a smooth $\qmsm$ 1-form and a smooth 2-covariant $\sm$-tensor respectively. Moreover, $\hat{\tilde{X}}$ denotes the traceless part of $\tilde{X}$ with respect to $\qg_M$ and $\tilde{X}_{l=0}$ denotes the spherically symmetric part of $\tilde{X}$ (see section \ref{Thel=0,1spheriacalharmonicsandqmtensors} in the bulk of the paper).

The utility of such a choice is made manifest by the observation that, modulo the `$l=0$ modes', in this gauge all zero'th order terms in the generalised wave gauge condition \eqref{OVeqnlingrav2} vanish. This can be verified, for instance, from equation \eqref{ovpoo2}. Conversely, this modification of the `$l=0$' modes is related to the special solutions we are to discuss in the next section of the overview.

Of course, one should verify that such a map can indeed be realised as the linearisation of the map $\boldsymbol{f}$ described in our earlier formal linearisation procedure. However, this is readily seen to be true -- for example\footnote{One can of course construct many other examples!}, the linear map $\boldsymbol{{f}}:\mathscr{T}^2(\mcalm)\rightarrow \mscrT\mcalm$ defined according to
	\begin{align*}
	\big(\boldsymbol{f}(X)\big)^\flat&=\frac{2}{r}\tilde{X}_P-\frac{1}{r}\big(\hat{\tilde{X}}_{l=0}\big)_P+\frac{2}{r}\stkout{X}_P-\frac{1}{r}\ud r\,\str\slashed{X}
	\end{align*}
	satisfies $\boldsymbol{{f}}(g_M)=0$ and moreover trivially linearises to the map $Df\big|_{g_M}$ of $\eqref{chin}$. Here, $\flat$ is the index-lowering musical isomorphism on $\Mgs$.
	
Motivated by \cite{D--H--R}, we use the collective notation
\begin{align*}
\Solution
\end{align*}
to denote solutions to the system of equations \eqref{OVeqnlingrav1}-\eqref{OVeqnlingrav2} that result from defining the vector field $\flin$ according to the map $Df\big|_{g_M}$ as given in \eqref{chin}. We shall henceforth refer to this system as \textbf{the equations of linearised gravity}. The remainder of the paper is then concerned with solutions to this system of equations. A sample of said system decomposed under the $2+2$ formalism is presented below.\newline

\noindent For the $\mcalq\otimes S$ 1-form $\mling$,
\begin{align}
\qbox\mling+\slap\mling-\frac{2}{r}\big(\qn{\otimes}\mling\big)_P-\frac{1}{r^2}\mling+\frac{2}{r^2}\ud r{\otimes}\mling_P=\frac{2}{r}\ud r\,{\otimes}\sdiv\sghatlin\label{OVeqnlingrav3}.
\end{align}
For the symmetric, traceless 2-covariant $\sm$-tensor $\sghatlin$,
\begin{align}
\qbox\sghatlin+\slap\sghatlin-\frac{2}{r}\qn_{P}\sghatlin-\frac{4}{r}\frac{\mu}{r}\sghatlin&=0\label{OVeqnlingrav4}.
\end{align}
Finally, the $S$-component of the wave gauge condition 
\begin{align}
-\qd\mling-\frac{1}{2}\sn\trgablin+\sdiv\sghatlin&=0\label{OVeqnlingrav7}.\newline
\end{align}
\\
We emphasize to the reader that a considerable advantage of employing a generalised wave gauge with respect to $g_M$, i.e. with respect to the metric one is to linearise about, is the absence of any explict first order derivatives of $\glin$ in the linearisation of the Einstein vacum equations as expressed in this gauge, the system of equations given by \eqref{OVeqnlingrav1}, and the absence of any explicit zero'th terms of $\glin$ in the linearisation of the generalised wave gauge condition itself, the system of equations given by \eqref{OVeqnlingrav2}. For instance, returning to the schematic description of the Einstein vacuum equations \eqref{m2}, one sees that under linearisation\footnote{Since $\bC_{g_M, g_M}=0$.} the terms quadratic in the connection tensor vanish. This would be patently false if one instead linearises the Einstein vacuum equations, as expressed in \emph{wave coordinates}, about the solution\footnote{Since, in a system of wave coordinates on $\Mgs$, $\Gamma_M\neq0$.} $g_M$. The resulting structure this vanishing implies for the equations of linearised gravity thus constitutes the amenable structure we mentioned in the introduction.
\newline\newline

As to the significance of the terms quadratic in the Christoffel symbols vanishing, see the non-schematic description of the Einstein vacuum equations as expressed in a generalised wave gauge presented in section \ref{Thelinearisationprocedure}.

Finally, see section \ref{Tensoranalysis} in the bulk of the paper for a precise definition of the various operations and differential operations on $\Mgs$ that follow naturally from this $2+2$ formalism and which shall appear frequently in the remainder of the overview.

\subsubsection{Special solutions to the equations of linearised gravity: Linearised Kerr and pure gauge}\label{OVSpecialsolutionsLinearisedKerrandpuregauge}

We end this first part of the overview by discussing two special classes of solutions to the equations of linearised gravity, namely the linearised Kerr and pure gauge solutions.

This section of the overview corresponds to section \ref{Specialsolutionstotheequationsoflinearisedgravity} of the main body of the paper.
\newline

The first such class of solutions arises from suitably expressing members of the 2-parameter Kerr exterior family in the generalised wave gauge identified in the previous section and then linearising about the fixed Schwarzschild exterior solution $g_M$. This yields an explicit \emph{4-parameter\footnote{One indeed has a 4-parameter family of solutions since the Schwarzschild exterior background around which ones linearises has no preferred axis of rotation.} family} of solutions to the equations of linearised gravity corresponding to the fact that $g_M$ actually sits within a family of solutions to the Einstein vacuum equations. 

For instance, linearising a Kerr exterior solution with mass $M$ in the rotation parameter $\boldsymbol{a}$ yields the following 3-parameter family of solutions to the equations of linearised gravity:
\begin{align*}
	\mling&=-\frac{1}{1-\mu}\Big(\mu\qhd P-\ud r\Big)\shd\sn\mathfrak{a},\\
\gabhatlin=\trgablin=\sghatlin=\trsglin&=0.
\end{align*}
Here, $\mathfrak{a}$ is a smooth function on $S^2$ lying in the span of the classical $l=1$ spherical harmonics on $S^2$. Further, $\qhd$ and $\shd$ are the Hodge duals with respect to the $\qm$-metric $\qg_M$ and $\sm$-metric $\sg_M$ respectively and $\mathring{g}$ is the unit metric on the round sphere. Of course, that the above indeed represents a 3-parameter family of solutions to the equations of linearised gravity can be verified explicitly.

Following \cite{D--H--R}, we shall denote this special class of \emph{linearised Kerr solutions} by $\Ke_{\mathsmaller{\mfM,\mathfrak{a}}}$. Consequently, we note two key properties of this family -- stationarity\footnote{That is, $\mcall_T\glin=\mcall_T\flin=0$ where $T$ is the causal Killing field introduced in section \ref{OVTheexteriorSchwarzschildspacetime}.} and being supported entirely on the $l=0,1$ spherical harmonics (see section \ref{TheprojectionoftensorfieldsonMontoandawayfromthel=0,1sphericalharmonics} for a precise definition). We note that this latter fact is the reason for the `$l=0$' term in the definition of the map $Df\big|_{g_M}$ in \eqref{chin} in that it allows one to realise the linearisation of the Schwarzschild exterior family in the mass parameter, when expressed in the Schwarzschild-star coordinates of section \ref{OVTheexteriorSchwarzschildspacetime}, as a 1-parameter family of solutions to the equations of linearised gravity.

The second such class of solutions arises from the fact that, on an abstract Lorentzian manifold $\bmcalm$, there exists a residual class of diffeomorphisms on $\bmcalm$ under which the generalised wave gauge is preserved.

Indeed, suppose that the smooth, globally hyperbolic Lorentzian manifolds $\big(\bold{\mcalm}, \bold{g}\big)$ and $\big(\bold{\mcalm}, \bold{\overline{g}}\big)$ are solutions to the Einstein vacuum equations for which $\boldsymbol{g}$ is in a generalised $\boldsymbol{f}$-wave gauge with respect to $\boldsymbol{\overline{g}}$. Then there exists a non-trivial class of diffeomorphisms $\boldsymbol{\phi}$ on $\bmcalm$ for which $\boldsymbol{\phi^*g}$ is in a generalised $\boldsymbol{f}$-wave gauge with respect to $\bog$. In particular, $\boldsymbol{\phi^*g}$ and $\bg$ thus reside within the same equivalence class of solutions to the Einstein vacuum equations.

This phenomena is manifested in the linear theory by the existence of a 1-parameter family of diffeomorphisms on $\mcalm$ for which each of the 1-parameter family of Lorentzian metrics given as the pullback of $g_M$ under this 1-parameter family of diffeomorphisms are in a generalised $\boldsymbol{f}$-wave gauge with respect to $g_M$ to first order. Here, $\boldsymbol{f}:\mscrT^2(\mcalm)\rightarrow\mscrT\mcalm$ is any smooth map which linearises to the map $Df\big|_{g_M}$ defined as in \eqref{chin}. This yields, upon linearisation, a family of solutions to the equations of linearised gravity corresponding to \emph{residual gauge freedom} in the non-linear theory. Following \cite{D--H--R}, we shall denote this special class of \emph{pure gauge solutions} by $\Ga$.

Indeed, let $v$ be a smooth vector field on $\mcalm$ satisfying
\begin{align*}
	\Box v=Df\big|_{g_M}(\mcall_vg_M)
\end{align*}
where $Df\big|_{g_M}$ is defined as in \eqref{chin}. Then the following
\begin{align}\label{OVeqnpuregauge}
	\glin=\mcall_vg_M
\end{align}
is a smooth solution to the equations of linearised gravity arising from the linearisation of the 1-parameter family of Lorentzian metrics given as the pullback of $g_M$ under the 1-parameter family of diffeomorphisms generated by $v$. That the above solve the equations of linearised gravity can indeed by verified from equations \eqref{OVeqnlingrav1} and \eqref{OVeqnlingrav2}.\newline

A consequence of the existence of these solutions is the expectation that, at best, a general solution to the equations of linearised gravity decays to a member of the linearised Kerr family, after the addition of some pure gauge solution.

\subsection{The Regge--Wheeler and Zerilli equations and the gauge-invariant hierarchy}\label{OVAneffectivescalarisationoftheequationsoflinearisedgravity}

In this second part of the overview we discuss how one extracts, via a hierarchy of \emph{gauge-invariant} quantities, from the equations of linearised gravity the two fully decoupled scalar wave equations described by the Regge--Wheeler and Zerilli equations respectively.

This section of the overview corresponds to section \ref{TheRegge-WheelerandZerilliequationsandthegaugeinvariantheirarchy} of the main body of the paper.\newline

Owing to the potential complications arising from the existence of the linearised Kerr and pure gauge solutions of section \ref{OVSpecialsolutionsLinearisedKerrandpuregauge}, it is natural to consider those quantities which vanish for all such solutions.

To isolate these \emph{gauge-invariant} quantities, we first recall from section \ref{OVSpecialsolutionsLinearisedKerrandpuregauge} that members of the linearised Kerr family are supported only on the $l=0,1$ spherical harmonics. A more useful characerisation is the statement that the linearised Kerr solutions in fact lie in the kernels of the family of `$\sm$-operators' on the 2-spheres
\begin{align*}
\spi=\big\{\sps, \spv, \spt\big\}
\end{align*}
 which respectively map $\qm$-tensors, $\qmsm$ 1-forms and $\sm$-tensors into the space $\LM$ of $\qm$-tensors which are supported \emph{outside} of the $l=0,1$ spherical harmonics.  For instance, the operator $\spv$ acts on smooth $\qmsm$ 1-forms $\momega$ by mapping them into the smooth pair of $\qm$ 1-forms in $\LM$ defined as
\begin{align*}
\spv\momega:=\Big(r^4\sdiv\sdiv\sn\otimeshat\momega, r^4\scurl\sdiv\sn\otimeshat\momega\Big).
\end{align*}
See section \ref{ThefamilyofoperatorsPandT} in the bulk of the paper for the definition of the operators $\sps$ and $\spt$ respectively.

The operators $\spi$ thus serve to \emph{project out} the $l=0,1$ modes of tensor fields on $\mcalm$ whilst concurrently facilitating a comparison between tensor fields of different `$\sm$-type'. In particular, the quantities we seek are thus determined by applying the $\spi$ operators to the linearised metric and isolating those combinations which vanish for all pure gauge solutions $\Ga$. An example of such a quantity is the $\qm$ 1-form $\etalin$ defined by
\begin{align}\label{OVeqndefneta}
\etalin:=r^4\scurl\sdiv\sn\otimeshat\mling-\frac{1}{2}r^2\qexd\bigg(r^2\scurl\sdiv\sghatlin
\bigg).
\end{align}

The full collection of gauge-invariant quantities, originally discovered\footnote{In fact, it took the later work of Moncrief in \cite{Moncrief} to clarify the gauge-invariance of these quantities. Moreover, the $\slashed{\Pi}$ operators were completely absent in \cite{RW}, their roles being fulfilled by the tensor, vector and scalar spherical harmonics employed in the mode decomposition.} by Regge--Wheeler in \cite{RW} in the context of a full mode decomposition of the linearised\footnote{In particular, the generalised wave gauge was not imposed.} Einstein equations, comprise of:
\begin{itemize}
	\item the symmetric, traceless 2-covariant $\qm$-tensor field $\tauhatlin$
	\item the scalar function $\trtaulin$ on $\mcalm$
	\item the $\qm$ 1-form $\etalin$
	\item the scalar function $\sigmalin$ on $\mcalm$.
\end{itemize}
Remarkably, as was first discovered by Regge--Wheeler \cite{RW} and Zerilli \cite{Z}, the system of gravitational perturbations force a decoupling of the gauge-invariant quantities into two \emph{scalar waves}.

For example, equations \eqref{ovpoo} and \eqref{ovpoo1} of section \ref{OVTheequationsoflinearisedgravity} imply that the gauge invariant $\qm$ 1-form $\etalin$ satisfies the decoupled wave equation
\begin{align*}
\qbox\etalin+\slap\etalin-\frac{2}{r}\Big(\qn\etalin\Big)_P+\frac{2}{r^2}\ud r\text{ }\etalin_P&=0.
\end{align*}
Moreover, imposing the wave gauge condition \eqref{ovpoo2}  implies that $\etalin$ is divergence free:
\begin{align*}
\qd\etalin&=0.
\end{align*}
This serves as an integrability condition for $\etalin$ and an application of the Poincar\'e lemma yields
\begin{align}\label{OVeqndefnphi}
\etalin=-\qhd\qexd\Big(r\Philin\Big)
\end{align}
where the smooth function $\Philin\in\LM$ satisfies the celebrated \emph{Regge--Wheeler} equation
\begin{align}\label{OVeqnRW}
\qbox\Philin+\slap\Philin=-\frac{3}{r}\frac{\mu}{r}\Philin.
\end{align}
 A similar but more complicated hierarchical procedure applied now to the gauge-invariant quantities $\tauhatlin, \trtaulin$ and $\sigmalin$ returns
 \begin{align*}
 \trtaulin=0
 \end{align*}
 and
\begin{align}
\tauhatlin&=\qn\otimeshat\qexd\Big(r\Psilin\Big)+6\mu\ud r\otimeshat\szetap{1}\qexd\Psilin\label{wanker1},\\
\sigmalin&=-2r\slap\Psilin+4\qn_{P}\Psilin+12\mu r^{-1}\ommu\szetap{1}\Psilin,\label{wanker2}
\end{align}
where the smooth function $\Psilin\in\LM$ satisfies the celebrated \emph{Zerilli} equation
\begin{align}\label{OVeqnZ}
\qbox\Psilin+\slap\Psilin=-\frac{3}{r}\frac{\mu}{r}\Psilin+\frac{6}{r}\frac{\mu}{r}(2-3\mu)\szetap{1}\Psilin+18\frac{\mu}{r}\frac{\mu}{r}(1-\mu)\szetap{2}\Psilin.
\end{align}
Here, $\szetap{p}$ is the inverse of the elliptic operator $r^2\slap+2-\frac{6M}{r}$
applied $p$-times. This is indeed invertible over the space of smooth functions in $\LM$.

We note that, in their original pioneering works \cite{RW} and \cite{Z}, Regge--Wheeler and Zerilli derived the Regge--Wheeler and Zerilli equations in terms of both the previously mentioned full mode decomposition of the linearised Einstein equations combined with a  fixing of Regge--Wheeler coordinates on $\Mgs$. The covariant, non-mode decomposed version of these equations presented above is ultimately\footnote{Earlier works of \cite{G--S} and \cite{S--T} established a covariant derivation of the Regge--Wheeler and Zerilli equations.
On the other hand, in \cite{Jez} a non-covariant, non-modal derivation was given, albeit via a Hamiltonian formulation of the problem.} due to Chaverra, Ortiz and Sarbach in \cite{C--O--S}.\newline

We moreover note the interesting result of Corrolary \ref{corrPsiPhivanishimpliespuregauge} in the bulk of the paper which states that a sufficiently regular solution to the equations of linearised gravity for which both $\Philin$ and $\Psilin$ vanish is in fact the sum of a linearised Kerr and pure gauge solution.
We further remark that, as discovered by Chandrasekhar, there exists a transformation theory mapping solutions of the Zerilli equation to solutions of the Regge--Wheeler equation, although we will not make use of this in this paper. See \cite{Chandbook} for details. 

Lastly, we remind the reader that the Regge--Wheeler equation appears and plays a major role in the work of Dafermos, Holzegel and Rodnianski in \cite{D--H--R}. It is remarkable that the same equation appears in these two different contexts. Note that the linearised Robinson--Trautman \cite{R--T} solutions that appear in their work as solutions with \emph{vanishing} `Regge--Wheeler quantities' are related to the aforementioned transformation mapping solutions of the Zerilli equation to the Regge--Wheeler equation, which has non-trivial kernel consisting of the so-called algebraically special modes. \newline\newline

The definition of the full collection of gauge-invariant quantities can be found in section \ref{Gauge-invariantquantities} with the corresponding gauge-invariant hierarchy derived in section \ref{Theconnectiontothesystemofgravitationalperturbations}.

\subsection{The Cauchy problem for the equations of linearised gravity}\label{OVTheCauchyproblemfortheequationsoflinearisedgravity}

In this third part of the overview we discuss a well-posedness result for the equations of linearised gravity as a Cauchy problem.

Given that, as our discussion in section \ref{OVSpecialsolutionsLinearisedKerrandpuregauge} of the overview indicates, the issues surrounding the existence of both the Kerr family and residual gauge freedom in the nonlinear theory can be understood completely in the context of the linear theory, this foundational statement, although complicated by the presence of constraints, allows one to develop the theory of linearised gravity \emph{without further reference to their origin}. 

This part of the overview corresponds to section \ref{Initialdataandwell-posednessoflinearisedgravity} of the main body of the paper.\newline

It is well known that the Einstein equations \eqref{introEE} must satisfy certain \emph{constraints} which arise as a consequence of the Gauss--Codazzi equations restricting the embedding of a 3-manifold as a hypersurface in a 4-manifold. These constraints are of course inherited by the linearisation, although they can be understood purely within the context of the linear theory.  Moreover, in regards to the \emph{Cauchy} problem for the equations of a linearised gravity, a further constraint is imposed by the generalised wave gauge condition \eqref{OVeqnlingrav2}. Therefore, any well-posedness result for the equations of linearised gravity must neccesarily include a procedure by which one generates \emph{admissible} initial data. A definition of such admissible initial data, along with the constraints it must satisfy, is to be found in section \ref{Admissibleinitialdatasetsforthesystemofgravitationalperturbations} of the bulk of the paper.

We confront this issue by introducing a notion of \emph{seed data} for the equations of linearised gravity. Indeed, it is from this \emph{freely prescribed} seed data that we are able to generate, uniquely,  an admissible initial data set for the system of gravitational perturbations. 

Fixing an initial Cauchy hypersurface $\Sigma$, corresponding to a level set of the time function $t^*$ introduced in section \ref{OVTheexteriorSchwarzschildspacetime}, this seed data consists of a collection of freely prescribed quantities on $\Sigma$. We denote this collection of seed by $\Sos$, with corresponding admissible initial data denoted by $\Soa$. An example of the seed quantities are
\begin{itemize}
	\item the two smooth functions $\oPhilin$ and $\uPhilin$ on $\Sigma$ that are supported outside of $l=0,1$
\end{itemize}
See section \ref{Seeddataforlinearisedgravity} in the bulk of the paper for a full description of the seed data $\Sos$.

The procedure by which we extend this seed data to a full admissible initial data set exploits the existence of three explicit classes of solutions to the equations of linearised gravity. Two have been discussed already, namely the linearised Kerr and pure gauge solutions. The third class, which are parametrised by two scalar quantities satisfying the Regge--Wheeler and Zerilli equations respectively, arise from inverting the family of operators $\mcalt$ in the expressions found in i) and ii) of the conditions describing the Regge--Wheeler gauge in section \ref{OVGaugenormalisedsolutionstotheequationsoflinearisedgravity} of the overview. Consequently, by (appropriately) projecting these three classes of solutions onto $\Sigma$ one generates three explicit classes of solutions to the linearised constraint equations and thus prescribing seed data in such a way as to generate said solutions determines, by linearity, an admissible initial data set given as their sum.

For instance, the seed quantities $\uPhilin$ and $\oPhilin$ introduced above are to determine Cauchy data for the gauge-invariant quantity $\Philin$ of section \ref{OVAneffectivescalarisationoftheequationsoflinearisedgravity} which satisfies the Regge--Wheeler equation. This part of the seed data thus explicitly generates a solution to the linearised constraint equations corresponding to the third class of explicit solutions. Moreover since, once certain `gauge-considerations' are taken into account, the freely prescribed seed data contains four\footnote{Corresponding exactly to Cauchy data for the gauge-invariant quantities $\Philin$ and $\Psilin$.} functional degrees of freedom, the full prescription of seed data generates a class of solutions to the linearised constraint equations which agrees with the full functional degrees of freedom one associates to the equations of linearised gravity by a crude function counting arguement (see the book of Wald \cite{Wald}). 

One advantage of this method to solving the linearised constraint equations is that, since it involves at most applying derivatives and inverting elliptic operators on 2-spheres, we are thus able to generate solutions to the constraint equations, linearised about the Schwarzschild exterior, that agree with admissible initial data for a linearised Kerr solution\footnote{The exact linearised Kerr solution one matches with is to be determined explicitly from the seed data. See section \ref{OVGaugenormalisedsolutionstotheequationsoflinearisedgravity}.} outside a compact set \emph{without introducing a gluing region}. This is a non-trivial statement since the linearised constraint equations in effect correspond to an elliptic system on the non-compact manifold $\Sigma$. In regards to generating  solutions to the full \emph{nonlinear} constraints for general asymptotically flat manifolds that agree with a \emph{Kerr} solution outside a compact set, see the pioneering work of Corvino--Schoen in the celebrated \cite{C--S}. See section \ref{AppendixConstructingadmissibleinitialdatafromseeddata} of the Appendix attached to the bulk of the paper for full details of the method.

In order to evolve this admissible initial data into a full solution to the equations of linearised gravity we appeal to the argument employed by Choquet-Bruhat in her celebrated \cite{C-B}. In brief, one first constructs a unique solution to the `well-posed' equation \eqref{OVeqnlingrav1} with corresponding Cauchy data given by the collection $\Soa$. It is then a classical\footnote{For instance, see the book of Choquet-Bruhat \cite{C-Bbook}.} result that the admissibility of said data in fact implies that this solution must neccesarily satisfy the gauge condition \eqref{OVeqnlingrav2}. This leads to the well-posedness theorem of section \ref{Thewell-posednesstheorem} in the bulk of the paper, which we summarise here as:
\begin{customthm}{0}\label{OVthmwellposedness}
The smooth seed data set $\Sos$ prescribed on $\Sigma$ gives rise to a unique, smooth solution $\So$ to the equations of linearised gravity on $\mcalm\cap D^+(\Sigma)$.
\end{customthm}

The boundedness and decay statements we discuss in section \ref{OVThemaintheoremandoutlineoftheproof} for solutions to the equations of linearised gravity will require that the initial data from which they arise satisfy in addition a notion of \emph{asymptotic flatness}. Importantly, we are able to provide such a notion on the freely prescribed seed data alone, which then propagates under our procedure of determining the full initial data and thus these \emph{asymptotically flat} solutions can be understood purely within the context of the above well-posedness theorem. See Theorem \ref{Apppropasymptoticflatness} in section \ref{AppPropagationofasymptoticflatness} of the appendix attached to the main body of the paper. Moreover, the precise definition of asymptotic flatness we require can be found in section \ref{Pointwisestrongasymptoticflatness} in the bulk of the paper.\newline

Finally, we end this part of the overview with a comments towards potential future work. Indeed, it should in principle be possible to show that under sufficient regularity this method we develop, along with the free data we prescribe, actually parametrises the full space of solutions to the linearised constraint equations.\newline\newline

Further details pertaining to the initial hypersurface $\Sigma$ and the corresponding notion of $\smi$-tensors can be found in section \ref{TheCauchyhypersurfaceSigma}.

\subsection{Gauge-normalised solutions to the equations of linearised gravity and identification of the Kerr parameters}\label{OVGaugenormalisedsolutionstotheequationsoflinearisedgravity}

In this fourth part of the overview we discuss particular solutions to the equations of linearised gravity that arise from exploiting the residual gauge freedom afforded from the existence of the pure gauge solutions of section \ref{OVSpecialsolutionsLinearisedKerrandpuregauge}.

This part of the overview corresponds to section \ref{GaugenormalisedsolutionsandidentificationoftheKerrparameters} of the main body of the paper.\newline

The first gauge under consideration in fact corresponds to a 4-parameter family of gauge choices parametrised by a constant $\mfM\in\mathbb{R}$ and a smooth function $\mathfrak{a}$ on $S^2$ lying in the span of the $l=1$ spherical harmonics. We shall term a member of this family as \emph{a $\Ke_{\mathsmaller{\mfM, \mathfrak{a}}}$-adapted Regge--Wheeler gauge}\footnote{The gauge is so-named because it is an adaption of the celebrated Regge--Wheeler gauge, originally used by Regge and Wheeler in their study \cite{RW} of the linearised Einstein equations on the Schwarzschild exterior, to the equations of linearised gravity.} and a smooth solution $\So$ to the equations of linearised gravity is said to be in such a gauge if the following conditions on $D^+(\Sigma)$:
\begin{enumerate}[i)]
	\item $\sdiv\mling=0$
	\item $\sghatlin=0$
	\item $\mpart{\So}+\dpart{\So}=\Ke_{\mathsmaller{\mfM, \mathfrak{a}}}$
\end{enumerate}
Here, the subscripts $l=0,1$ denote a projection onto the $l=0,1$ spherical harmonics, to be made precise in section \ref{Theprojectionontol=0andl=1} in the bulk of the paper, and we recall that $\Ke_{\mathsmaller{\mfM, \mathfrak{a}}}$ is a member of the linearised Kerr family introduced in section \ref{OVSpecialsolutionsLinearisedKerrandpuregauge}.

Indeed, it is smooth solutions to the equations of linearised gravity satisfying the above, when supplemented with an appropriate asymptotic flatness condition, that will be subject to our quantitative boundedness and decay statement to be discussed in section \ref{OVTheorem1:boundednessofthesolutionSi} of the overview. This will follow as a consequence of the quantitative boundedness and decay bounds we are able to derive for the gauge-invariant quantities $\Philin$ and $\Psilin$ of section \ref{OVAneffectivescalarisationoftheequationsoflinearisedgravity} satisfying the Regge--Wheeler and Zerilli equations respectively. For a solution $\So$ to the equations of linearised gravity which satisfies conditions i)-iii) in fact must obey on $D^+(\Sigma)$ the relations
\begin{align}
\mcaltf\gabhatlin&=\tauhatlin,\label{a1}\\
\mcaltv\mling&=\sdso\Big(0, \etalin\Big),\label{a2}\\
\mcaltf\trsglin&=\sigmalin\label{a3}
\end{align}
and
\begin{align*}
\trgablin=0.
\end{align*}
Here, $\tauhatlin, \etalin$ and $\sigmalin$ are determined from the gauge-invariant quantities $\Philin$ and $\Psilin$ as in section \ref{OVAneffectivescalarisationoftheequationsoflinearisedgravity}.
Moreover, the pair $\mcaltf$ and $\mcaltv$
are a pair of \emph{elliptic} operators on 2-spheres that are naturally related to the $\spi$ family of operators introduced in section \ref{OVAneffectivescalarisationoftheequationsoflinearisedgravity}. For instance,
\begin{align*}
\mcaltv:=\sdso\spv
\end{align*}
where the operator $\sdso$ maps pairs of $\qm$ 1-forms $\tilde{q}_1, \tilde{q}_2$ into the $\qmsm$ 1-form defined as
\begin{align*}
\sdso\big(\tilde{q}_1, \tilde{q}_2\big):=-\sn\tilde{q
}_1-\shd\sn\tilde{q}_2.
\end{align*}
See section \ref{ThefamilyofoperatorsPandT} for the definition of the operator $\mcaltf$ respectively.

In particular, applying elliptic estimates on the operators $\mcaltf$ and $\mcaltv$ thus readily transfers bounds on the quantities $\Philin$ and $\Psilin$ to the solution $\So$. Moreover, as the operators $\mcaltf$ and $\mcaltv$ have kernels spanned by the $l=0,1$ spherical harmonics (see section \ref{ThefamilyofoperatorsPandT} in the bulk of the paper), condition iii) ensures that this kernel is in fact completely described by the linearised Kerr solution $\Ke_{\mathsmaller{\mfM, \mathfrak{a}}}$.

Of course, the question remains as to whether there exist solutions that satisfy conditions i) and ii) which moreover verify the required conditions on asymptotic regularity. That such a gauge can indeed be realised is the content of Theorem \ref{thminitialdatagauge} in the bulk of the paper which states that, given the solution $\So$ to the equations of linearised gravity arising from the seed data set $\Sos$ courtesy of Theorem \ref{OVthmwellposedness}, one can add to it a pure gauge solution $\GfMa$ for which the resulting solution $\SfMa$ is in a $\Ke_{\mathsmaller{\mfM, \mathfrak{a}}}$-adapted Regge--Wheeler gauge gauge. The precise gauge one is in is to be determined explicitly from the seed data $\Sos$, part of which contains
\begin{itemize}
	\item a smooth function $\alin$ on the horizon sphere $\hplusi$ lying in the span of the $l=1$ spherical harmonics
	\item a constant $\mfmlin$ 
\end{itemize}
The latter thus determines the parameter $\mfM$ whereas the $Y^1_{i}$ modes of the former determine the parameters $\mfa_{-1}, \mfa_0$ and $\mfa_1$. The theorem further states that if the seed data, which we now explicitly denote by $\SosMa$, is moreover asymptotically flat then this property is inherited by the pure gauge solution $\GfMa$ and hence the solution $\SfMa$. 

An important property of the pure gauge solution $\GfMa$ that one can conclude from the proof of Theorem \ref{thminitialdatagauge} is that the initial data from which it arises is actually constructed explicitly\footnote{In particular, in the construction, one is only required to perform `operations' on the initial hypersurface $\Sigma$.} from the seed data of the solution $\So$ alone. For instance, by equation \eqref{OVeqnlingrav4} to ensure that condition ii) holds it suffices to impose trivial Cauchy data for $\sghatlin$ on $\Sigma$. See section \ref{GlobalpropertiesofSf} in the bulk of the paper.

Consequently, a $\Ke_{\mathsmaller{\mfM, \mathfrak{a}}}$-adapted Regge--Wheeler gauge is in fact realisable \emph{purely from seed data $\SosMa$ alone} and can therefore be said to be \emph{well-posed}. The solution $\SfMa$ is thus said to be \emph{initial-data-normalised}.\newline

We remark that previous works which have employed the Regge--Wheeler gauge have done so within the context of the linearised Einstein equations where no gauge has been fixed, a system of equations which are \emph{not} well-posed. In particular, \emph{it is not possible within the context of that formulation to identify the Regge--Wheeler gauge as an `initial-data-gauge'}. Subsequently, the novelty of our work is to thus appreciate the remarkably useful Regge--Wheeler gauge within the framework of a \emph{well-posed} formulation of linearised gravity around Schwarzschild, namely by imposing a generalised wave gauge in the full nonlinear theory and then linearising. Our sagacious choice of the linear map $Df\big|_{g_M}$ in section \ref{OVTheequationsoflinearisedgravity} then allows for the Regge--Wheeler gauge to be realised as a \emph{residual gauge choice at the level of initial data}. We remark that, in addition to our earlier noted motivation for the choice of the map $Df\big|_{g_M}$, we were further inspired in making this choice by inserting the expression for the linearised metric quantities as they appear in the Regge--Wheeler gauge into the linearised Einstein equations \eqref{OVeqnlingrav1}-\eqref{OVeqnlingrav2}, where the expression $\flin$ has not yet been fixed, and then evaluating the `error terms'.

Of course, there is the point of view in which the Regge--Wheeler gauge is considered to be a gauge choice in of itself, thereby forgoing any reference to the larger framework of the generalised wave gauge. Indeed, by virtue of the well-posedness of the respective Regge--Wheeler and Zerilli equations, one can determine the gauge-invariant quantities $\Philin$ and $\Psilin$ explicitly from a prescription of seed data. This subsequently determines the collection of quantities on $\mcalm$ appearing on `the right hand side' of \eqref{a1}-\eqref{a3}. Inverting the operators $\mcaltf$ and $\mcaltv$ thus \emph{a posteriori} constructs a solution to the linearised Einstein equations that arises purely from the seed data and which is manifestly in the Regge--Wheeler gauge. However, \emph{this procedure rests entirely on the decoupling of the quantities $\Philin$ and $\Psilin$}, a phenomena which is completely artificial\footnote{Indeed, see the paper of \cite{K--S} where the quantity $\mfqlin$, which in the linear theory completely decouples, but in the nonlinear theory does quite the opposite!} to the linearised setting. Consequently, with nonlinear applications in mind, it is indeed more appropriate to view the Regge--Wheeler gauge as a residual gauge choice associated to a judicious choice of a (well-posed) generalised wave gauge.

\subsection{The main theorems and outline of the proof}\label{OVThemaintheoremandoutlineoftheproof}

In this fifth part of the overview we discuss the statement and proof our main theorem from which one concludes the linear stability of the Schwarzschild exterior family as solutions to the Einstein vacuum equations in a generalised wave gauge.

This part of the overview corresponds to sections \ref{Precisestatementsofthemaintheorems}-\ref{Proofoftheorem1} in the main body of the paper.

\subsubsection{Theorem \ref{OVthmboundedness}: Boundedness and decay of the solution $\SfMa$}\label{OVTheorem1:boundednessofthesolutionSi}

We begin with a rough statement of the main theorem which concerns a boundedness and decay statement for the initial-data-normalised solution $\SfMa$ and which is a more precise version of the Theorem discussed in the introduction.

The precise statement of the theorem can be found in section \ref{BoundednessforsolutionsintheKerradaptedtracelessgauge} in the main body of the paper.\newline

The  statement in question involves certain natural $r$-weighted energy and integrated decay norms on hypersurfaces which penetrate both the future event horizon and future null infinity which are generalisations of the norms introduced in section \ref{OVAside:thescalarwaveequationontheSchwarzschildexteriorspacetime} for scalar waves to the full system of gravitational perturbations (in particular, tensor fields on $\mcalm$). See section \ref{Fluxandintegrateddecaynorms} in the bulk of the paper for the full description. Conversely, a description of the hypersurfaces can be found in section \ref{AfoliationofMthatpenetratesbothhplusandiplus} although see Figure 2 for a Penrose diagram depicting this foliation of $\Mgs$.

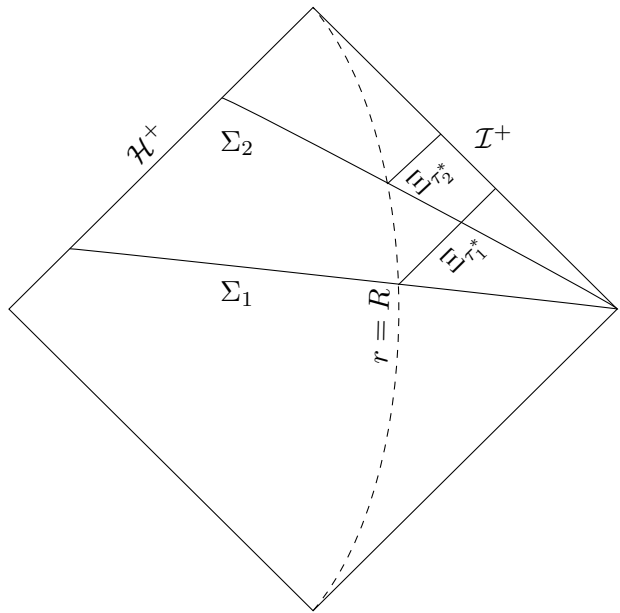
\begin{figure}[!h]
	\centering
	\begin{tikzpicture}
	\node (I)    at ( 4,0)   {};
	\path 
	(I) +(90:4)  coordinate (Itop)
	+(-90:4) coordinate (Ibot)
	+(180:4) coordinate (Ileft)
	+(0:4)   coordinate (Iright)
	;
	\draw  (Ileft) -- (Itop)node[midway, above, sloped] {$\hplus$} -- (Iright) node[midway, above right]    {$\cal{I}^+$} -- (Ibot)node[midway, below right]    {} -- (Ileft)node[midway, below, sloped] {} -- cycle;
	
	\draw[dashed]  
	(Ibot) to[out=50, in=-50, looseness=0.75] node[pos=0.475, above, sloped] {$r=R$        }($(Itop)!.5!(Itop)$) ;
	
	\draw 
	($(Itop)!.8!(Ileft)$) to[out=0, in=0, looseness=0.05] ($(Iright)!.5!(Iright)$);
	
	\draw 
	($(Itop)!.3!(Ileft)$) to[out=0, in=0, looseness=0.05] ($(Iright)!.5!(Iright)$);
	
	\draw 
	($(Itop)!.6!(Iright)$) --node[midway, below, sloped]{$\Xi _{\tau_1 ^*}$} (5.12, 0.32);
	
	\draw 
	($(Itop)!.42!(Iright)$) --node[right, below, sloped]{$\Xi _{\tau_2 ^*} $} (4.975, 1.66);
	
	\node [below] at (3,0.5) {$\Sigma_1$};
	
	\node [below] at (3,2.5) {$\Sigma_2$};
	
	\end{tikzpicture}
	\caption{A Penrose diagram of $\Mgs$ depicting the hypersurfaces $\Xi_{\taus}$ which penetrate both $\hplus$ and $\iplus$. Here, the hypersurfaces $\Sigma_{t^*}$ are level sets of the time function $t^*$.} 
	\label{Figure4}
\end{figure}

A rough formulation of the theorem is given below.

 In what follows, we shall employ the schematic notation of an estimate for a norm on the solution $\So$ to denote that said estimate holds, in that norm, for any quantity contained within the collection $\So$.

\begin{customthm}{1}\label{OVthmboundedness}
Let $\So$ be the smooth solution to the equations of linearised gravity arising from a smooth, asymptotically flat seed data set $\SosMa$ on $\Sigma$ in accordance with Theorem \ref{OVthmwellposedness}. We consider the initial-data-normalised solution
\begin{align*}
\SfMa=\So+\GfMa
\end{align*}
of section \ref{OVGaugenormalisedsolutionstotheequationsoflinearisedgravity} and define its projection
\begin{align*}
\So':=\SfMa-\Ke_{\mathsmaller{\mfM, \mathfrak{a}}}.
\end{align*}
Then the solution $\So'$ satisfies the following $r$-weighted energy and integrated decay estimates, concerning in particular up to $6$ angular derivatives of $\So'$, with the initial $r$-weighted energies on $\Sigma$ finite by assumption of asymptotic flatness.
\begin{enumerate}[i)]
	\item the flux estimates
	\begin{align}\label{cough1}
	\mathbb{F}^{5, \sn}[r^{-\frac{3}{2}}\So']\lesssim\mathbb{D}^5[\So'].
	\end{align}
	\item the integrated decay estimates
	\begin{align}\label{cough2}
	\mathbb{M}^{4, \sn}[r^{-\frac{3}{2}}\So']+\mathbb{I}^{5, \sn}[r^{-\frac{3}{2}}\So']\lesssim\mathbb{D}^5[\So'].
	\end{align}
\end{enumerate}
\end{customthm}
The flux norms in \eqref{cough1} and the integrated decay norms in \eqref{cough2} are such that, via now the Dafermos--Rodnianski hierarchy \cite{newmethod} of $r$-weighted energies to be discussed in section \ref{OVAside:thescalarwaveequationontheSchwarzschildexteriorspacetime} of the overview, one has by Sobolev embedding the corollary:
\begin{corollary*}[Uniform pointwise decay]
	The solution $\So'$ of Theorem 1, in addition to satisfying the uniform $r$-weighted pointwise bounds
	\begin{align*}
	\big|r^{-\frac{1}{2}}\Sf'\big|\lesssim\mathbb{D}^5[\So'],
	\end{align*}
	in fact decays to the future at an inverse polynomial rate:
	 \begin{align*}
	 \big|r^{-\frac{1}{2}}\Sf'\big|\lesssim\frac{1}{\sqrt{\taus}}\cdot\mathbb{D}^5[\So'].
	 \end{align*}
	
	In particular, the solution $\SfMa$ decays inverse polynomially to the linearised Kerr solution $\Ke_{\mathsmaller{\mfM, \mathfrak{a}}}$.
\end{corollary*}
Here, $\taus$ is a function on $\mcalm$ the level sets of which correspond to the aforementioned hypersurfaces which penetrate both $\hplus$ and $\iplus$. See section \ref{AfoliationofMthatpenetratesbothhplusandiplus} in the bulk of the paper for a precise definition. 

We make the following remarks regarding Theorem \ref{OVthmboundedness}.

The first remarks concern the loss of derivatives in the statement of the theorem. Indeed, the loss of derivatives (at the top order) in \eqref{cough2} occurs due to the celebrated trapping effect and cannot be removed (see sections \ref{OVAside:thescalarwaveequationontheSchwarzschildexteriorspacetime} for further discussion. However, the loss in all but angular derivatives arises as a consequence of applying elliptic estimates on the angular operators $\mcaltf$ and $\mcaltv$ to the solution $\So'$ as given in the $\Ke_{\mathsmaller{\mfM, \mfa}}$-adapted Regge--Wheeler gauge and is recoverable with further work (see our upcoming \cite{J}).

The second remark concerns asymptotic flatness. Indeed, it is clear from the Corollary that pointwise asymptotic flatness is not preserved for the solution $\So'$. Moreover, there is a loss in $r$-weights in the norms of the Theorem statement\footnote{We are quite wasteful with the loss here -- in fact, for any quantity other than $\mling$ associated to the solution $\So$ one can replace the weight $r^{-\frac{3}{2}}$ to $r^{-1}$.}. However, this can be rectified by modifying the choice of the linear map $Df\big|_{g_M}$ in section \ref{OVTheequationsoflinearisedgravity} and this modification shall be performed in our upcoming \cite{J}. Since this procedure is slightly cumbersome however, for the purposes of this paper we prefer the simpler choice of gauge which we have utilised throughout. 

The third remark concerns the initial data norm required in the statement of the theorem. Indeed, we only requires a norm on the gauge-invariant quantities $\Philin$ and $\Psilin$ to be finite initially, and whether this is so can be verified explicitly from the seed data alone. We further note that this initial data norm in fact  propagates under evolution.

The final remark concerns whether a variant of the above Theorem holds for other potential choices of a generalised wave gauge. Indeed, the reader would agree that the most natural choice of the map $\boldsymbol{f}$ in section \ref{OVTheequationsoflinearisedgravity} would be the trivial one. In this case, the additional tensorial structure of the resulting equations \eqref{OVeqnlingrav1}-\eqref{OVeqnlingrav2} makes there analysis rather complicated. However, by appealing to a remarkable Maxwell-like structure within those equations, one can indeed formulate a weak version of the above Theorem for the linearised Einstein equations in such a gauge. See our upcoming \cite{J}.

\subsubsection{Aside: The scalar wave equation on the Schwarzschild exterior spacetime}\label{OVAside:thescalarwaveequationontheSchwarzschildexteriorspacetime}

We now make a brief aside to discuss a certain analogue of Theorem 1 which is known to hold for the simpler case of the scalar wave equation $\Box_{g_M}\psi=0$ on $\Mgs$ -- the main objective here then is to motivate the norms that appear in the statement of Theorem 1. Conversely, we point the reader to \cite{D--H--R} for a definitive review as to how one obtains these results.\newline

We consider a smooth solution $\psi$ to the scalar wave equation on Schwarzschild:
\begin{align}\label{OVwaveeqn}
\Box_{g_M}\psi=0.
\end{align}
We associate to $\psi$ the flux norms:
\begin{align}
\mathbb{F}_p[\psi]:&=\sup_{\taus\geq\taus_0}\int_{\Xi_{\taus}\cap\{r\leq R\}}\Big(|\pt  \psi|^2+|\pr  \psi|^2+|\sn  \psi|^2\Big)\ud r\ud\hat{\sigma}+\sup_{\taus\geq\taus_0}\int_{\Xi_{\taus}\cap\{r\geq R\}}\Big(r^p|\pv  (r\psi)|^2+|\sn  (r\psi)|^2\Big)\ud v\ud\hat{\sigma},\label{OVflux}\\
\mathbb{D}_p[\psi]:&=\int_{\Sigma}r^{p}\Big(|\pt  (r\psi)|^2+|\pr  (r\psi)|^2+|\sn  (r\psi)|^2\Big)\ud r\ud\hat{\sigma}\label{OVinitialflux}
\end{align}
along with the integrated decay norms:
\begin{align}
\mathbb{I}_p[\psi]:=&\int_{\taus_0}^\infty\int_{\Xi_{\taus}\cap\{r\leq R\}}(r-3M)^2\Big(|\pt  (r\psi)|^2+|\pr  (r\psi)|^2+|\sn  (r\psi)|^2+|  (r\psi)|^2\Big)\ud\taus\ud r\ud\hat{\sigma}\nonumber\\
+&\int_{\taus_0}^\infty\int_{\Xi_{\taus}\cap\{r\geq R\}}r^{p-1}\Big(|\pv  (r\psi)|^2+(2-p)|\sn  (r\psi)|^2\Big)\ud\taus\ud v\ud\hat{\sigma},\label{OViled}\\
\mathbb{M}[\psi]:=&\int_{\taus_0}^\infty\int_{\Xi_{\taus}}r^{-3}\Big(|\pt  (r\psi)|^2+|\pr  (r\psi)|^2+|\sn  (r\psi)|^2+|  (r\psi)|^2\Big)\ud\taus\ud r\ud\hat{\sigma}\label{OVied}.
\end{align}
Here, $\hat{\sigma}$ is the volume form on the unit round sphere, $\pv$ is the null vector\footnote{In the bulk of the paper we actually use a null geodesic vector field to define both the norm and the foliation as this is more geometric - the choice of $\pv$ here is for ease of presentation. A similar remark holds for the function $\taus$.} (in Schwarzschild-star coordinates)
\begin{align*}
\pv:=\opmu\pt+\ommu\pr
\end{align*}
and $\Xi_{\taus}$ is the level set of the function:
\begin{align*}
\tau^\star:=\begin{cases} 
t^* & r\leq R \\
u & r\geq R,
\end{cases}
\end{align*}
where $\pv u=0$ and $u|_{r=R}=t^*$. The flux norms \eqref{OVflux} and \eqref{OVinitialflux} thus denote energy norms containing all tangential and normal derivatives to $\Xi_{\taus}$ and $\Sigma$.

The above defined norms\footnote{The full norms we consider actually require supplementing the above with extra flux norms, in particular a flux norm along future null infinity. See section \ref{Fluxandintegrateddecaynorms}.} are the scalar wave prototypes of those found in the statement of Theorem 1 where $p=2$. Indeed, one has the following theorem which yields an analogue of Theorem 1 for the scalar wave equation on Schwarzschild.

\begin{theorem*}[Dafermos--Rodnianski - \cite{D--R}, \cite{newmethod}]
	Let $\psi$ be a smooth solution to \eqref{OVwaveeqn}. Then for any $n\geq1$ the following estimates hold, provided that the fluxes on the right hand side are finite.
	\begin{enumerate}[i)]
		\item for $0\leq p\leq 2$ the $r$-weighted flux estimates
		\begin{align}\label{OVfluxestimate}
		\mathbb{F}^n_p[\psi]+\mathbb{I}^n_p[\psi]\lesssim\mathbb{D}^n_p[\psi].
		\end{align}
		\item the integrated decay estimate
		\begin{align}\label{OViedestimate}
		\mathbb{M}^{n-1}[\psi]\lesssim\mathbb{D}^{n}[\psi].
		\end{align}
	\end{enumerate}
\end{theorem*}
Here, the above are natural higher order norms defined by replacing $\psi$ in \eqref{OVflux}-\eqref{OVied} with the appropriate derivatives. 

We note the loss of derivative in the integrated decay estimate - this is a consequence of the celebrated trapping effect on black hole spacetimes, arising in this instance from the existence of trapped null geodesics at $r=3M$. Indeed, a result of Sbierski \cite{Sbierski} shows that a statement such as \eqref{OViedestimate} cannot hold without a loss of derivative. In fact, the authors of \cite{Marz--Met--Tat--Toh} improve the loss in \eqref{OViedestimate} to only a logarthimic loss.

\subsubsection{Outline of the proof of Theorem 1}\label{OVOutlineoftheproofI-boundednessanddecayforsolutionstoequationsofRegge--WheelertypeontheSchwarzschildexteriorspacetime}

We return now to the equations of linearised gravity by discussing the proofs of Theorems 1. In fact, as previously discussed in section \ref{OVGaugenormalisedsolutionstotheequationsoflinearisedgravity}, the proof essentially follows from the ellipticity of the operators $\mcaltf$ and $\mcaltv$ combined with the following Theorem regarding solutions to the Regge--Wheeler and Zerilli equations and shall therefore not be discussed further.

For further details of the proof, see section \ref{Proofoftheorem1} in the bulk of the paper.\newline

Indeed, the result we are to use is as follows.
\begin{customthm}{2}
	Let $\Phi, \Psi\in\LM$ be smooth solutions to the Regge--Wheeler and Zerilli equations respectively:
	\begin{align*}
	\qbox\Phi+\slap\Phi=-\frac{6}{r^2}\frac{M}{r}\Phi,\qquad\qbox\Psi+\slap\Psi=-\frac{6}{r^2}\frac{M}{r}\Psi+\frac{24}{r^3}\frac{M}{r}(r-3M)\szetap{1}\Psi+\frac{72}{r^5}\frac{M}{r}\frac{M}{r}(r-2M)\szetap{2}\Psi.
	\end{align*}
	We assume finiteness of the initial flux norms
	\begin{align*}
	\norm{D}{r^{-1}\Phi, r^{-1}\Psi}[5].
	\end{align*}
	Then the following estimates hold.
	\begin{enumerate}[i)]
		\item the flux estimates 
		\begin{align*}
		\norm{F}{r^{-1}\Phi, r^{-1}\Psi}[5]&\lesssim\norm{D}{r^{-1}\Phi, r^{-1}\Psi}[5]
		\end{align*}
		\item the integrated decay estimates
		\begin{align*}
		\norm{M}{r^{-1}\Phi, r^{-1}\Psi}[4]+\norm{I}{r^{-1}\Phi, r^{-1}\Psi}[5]&\lesssim\norm{D}{r^{-1}\Phi, r^{-1}\Psi}[5]
		\end{align*}
	\end{enumerate}
	In addition, for any $\taus\geq\taus_0$, one has the decay estimates
	\begin{align*}
	\norm{E}{r^{-1}\Phi, r^{-1}\Psi}[3](\tau^\star)\lesssim\frac{1}{{\tau^\star}^2}\cdot\norm{D}{r^{-1}\Phi, r^{-1}\Psi}[5].
	\end{align*}
\end{customthm}
We note that the presence of the $r$-weight in the theorem statement is necessary as one expects solutions to the Regge--Wheeler and Zerilli equations to be only bounded on future null infinity. Indeed
\begin{align*}
\qbox\cdot+\slap\cdot=r\,\Box(r^{-1}\cdot)+\frac{1}{r}\frac{\mu}{r}\cdot.
\end{align*}

That a result such as Theorem 2 holds in by now well-known in the literature (see \cite{Holz}, \cite{B--S}, \cite{Me} and \cite{H--K--W}) and shall therefore be applied freely in this paper, although for completeness we shall present a proof in our upcoming \cite{J}.\newline\newline

So ends our detailed overview regarding the linear stability of the Schwarzschild exterior family as solutions to the Einstein vaccum solutions when expressed in a generalised wave gauge. Before we begin the paper proper however in the next section we shall state the conjecture due to Dafermos, Holzegel and Rodnianski that relates to a restricted nonlinear stability result for the Schwarzschild exterior family which the results of this paper should be in principle sufficient to resolve in the affirmative. 

\subsection{The restricted nonlinear stability conjecture of Dafermos, Holzegel and Rodnianski}\label{OVArestrictednonlinearstabilityconjecture}

The conjecture, lifted verbatim from \cite{D--H--R}, is as follows.
\begin{conjecture*}[Dafermos--Holzegel--Rodnianski, \cite{D--H--R}]
Let $(\Sigma_M, \bar{g}_M, K_M)$ be the induced data on a spacelike asymptotically flat slice of the Schwarzschild solution of mass M crossing the future horizon and bounded by a trapped surface. Then in the space of all nearby vacuum data $(\Sigma, \bar{g}, K)$, in a suitable norm, there exists a codimension-3 subfamily for which the corresponding maximal vacuum Cauchy development $(\mcalm, g)$ contains a black-hole exterior region (characterized as the past $J^-(\iplus)$ of a complete future null infinity $\mcali^+$), bounded by a non-empty future affine-complete event horizon $\hplus$, such that in $J^-(\iplus)$ (a) the metric remains close to $g_M$ and moreover (b) asymptotically settles down to a nearby Schwarzschild metric $g_{\tilde{M}}$ at suitable inverse polynomial rates.
\end{conjecture*}

Consequently, the body of work we are about to present, discussed in detail over the previous sections of the overview, is in principle sufficient to try and resolve the above conjecture in the affirmative.\newline

\subsection{Acknowledgements}

I would like to thank my advisor Gustav Holzegel for the many helpful comments, insights and assistance provided over the period of time in which this paper was written. I would also like to thank Gustav Holzegel and Mihalis Dafermos for their invaluable help in preparing this document. I also thank Martin Taylor and Joe Keir for many helpful and interesting discussions over the years.  We also emphasize once more the significant role played by the paper \cite{D--H--R} in the organisational structure of this paper.

Finally, we acknowledge support through an ERC Starting Grant of Gustav Holzegel and the EPSRC grant
EP/K00865X/1.

\section{The Einstein vacuum equations in a generalised wave gauge}\label{ThevacuumEinsteinequationsinageneralisedwavegauge}

In this section we introduce the notion of a generalised wave gauge on an abstract Lorentzian manifold and then review the structure of the Einstein vacuum equations when expressed in such a gauge.

It is these equations that we shall linearise about a fixed Schwarzschild exterior solution in section \ref{TheequationsoflinearisedgravityaroundSchwarzschild}.

\subsection{The generalised wave gauge}\label{Thegeneralisedwavegauge}

In this section we provide the definition of a generalised wave gauge on an abstract Lorentzian manifold as it is found in \cite{C-Bbook}.

We note that the specification of such a gauge plays a vital role in the works \cite{H--V} and \cite{H}.\newline

Let $\big(\bold{\mcalm}, \bold{g}\big)$ and $\big(\bold{\mcalm}, \bold{\overline{g}}\big)$ be $3+1$ Lorentzian manifolds with\footnote{Here we recall the notation $T^k(\mcalm)$ for the space of $k$-covariant tensor fields on $\mcalm$.} $\boldsymbol{f}:T^2(\bmcalm)\times T^2(\bmcalm)\rightarrow T\bmcalm$ a smooth map. 

Then we say that $\boldsymbol{g}$ is in a generalised $\boldsymbol{f}$-wave gauge with respect to $\boldsymbol{\overline{g}}$ iff the identity map
\begin{align*}
\text{Id}:\big(\boldsymbol{\mcalm}, \boldsymbol{g}\big)\rightarrow\big(\boldsymbol{\mcalm}, \boldsymbol{\overline{g}}\big)
\end{align*}
is an $\boldsymbol{f}(\bg, \bog)$-wave map. Denoting by $\boldsymbol{C_{g, \overline{g}}}$ the connection tensor of $\bg$ and $\bog$, 
\begin{align}\label{eqndefnconnectiontensor}
\big(\boldsymbol{C_{g, \overline{g}}}\big)\boldsymbol{^\alpha_{\beta\gamma}}:=\frac{1}{2}\bold{\big(g^{-1}\big)}[][\alpha\delta]\Big(2\bold{\overline{\nabla}}[(\beta]\bold{g}[\gamma)\delta]-\bold{\overline{\nabla}}[\delta]\bold{g}[\beta\gamma]\Big)
\end{align}
with $\boldsymbol{\overline{\nabla}}$ the Levi-Civita connection associated to $\boldsymbol{\overline{g}}$, the imposition of such a gauge is equivalent to the condition
\begin{align}\label{eqndefnofageneralisedwavegauge}
\boldsymbol{g^{-1}}\cdot\boldsymbol{C_{g, \overline{g}}}=\boldsymbol{f}(\bg, \bog).
\end{align}

\subsection{The Einstein equations}\label{TheEinsteinequations}

In this section we present the vacuum Einstein equations assuming that a generalised wave has been imposed.\newline

Indeed, if $\bg$ is in a generalised $\boldsymbol{f}$--wave gauge with respect to $\bog$ the vacuum Einstein equations for $\boldsymbol{g}$,
\begin{align*}
\bold{\textbf{Ric}}[\alpha\beta] [\bold{g}]=0,
\end{align*}
take the following form:
\begin{align}\label{einstein equations in a generalised wave gauge}
\bold{\big(g^{-1}\big)}[][\gamma\delta]\bold{\overline{\nabla}}[\gamma]\bold{\overline{\nabla}}[\delta]\bold{g}[\alpha\beta]+2\boldsymbol{C^{\gamma}_{\delta\epsilon}}\cdot\bold{g}[\gamma(\alpha]\bold{\on}[\beta)]\bold{\big(g^{-1}\big)}[][\delta\epsilon]-4\bold{g}[\delta\epsilon]\boldsymbol{C^{\epsilon}_{\beta[\alpha}}\bold{\on}[{\gamma]}]\bold{\big(g^{-1}\big)}[][\gamma\delta]-4\boldsymbol{C^\delta_{\beta[\alpha}}\boldsymbol{C^\gamma_{\gamma]\delta}}+2\bold{g}[][\gamma\delta]\bold{g}[\epsilon(\alpha]\bold{\overline{\textbf{Riem}}}[\beta)\gamma\delta][\epsilon]\nonumber\\=2\bold{g}[\gamma(\alpha]\bold{\overline{\nabla}}[\beta)]\bold{f}[][\gamma](\bg, \bog).
\end{align}
Here, $\boldsymbol{\overline{\textbf{Riem}}}$ is the Riemann tensor of $\boldsymbol{\overline{g}}$ and $\boldsymbol{C}$ is defined as in \eqref{eqndefnconnectiontensor}.

\section{The exterior Schwarzschild background}\label{The exteriorSchwarzschildbackground}

In this section we define the Schwarzschild exterior spacetime as well as introducing various background objects and operations that shall prove vital throughout the remainder of the paper.

\subsection{The differential structure and metric of the Schwarzschild exterior spacetime}\label{ThedifferentialstructureandmetricoftheSchwarzschildexteriorspacetime}

We begin in this section by defining the differential structure and metric of the Schwarzschild exterior spacetime $\big(\mcalm,g_M\big)$.\newline

Let $M>0$ be a fixed parameter.

 We define the smooth manifold with boundary 
\begin{align*}
\mcalm:=(-\infty, \infty)\times [2M, \infty)\times S^2
\end{align*}
and endow it with the coordinate system $\big(t^*, r, \theta, \varphi\big)$. Here $S^2$ is the 2-sphere with $(\theta, \varphi)$ the standard angular coordinates. Equipping $\mcalm$ with the smooth Ricci-flat Lorentzian metric
\begin{align}\label{schwarzschildmetric}
g_M=-\bigg(1-\frac{2M}{r}\bigg)\ud {t^*}^2+\frac{4M}{r}\ud t^*\ud r+\bigg(1+\frac{2M}{r}\bigg)\ud r^2 +r^2\mathring{g},
\end{align}
with $\mathring{g}$ the metric on the unit round sphere, thus defines the Schwarzschild exterior spacetime (of mass $M$) as the Lorenztian manifold with boundary $\Mgs$. It is moreover time-orientable, with time-orientation given by the hypersurface-orthogonal vector field
\begin{align*}
T=\pt.
\end{align*}

The boundary of $\mcalm$, which we denote by
\begin{align*}
\mathcal{H}^+:=(-\infty,\infty)\times\{2M\}\times S^2,
\end{align*}
is a null hypersurface termed \emph{the future event horizon}. 

Subsequently, the coordinate system described by the coordinates $(t^*, r, \theta, \varphi)$, the standard degeneration of the angular coordinates understood, defines\footnote{Strictly speaking there is a ambiguity up to translations by a constant in the time function $t^*$. This will be removed in section \ref{TheCauchyhypersurfaceSigmaandsmitensoranalysis} by the specification of the initial hypersurface $\Sigma$.} the so-called \emph{Schwarzschild-star} coordinate system on $\mcalm$. Observe thus that the coordinate $r$ is an area radius function for the 2-spheres:
\begin{align*}
\textnormal{Area}\big(S^2_{t^*,r}\big)=4\pi r^2
\end{align*}
with
\begin{align*}
S^2_{t^*,r}:=\{t\}\times\{r\}\times S^2\subset\mcalm.
\end{align*}
It is moreover manifest that the generators $\Omega$ of the rotation group $SO(3)$ and the causal vector field $T$
are Killing fields for $g_M$:
\begin{align*}
\mcall_\Omega g_M=0,\qquad\mcall_Tg_M=0.
\end{align*}
The Schwarzschild exterior spacetime is thus both \emph{static} and \emph{spherically symmetric}. \newline

A Penrose diagram of the exterior Schwarzschild spacetime can be found in section \ref{Thetaustarfoliation}.

\subsection{The 2+2 formalism}\label{The2+2formalism}

We continue in this section by detailing the so-called $2+2$ formalism on $\mcalm$.

We shall make heavy use of the enlargened mathematical toolbox this formalism provides throughout the paper.

\subsubsection{The 2+2 decomposition of tensor fields on $\mcalm$}\label{The2+2decompositionoftensorfieldsonM}

We begin by employing this formalism to decompose tensor fields on $\mcalm$ into $\mcalq$-tensors, $S$-tensors and $\mcalq\otimes S$-tensors respectively.\newline

Observe that one can express the manifold $\mcalm$ as
\begin{align*}
\mcalm&=\mcalq\times S^2
\end{align*}
where\footnote{Here $\mathcal{H}$ is the half-space, not to be confused with the event horizon.} 
\begin{align*}
\mcalq\cong\mathbb{R}\times\mathcal{H}
\end{align*}
is a manifold with boundary. Consequently, if $V$ is a vector field on $\mcalm$ then we say
\begin{itemize}
	\item $V$ is a $\mcalq$-vector field iff
	$V\Big(f\big|_{S^2}\Big)=0$
	for all smooth functions $f\big|_{S^2}$ on $S^2$
	\item  $V$ is an $S$-vector field iff
	$V\Big(f\big|_{\mcalq}\Big)=0$
	for all smooth functions $f\big|_{\mcalq}$ on $\qm$
\end{itemize}
This leads to the following definition.
\begin{definition}
	Let $\mfT$ be an $n$-covariant tensor field on $\mcalm$. Then we say that $\mfT$ is an $n$-covariant $\qm$-tensor field iff $\mfT$ vanishes when acting on any $S$-vector field $\sV$:
	\begin{align*}
	\mfT\big(\cdot,...,\sV,...,\cdot\big)=0.
	\end{align*}
	Conversely, we say that $\mfT$ is an $n$-covariant $\sm$-tensor field iff $\mfT$ vanishes when acting on any $\qm$-vector field $\qV$:
	\begin{align*}
	\mfT\big(\cdot,...,\qV,...,\cdot\big)=0.
	\end{align*}
	Finally, if $\mfT$ is a symmetric 2-covariant tensor field on $\mcalm$, then we say that $\mfT$ is a $\qm\otimes\sm$ 1-form iff $\mfT$ vanishes when acting purely on $\qm$-vector fields $\qV$ or purely on $\sm$-vector fields $\sV$:
	\begin{align*}
	\mfT\big(\qV_1, \qV_2\big)=\mfT\big(\sV_1,\sV_2)=0.
	\end{align*}
\end{definition}
Note by convention we set a 0-covariant $\qm$-tensor field and a 0-covariant $\sm$-tensor field to be simply a scalar field on $\mcalm$.\newline

Given any tensor field on $\mcalm$ one projects it onto $\qm$-tensor fields and $\sm$-tensor fields as follows.

First, let $V$ be a vector field on $\mcalm$. Then we define
\begin{itemize}
	\item the projection of $V$ onto $\qmm$ is the $\qm$-vector field $\qV$ defined by $\qV\Big(f\big|_{\qm}\Big)=V\Big(f\big|_{\qm}\Big)$ for every smooth function $f\big|_{\qm}$ on $\qm$
	\item the projection of $V$ onto $\smm$ is the $\sm$-vector field $\sV$ defined by $\sV\Big(f\big|_{\sm}\Big)=V\Big(f\big|_{\sm}\Big)$ for every smooth function $f\big|_{\sm}$ on $S^2$
\end{itemize}
This leads to the subsequent definition.
\begin{definition}\label{defnprojection}
	Let $\mfT$ be an $n$-covariant tensor field on $\mcalm$. Then we define the projection of $\mfT$ onto $\qmm$ to be the $n$-covariant $\qm$-tensor $\widetilde{\mfT}$ defined by
\begin{align*}
\widetilde{\mfT}\big(V_1, ..., V_n\big)=\mfT\big(\qV_1,...,\qV_n\big),
\end{align*}
where $V_1,..., V_n$ are an $n$-tuple of vector fields on $\mcalm$ with $\qV_1,..., \qV_n$ the relative projections onto $\qmm$.

Conversely, we define the projection of $\mfT$ onto $\smm$ to be the $n$-covariant $\sm$-tensor $\slashed{\mfT}$ defined by
\begin{align*}
\slashed{\mfT}\big(V_1, ..., V_n\big)=\mfT\big(\sV_1,...,\sV_n\big),
\end{align*}
where $V_1,..., V_n$ are an $n$-tuple of vector fields on $\mcalm$ with $\sV_1,..., \sV_n$ the projections onto $\smm$.

Finally, if $\mfT$ is a symmetric 2-covariant tensor field on $\mcalm$ then we define the projection of $\mfT$ onto $\qmm\times\smm$ to be the $\qmsm$ 1-form $\stkout{\mfT}$ defined by
\begin{align*}
\text{\sout{\ensuremath{\mfT}}}\big(V_1, V_2\big)=\mfT\big(\qV_1, \sV_2\big)+\mfT\big(\sV_1, \qV_2\big).
\end{align*}
\end{definition}
Note that a symmetric 2-covariant tensor $\mfT$ is completely specified by the projections $\widetilde{\mfT}, \stkout{\mfT}$ and $\slashed{\mfT}$.\newline

One can use this to define the projection of maps
\begin{align*}
\qf(\mfT, \mfT'):=f(\widetilde{\mfT}, \widetilde{\mfT}')
\end{align*}
A particularly useful application of this decomposition is to the Schwarzschild metric $g_M$. Indeed
\begin{itemize}
	\item the projection of $g_M$ onto $\qmm$ results in the symmetric 2-covariant $\qm$-tensor $\qg_M$ which we refer to as a $\qm$-metric 
	\item the projection of $g_M$ onto $\qmm\times\smm$ is trivial
	\item the projection of $g_M$ onto $\smm$ results in the symmetric 2-covariant $\sm$-tensor $\sg_M$ which we refer to as an $\sm$-metric
\end{itemize}
Note that, in light of spherical symmetry,
\begin{align*}
\mathcal{L}_{\Omega}\qg_M=0,\qquad
\sg_M=r^2\mathring{g}.
\end{align*} 

\subsubsection{Tensor analysis}\label{Tensoranalysis}

We now develop a series of natural operations and differential operators on tensor fields on $\mcalm$ that arise as a result of the $2+2$ formalism of the previous section, in particular the $2+2$ decomposition of the metric $g_M$.

We begin with the operations.

In what follows, $\qepsilon$ and $\sepsilon$ are the unique 2-forms on $\mcalm$ such that $-\qepsilon\qdot\qepsilon=\sepsilon\sdot\sepsilon=2$. Moreover, $\mathfrak{T}$ and $\mathfrak{T}'$ denote $n$-covariant and $m$-covariant tensor fields on $\mcalm$ respectively with $\mathfrak{t}$ and $\mft'$ denoting 1-forms on $\mcalm$. Finally, we shall employ abstract index notation.
\begin{itemize}
\item the index raising operators $\tilde{\sharp}$ and $\slashed{\sharp}$ are defined by  
\begin{align*}
\Big(\mfT^{\tilde{\sharp}}\Big)^{a_1}_{\ \ a_2...a_n}:=\qg_M^{a_1b}\mfT_{ba_2...a_n},\qquad\Big(\mfT^{\slashed{\sharp}}\Big)^{a_1}_{\ \ a_2...a_n}:=\sg_M^{a_1b}\mfT_{ba_2...a_n}
\end{align*}
\item the contraction operators $\tilde{\cdot}$ and $\slashed{\cdot}$ are defined by
\[(\mfT\,\tilde{\cdot} \,\mfT')_{a_{m+1}...a_n}:=\mfT_{a_1...a_ma_{m+1}...a_n}\mfT'_{b_1...b_m}\qg_M^{a_1b_1}...\qg_M^{a_mb_m},\quad(\mfT\,\slashed{\cdot} \,\mfT')_{a_{m+1}...a_n}:=\mfT_{a_1...a_ma_{m+1}...a_n}\mfT'_{b_1...b_m}\qg_M^{a_1b_1}...\qg_M^{a_mb_m}\]
\item the norm $|.|_\sg$ operator are defined by
\begin{align*}
|\mfT|_{\sg_M}^2:=\mfT\,\slashed{\cdot}\,\mfT
\end{align*}
\item the trace operators $\qtr$ and $\str$ are defined by
\begin{align*}
\qtr \mfT:=\qg_M\,\tilde{\cdot}\, \mfT,\qquad\qtr \mfT:=\sg_M\,\slashed{\cdot}\, \mfT
\end{align*}
\item the symmetrised product operator (on 1-forms) $\otimes$ is defined by
\begin{align*}
\big(\mft\otimes \mft'\big)_{ab}:=\mft_a\mft'_b+\mft_b\mft'_a
\end{align*}
\item the traceless symmetrised product operators $\qotimeshat$ and $\sotimeshat$ are defined by
\begin{align*}
\big(\mft\qotimeshat \mft'\big)_{ab}:=\mft_a\mft'_b+\mft_b\mft'_a-\big({\qg_M}\big)_{ab}\mft^c\mft'_c,\qquad\big(\mft\sotimeshat \mft'\big)_{ab}:=\mft_a\mft'_b+\mft_b\mft'_a-\big({\sg_M}\big)_{ab}\mft^c\mft'_c
\end{align*}
\item the Hodge star operators $\qhd$ and $\shd$ are defined by
\begin{align*}
\tilde{\star} \mft:=\tilde{\epsilon}\qdot\mft,\qquad\slashed{\star} \mft:=\slashed{\epsilon}\sdot\mft
\end{align*}
\end{itemize}

Next are the differential operators. 

In what follows, $\qD$ and $\sD$ are the unique derivative operators on $\mcalm$ such that $\qD\qg_M=\sD\sg_M=0$. Moreover, $\mfp$ denotes a $p$-form on $\mcalm$ and $V, V_1,...,V_p$ denote smooth vector fields on $\mcalm$.
\begin{itemize}
	\item the derivative operators $\qn$ and $\sn$ are defined by
	\begin{align*}
	\qn_V\mfT:=\qD_{\qV}\mfT,\qquad\sn_V \mfT:=\sD_{\sV}\mfT
	\end{align*}
	\item the exterior derivative operators $\qexd$ and $\sexd$ are defined by
	\begin{align*}
	\big(\qexd\mfp\big)\big(V_1,...,V_p\big):=\big(\mathrm{d} \mfp\big)\big(\qV_1,..., \qV_p\big),\qquad\big(\sexd\mfp\big)\big(V_1,...,V_p\big):=\big(\mathrm{d} \mfp\big)\big(\sV_1,..., \sV_p\big)
	\end{align*}
	\item the divergence operators $\qd$ and $\sdiv$ are defined by
	\[\tilde{\delta} \mfT:=-\qtr\big(\qn\mfT\big),\qquad\sdiv\mfT:=\str\big(\sn\mfT\big)\]
	\item the wave operator $\qbox$ and the Laplace operator $\slap$ are defined by
	\begin{align*}
	\qbox \mfT:=-\qd\big(\qn \mfT\big),\qquad
	\slap\mfT:=\sdiv\big(\sn
	\mfT\big)
	\end{align*}
	\item the curl operators $(\qhd\qexd)$ and $\scurl$ are defined by
	\begin{align*}
	\tilde{\star}\tilde{\text{d}} \mft:=\frac{1}{2}\tilde{\epsilon}\,\tilde{\cdot}\,\qexd \mft,\qquad\scurl\mft:=\frac{1}{2}\slashed{\epsilon}\,\slashed{\cdot}\,\sexd \mft
	\end{align*}
\end{itemize}
Lastly, we end the section by introducing the following notation:
\begin{itemize}
	\item $\qn\otimes \mft$ and $\sn\otimes \mft$ denote the Lie derivatives $\mathcal{L}_{\mft^\sharp}\qg_M$ and $\mathcal{L}_{\mft^\sharp}\sg_M$ respectively
	\item $\qn\otimeshat \mft$ and $\sn\otimeshat \mft$ denote the traceless Lie derivatives $\mathcal{L}_{\mft^{\sharp}}\qg_M+\qg_M\,\tilde{\cdot}\,\tilde{\delta} \mft$ and $\mathcal{L}_{\mft^{\sharp}}\sg_M-\sg_M\cdot\sdiv\mft$ respectively
	\item for an $n$-tuple of vector fields $V_1,...,V_n$ we denote by $\mfT_{V_1...V_n}$ the contraction $(((\mfT\cdot V_1)\cdot V_2)...)\cdot V_n$
\end{itemize}

\subsection{The Cauchy hypersurface $\Sigma$}\label{TheCauchyhypersurfaceSigma}

In this section we consider a foliation of $\Mgs$ by Cauchy hypersurfaces $\Sigma_{t^*}$, thus identifying an initial Cauchy hypersurface $\Sigma$. In addition, a restricted version of the $2+2$ formalism to the hypersurface $\Sigma$ is detailed.

Initial data for the equations of linearised gravity will be prescribed on $\Sigma$ in section \ref{Pointwisestrongasymptoticflatness} with the aid of this restricted formalism.

\subsubsection{The Cauchy hypersurface $\Sigma$}\label{TheCauchyhypersurfaceSigmaandsmitensoranalysis}

We define the manifolds with boundary
\begin{align*}
\Sigma_{t^*}:=\{t^*\}\times[2M,\infty)\times S^2.
\end{align*}
As the gradient of $t^*$ is globally time-like on $\mcalm$, the family $\Sigma_{t^*}$ describe a foliation of $\mcalm$ by Cauchy hypersurfaces. We henceforth fix an initial time $t^*_0$:
\begin{align*}
\Sigma:=\Sigma_{t^*_0}.
\end{align*}
The initial hypersurface $\Sigma$ comes equipped with the Riemannian metric
\begin{align*}
h_M=\opmu\ud r^2+r^2\mathring{g}
\end{align*}
along with the associated second fundamental form
\begin{align*}
k:=\frac{1}{2}\mathcal{L}_n h_M=\frac{1}{2}\frac{1}{\bh}\frac{\mu}{r}(2+\mu)\ud r^2-\frac{\mu}{\bh}r\mathring{g}.
\end{align*}
Here, $n$ is the future-pointing unit normal to $\Sigma$
\begin{align*}
n=\bh\,\pt-\frac{\mu}{\bh}\pr
\end{align*}
where we have defined the \emph{lapse} function
\begin{align*}
\bh:=\sqrt{1+\mu}.
\end{align*}

\subsubsection{$\smi$-tensor analysis}\label{Smitensoranalysis}

We consider now a foliation of $\Sigma$ by the 2-spheres $S^2_r$ given as the level sets of the areal function $r$. Associated to this foliation is the inward pointing unit normal
\begin{align*}
\nu=\frac{1}{\bh}\pr.
\end{align*}
This leads to the following definition.
\begin{definition}\label{defnsmitensoranalysis}
	Let $\{\nu, e_1, e_2\}$ be a frame on $\Sigma$. Then we say that a smooth $n$-covariant tensor field $H$ on $\Sigma$ is an $n$-covariant $S_\nu$-tensor field if
	\begin{align*}
	H(\cdot,...,\nu,...,\cdot)=0.
	\end{align*}
\end{definition}
Given a symmetric 2-covariant tensor field $H$ on $\Sigma$ one projects it onto the function $\bar{H}$ on $\Sigma$, the $\smi$ 1-form $\stkout{H}$ and the symmetric 2-covariant $\smi$-tensor field as follows:
\begin{align}
\bar{H}&:=H(\nu,\nu),\label{n1}\\
\text{\sout{\ensuremath{H}}}(e_I)&:=H(\nu,e_I),\label{n2}\\
\slashed{H}(e_I,e_J)&:=H(e_I,e_J)\label{n3}.
\end{align}
It is natural to decompose the latter into its trace and trace-free parts with respect to the induced metric $\sg_M$ on $S^2_r$:
\begin{align*}
\slashed{H}=\hat{\slashed{H}}+\frac{1}{2}\sg_M\cdot\str\slashed{H}.
\end{align*}

A particularly useful application of this procedure is for the second fundamental form $k$. 

First one decomposes $k$ into its trace and tracefree parts with respect to $h_M$:
\begin{align*}
k=\hat{k}+\frac{1}{3}h\cdot\trk,\qquad\trk=-\frac{1}{2}\frac{\mu}{\bh}\frac{1}{r}\frac{2+3\mu}{1+\mu}.
\end{align*}
Decomposing $\hat{k}$ according to \eqref{n1}-\eqref{n3} then yields (supressing the hat notation):
\begin{align*}
\bk&=\frac{1}{3}\frac{\mu}{\bh}\frac{1}{r}\frac{4+3\mu}{1+\mu},\\
\str\slashed{k}&=-\bk,\\
\text{\sout{\ensuremath{k}}}=\hat{\slashed{k}}&=0.
\end{align*}

Lastly, one has a natural calculus on $\smi$ tensor fields induced from $(\Sigma,h_M)$. 

Indeed, the `angular' operations introduced in section \ref{Tensoranalysis} have analagous definitions for $\smi$-tensor fields. In addition, for smooth $\smi$-tensor fields we define the differential operator $\snnu$ as corresponding to the action of covariant differentiation (with respect to $h_M$) in the direction $\nu$:
\begin{align*}
\snnu\bar{H}&=\mcall_\nu\bar{H},\\
\snnu\text{\sout{\ensuremath{H}}}&=\mcall_\nu\text{\sout{\ensuremath{H}}}-\frac{1}{\bh}\frac{1}{r}\text{\sout{\ensuremath{H}}},\\
\snnu\slashed{H}&=\mcall_\nu\slashed{H}-\frac{1}{\bh}\frac{2}{r}\slashed{H}.
\end{align*}

\subsection{The $\taus$-foliation}\label{Thetaustarfoliation}

In this section we define a foliation of $D^+(\Sigma)$ by hypersurfaces which coincides with the foliation of $D^+(\Sigma)$ by the hypersurfaces $\Sigma_{t^*}$ in a compact region of spacetime and with a foliation of $D^+(\Sigma)$ by null-hypersurfaces in the non-compact region. A $\qm$-frame\footnote{A frame for $\qm$-vector fields on $\mcalm$.} which is adapted to this foliation is then defined.

It is through this foliation that we will capture the dispersive properties of solutions to the equations of linearised gravity. Moreover, the introduction of the frame will prove useful in the later analysis.

\subsubsection{A foliation of $D^+(\Sigma)$ that `penetrates' both $\hplus$ and $\iplus$}\label{AfoliationofMthatpenetratesbothhplusandiplus}

To define this foliation we first introduce the following null geodesic vector field $L$ on $D^+(\Sigma)$:
\begin{align}
\nabla_LL&=0,\label{eqndefiningL1}\\
L\big|_{\Sigma}&=\bh\big(n+\nu\big)\label{eqndefiningL2}.
\end{align}
Here, $\nabla$ is the Levi-Civita connection of $g_M$. This determines an optical function $U$ according to
\begin{align*}
LU&=0,\\
U\big|_\Sigma&=t^*_0-r
\end{align*}
along with the rescaled optical function
\begin{align*}
u:=-\log(U-t^*_0+2M)+t^*_0+\log(2M-R).
\end{align*}
Defining thus the function 
\begin{align*}
\tau^\star(t^*,r,\theta^A)=\begin{cases} 
u(t^*,R,\theta^A) & r\leq R \\
u(t^*,r,\theta^A) & r\geq R
\end{cases}
\end{align*}
the desired foliation arises as the union of the level sets $\Xi_{\taus}$ of the function $\taus$:
\begin{align*}
D^+(\Sigma)=\bigcup_{\taus\in\mathbb{R}}\Xi_{\taus}.
\end{align*}
Here, $\theta^A$ is any coordinate chart on $S^2$ and $R>>10M$ is a fixed constant. 

We shall informally refer to the limiting sphere as $r\rightarrow\infty$ along each hypersurface $\Xi_{\taus}$ as a sphere on \emph{future null infinity $\iplus$}. The function $\taus$ thus serves to parametrise $\iplus$. Moreover, we observe that $\taus$ is an increasing function of $t^*$ with $\taus\big|_{\Sigma\cap\{r\leq R\}}=t_0^*$. This yields the Penrose diagram of Figure 3.

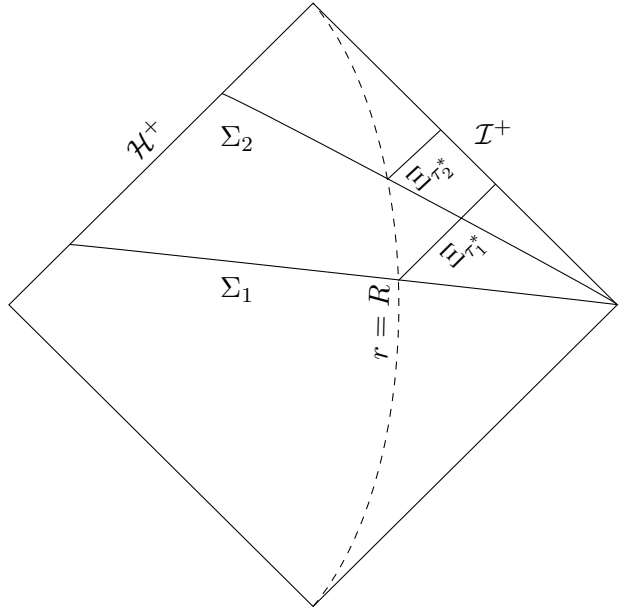
\begin{figure}[!h]
	\centering
	\begin{tikzpicture}
	\node (I)    at ( 4,0)   {};
	\path 
	(I) +(90:4)  coordinate (Itop)
	+(-90:4) coordinate (Ibot)
	+(180:4) coordinate (Ileft)
	+(0:4)   coordinate (Iright)
	;
	\draw  (Ileft) -- (Itop)node[midway, above, sloped] {$\hplus$} -- (Iright) node[midway, above right]    {$\cal{I}^+$} -- (Ibot)node[midway, below right]    {} -- (Ileft)node[midway, below, sloped] {} -- cycle;
	
	\draw[dashed]  
	(Ibot) to[out=50, in=-50, looseness=0.75] node[pos=0.475, above, sloped] {$r=R$        }($(Itop)!.5!(Itop)$) ;
	
	\draw 
	($(Itop)!.8!(Ileft)$) to[out=0, in=0, looseness=0.05] ($(Iright)!.5!(Iright)$);
	
	\draw 
	($(Itop)!.3!(Ileft)$) to[out=0, in=0, looseness=0.05] ($(Iright)!.5!(Iright)$);
	
	\draw 
	($(Itop)!.6!(Iright)$) --node[midway, below, sloped]{$\Xi _{\tau_1 ^*}$} (5.12, 0.32);
	
	\draw 
	($(Itop)!.42!(Iright)$) --node[right, below, sloped]{$\Xi _{\tau_2 ^*} $} (4.975, 1.66);
	
	\node [below] at (3,0.5) {$\Sigma_1$};
	
	\node [below] at (3,2.5) {$\Sigma_2$};
	
	\end{tikzpicture}
	\caption{A Penrose diagram of $\Mgs$ depicting the hypersurfaces $\Sigma_{t^*}$ and $\Xi_{\taus}$.} 
	\label{Figure1}
\end{figure}
We further present in Figure 4 a Penrose diagram depicting the spacetime region given as the future of the (initial) hypersurface $\Xi_{\taus_0}$.

\begin{figure}[!h]
	\centering
	\begin{tikzpicture}
	\node (I)    at ( 4,0)   {};
	\path 
	(I) +(90:4)  coordinate (Itop)
	+(-90:4) coordinate (Ibot)
	+(180:4) coordinate (Ileft)
	+(0:4)   coordinate (Iright)
	;
	\draw  (Ileft) -- (Itop)node[midway, above, sloped] {$\hplus$} -- (Iright) node[midway, above right]    {$\cal{I}^+$} -- (Ibot)node[midway, below right]    {} -- (Ileft)node[midway, below, sloped] {} -- cycle;
	
	\draw 
	($(Itop)!.9!(Ileft)$) to[out=0, in=0, looseness=0.05] node[midway, below, sloped]{$\Sigma$}($(Iright)!.5!(Iright)$);
	
	\draw 
	($(Itop)!.6!(Iright)$) --node[midway, below, sloped]{$\Xi_{\tau_0 ^*}$} (5.1, 0.15);
	
	\draw 
	($(Itop)!.8!(Ileft)$) to[out=0, in=0, looseness=0.05] (2,.325);
	
	\draw 
	($(Itop)!.7!(Ileft)$) to[out=0, in=0, looseness=0.05] (3,0.275);
	
	\draw 
	($(Itop)!.6!(Ileft)$) to[out=0, in=0, looseness=0.05] (3.9,0.25);
	
	\draw 
	($(Itop)!.5!(Ileft)$) to[out=0, in=0, looseness=0.05] (5.13,0.15);
	
	\draw 
	($(Itop)!.4!(Ileft)$) to[out=0, in=0, looseness=0.05] (5.5,0.6);
	
	\draw 
	($(Itop)!.3!(Ileft)$) to[out=0, in=0, looseness=0.05] (5.8,1); 
	
	\draw 
	($(Itop)!.2!(Ileft)$) to[out=0, in=0, looseness=0.05] (6.25,1.5);
	
	\draw 
	($(Itop)!.1!(Ileft)$) to[out=0, in=0, looseness=0.05] (5.5,2.5);
	
	\end{tikzpicture}
	\caption{A Penrose diagram of $\Mgs$ depicting the causal future of $\Xi_{\taus_0}$.}
	\label{Figure2}
\end{figure}
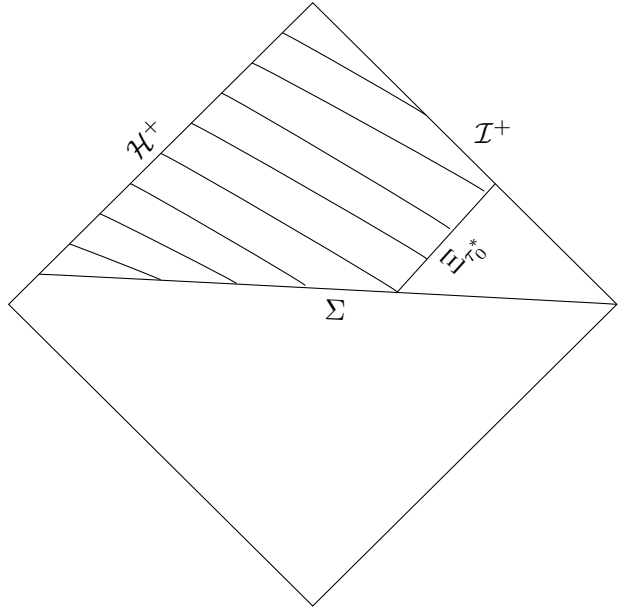

\subsubsection{The $\qm$-frame $\{S, L\}$}\label{TheqmframeSL}

We now introduce a $\qm$-frame on $D^+(\Sigma)$ adapted to the \emph{$\taus$-foliation} of the previous section.

In what follows, $S$ is the inward pointing (towards $\iplus$) unit normal to the 2-sphere given as the level set of the area radius function $r$ in the hypersurface $\Xi_{\taus}$. Moreover, we define the smooth function $\kappa=1-\frac{2M}{t^*_0-U}$ on $D^+(\Sigma)$ noting\footnote{To show this one derives from \eqref{derivativesofkappaandrinSLframe} the propagation equation $\qn_L\frac{\kappa}{1-\mu}=-\frac{1}{2M}\frac{\kappa^2}{\ommu^2}$ along $\hplus$ and then one observes that on $\hplus$ $L=e^\frac{t^*}{4M}T$.} that on $\hplus$
\begin{align*}
\frac{\kappa}{1-\mu}=\Big(1+2(e^{\frac{1}{4M}t^*}-e^{\frac{1}{4M}t^*_0})\Big)^{-1}.
\end{align*}
\begin{proposition}\label{propSLframe}
The vector fields $S$ and $L$ are linearly independent $\qm$-vector fields on the spacetime region $D^+(\Sigma)$. In particular, the set $\{S,L\}$ determines a (smooth, spherically symmetric) $\qm$-frame on $D^+(\Sigma)$ with
\begin{align}\label{qmlengthsofframe}
\qg_M(S,S)=1,\qquad\qg_M(S,L)=\bh\frac{\kappa}{1-\mu},\qquad\qg_M(L,L)=0
\end{align}
and the (smooth) connection coefficients
\begin{align}
\qn_SS&=-\frac{1}{2}\frac{1}{\bh^3}\frac{\mu}{r}S+\frac{1}{2}\frac{1}{\bh^2}\frac{\mu}{r}\frac{2+\mu}{1+\mu}\frac{1-\mu}{\kappa}L,\qquad\qquad\quad\,\,\qn_S L=-\frac{\bh}{2}\frac{1}{1-\mu}\bigg(\frac{\mu}{r}-\frac{1}{M}\frac{(1-\kappa)^2}{2-\kappa}\bigg) L,\label{connectioncoefficeintsinqmframe1}\\
\qn_LS&=\frac{1}{2}\frac{1}{1+\mu}\frac{\mu}{r}\bigg(1-\frac{\kappa}{1-\mu}\bigg)S-\frac{1}{2}\frac{1}{\bh}\frac{3+\mu}{1+\mu}\frac{\mu}{r} L,\qquad\qn_LL=0\label{connectioncoefficeintsinqmframe2}.
\end{align}
Moreover, 
\begin{align}\label{derivativesofkappaandrinSLframe}
\qn_S\kappa=\frac{1}{2}\frac{\bh}{M}\frac{(1-\kappa)^2}{2-\kappa}\frac{\kappa}{1-\mu},\qquad \qn_Lr=\kappa.
\end{align}
\end{proposition}
\begin{proof}
The first statement follows from spherical symmetry of $\Mgs$ and uniqueness of solutions to \eqref{eqndefiningL1}-\eqref{eqndefiningL2}. This immediately gives the first and third parts of \eqref{qmlengthsofframe} from the metric \eqref{schwarzschildmetric}.

To prove the remaining parts of \eqref{qmlengthsofframe}-\eqref{connectioncoefficeintsinqmframe2} we follow Christodoulou in \cite{Christ}. 

Indeed, first we introduce the null second fundamental form $\chi$:
\begin{align}
2\chi:&=\mcall_L\sg_M,\nonumber\\
&=\frac{2}{r}\qn_Lr\sg_M\label{eqnfortrchi}.
\end{align}
Decomposing $\chi$ into its trace and trace-free parts
\begin{align*}
\chi=\chihat+\frac{1}{2}\sg_M\cdot\trchi
\end{align*}
one has by the Raychauduhri equation ((3.1c) of \cite{Christ})
\begin{align*}
\qn_L\trchi+\frac{1}{2}(\trchi)^2=-\chihat\sdot\chihat.
\end{align*}
However, since by spherical symmetry $\chihat=0$ (see Proposition \ref{propsmtwotensorsphericalharmonicdecomposition}) it follows that
\begin{align*}
\qn_L^2r=0.
\end{align*}
Hence by the initial condition \eqref{eqndefiningL2}
\begin{align*}
\qn_Lr=\kappa.
\end{align*}
In particular, as $\qn_Tr=0$, $\qn_Sr=\bh^{-1}$ and $L$ is future-directed\footnote{As $u$ is an increasing function with respect to $t^*$.} one has
\begin{align*}
T=\frac{1-\mu}{1+\mu}\Big(\kappa^{-1}L-\bh\,S\Big)
\end{align*}
thus yielding the first part of \eqref{connectioncoefficeintsinqmframe1} and moreover completing \eqref{qmlengthsofframe}. 

To continue we introduce the conjugate, future-directed null vector to $L$:
\begin{align}\label{LbarinSLframe}
\uL:=-\frac{2}{\bh}\frac{1-\mu}{\kappa}\, S+\frac{1}{1+\mu}\frac{\ommu^2}{\kappa^2} L.
\end{align}
As $\uL$ is the unique such null vector with $\qg_M(\uL, L)=-2$ it follows from spherical symmetry and (2.3)-(2.4), (4.1g) and (4.3) of \cite{Christ} that
\begin{align}\label{connectioncoefiicentsinnullframe}
\qn_L\uL=0,\qquad\qn_{\uL}L=-2\uomega\,L.
\end{align}
Here, the smooth function $\uomega$ on $D^+(\Sigma)$ satisfies the propagation equation
\begin{align*}
\qn_L\uomega=\frac{1}{r}\frac{\mu}{r}.
\end{align*}
Consequently, as $L-\opmu\uL=2\bh\,S$ it follows from \eqref{connectioncoefficeintsinqmframe1} and \eqref{connectioncoefiicentsinnullframe} that on $\Sigma$
\begin{align*}
\uomega=\frac{1}{2}\frac{1}{1+\mu}\frac{\mu}{r}
\end{align*}
and so
\begin{align*}
\uomega=-\frac{1}{2}\frac{\mu}{r}\frac{1}{\kappa}+\frac{1}{2}\frac{1}{M}\frac{(1-\kappa)^2}{2-\kappa}\frac{1}{\kappa}
\end{align*}
thus completing \eqref{connectioncoefficeintsinqmframe1}-\eqref{connectioncoefficeintsinqmframe2}.

Finally, to derive the formulae of \eqref{derivativesofkappaandrinSLframe} one uses \eqref{LbarinSLframe} along with the fact that $\qn_{\uL}\uL=2\uomega\, \uL$.
\end{proof}

Note that the last conclusion of Proposition \ref{propSLframe} in particular allows one to express all operations introduced in section \ref{Tensoranalysis} with respect to the frame $\{S, L\}$.

The introduction of the frame also allows for a pointwise norm acting on tensor fields  
to be defined.

\begin{definition}\label{defnqmnorm}
We define the (pointwise) norm $\mnorm{\cdot}$ acting on $\qm$-tensor fields, $\qmsm$ 1-forms and $\sm$-tensor tensor fields
as follows.

If $Q,q$ and $f$ are respectively a symmetric $2$-covariant $\qm$-tensor, a $\qm$ 1-form and a function on $D^+(\Sigma)$ then
\begin{align*}
\mnorm{Q}^2:&=2|Q_{SS}Q_{LL}|+|Q_{SL}|^2+2|Q_{SL}Q_{LL}|+|Q_{LL}|^2,\\
\mnorm{q}^2:&=2|q_{S}q_{L}|+|q_{L}|^2,\\
\mnorm{f}^2:&=|f|^2.
\end{align*}
Conversely, if $\Theta$ is an $n$-covariant $\sm$-tensor then
\begin{align*}
\mnorm{\Theta}^2:=|\Theta|_\sg^2.
\end{align*}
Finally, if $\momega$ is a $\qmsm$ 1-form then
\begin{align*}
\mnorm{\momega}^2:=|\momega_S\momega_L|_\sg+|\momega_L|_\sg^2.
\end{align*}
\end{definition}

\subsection{The projection of tensor fields on $\mcalm$ onto and away from the $l=0,1$ spherical harmonics}\label{TheprojectionoftensorfieldsonMontoandawayfromthel=0,1sphericalharmonics}

In this section we provide a notion of tensor fields on $\mcalm$ having support on and outside of the $l=0,1$ spherical harmonics and then define a projection map onto each respective space.

In section \ref{GaugenormalisedsolutionsandidentificationoftheKerrparameters} we will establish an identification between the projection of a solution to the equations of linearised gravity onto $l=0,1$ and stationary solutions to said equations.\newline

This section follows closely section 4.4 of \cite{D--H--R}.

\subsubsection{The $l=0,1$ spherical harmonics and the spherical harmonic decomposition}\label{Thel=0,1sphericalharmonicsandthesphericalharmonicdecomposition}

We recall the well known spherical harmonics $Y^l_m$ with $l\in\mathbb{N}$ and $m\in\{-l,...,0,...l\}$. We explicitly note the form of the $l=0$ and $l=1$ modes
\begin{align}
Y^{l=0}_{m=0}&=\frac{1}{\sqrt{4\pi}}\label{l=0sphericalharmonics}\\
Y^{l=1}_{m=-1}=\sqrt{\frac{3}{4\pi}}\sin\theta\cos\varphi,\qquad Y^{l=1}_{m=0}&=\sqrt{\frac{3}{8\pi}}\cos\theta,\qquad Y^{l=1}_{m=1}=\sqrt{\frac{3}{4\pi}}\sin\theta\sin\varphi.\label{l=1sphericalharmonics}
\end{align}

\begin{definition}\label{defnfunctionsupportedl=0,1}
Let $f$ be a smooth function on $S^2$. 

Then we say that $f$ is supported only on $l=0,1$ iff for $l\geq 2$
\begin{align*}
\int_{S^2}f\cdot Y^l_m\ud \mathring{\epsilon}=0
\end{align*}
where $\mathring{\epsilon}$ is the volume form associated to the unit metric on the round sphere.

Conversely, we say that $f$ has vanishing projection to $l=0,1$ iff for $l=0,1$
\begin{align*}
\int_{S^2}f\cdot Y^l_m\ud \mathring{\epsilon}=0.
\end{align*}
\end{definition}
This leads to the classical spherical harmonic decomposition of square integrable functions on $S^2$.
\begin{proposition}\label{propsphericalharmonicdecomp}
Let $f\in L^2(S^2)$. Then one has the unique, orthogonal decomposition
\begin{align*}
f=\mpart{f}+f'
\end{align*}
where the function $\mpart{f}\in L^2(S^2)$ is supported only on $l=0,1$ and the function $f'\in L^2(S^2)$ has vanishing projection to $l=0,1$.
\end{proposition}

The remainder of this section is concerned with generalising this result to tensor fields on $\mcalm$.

\subsubsection{The $l=0,1$ spherical harmonics and $\qm$-tensors}\label{Thel=0,1spheriacalharmonicsandqmtensors}

We begin with $\qm$-tensors.

In what follows $n\geq 0$ is an integer and we denote by $S^2_{\taus, r}$ the 2-sphere given as the intersection of the level sets of the functions $\taus$ and $r$ on $\mcalm$.

\begin{definition}\label{defnqmtensorssupportedonlgeq2}
Let $Q$ be a smooth $n$-covariant $\qm$-tensor.

Then we say that $Q$ is supported only on $l=0,1$ iff for $l\geq 2$ all components of $Q$ in the frame $\{S,L\}$ satisfy
	\begin{align*}
	\int _{S^2_{\taus, r}}Q_{\cdot\cdot\cdot S\cdot\cdot\cdot L\cdot\cdot\cdot}\cdot Y^l_m\ud\mathring{\epsilon}=0.
	\end{align*}
	
	Conversely, we say that $Q$ has vanishing projection to $l=0,1$ iff for $l=0,1$ all components of $Q$ in the frame $\{S,L\}$ satisfy
\begin{align*}
\int _{S^2_{\taus, r}}Q_{\cdot\cdot\cdot S\cdot\cdot\cdot L\cdot\cdot\cdot}\cdot Y^l_m\ud\mathring{\epsilon}=0.
\end{align*}
\end{definition}

We shall denote by $\LM$ the space of $\qm$-tensor fields that have vanishing projection to $l=0,1$.

Subsequently, applying Proposition \ref{propsphericalharmonicdecomp} to each frame component yields the following.
\begin{proposition}\label{propsphericalharmonicdecompositionofqmtensors}
	Let $Q$ be a smooth $n$-covariant $\qm$-tensor field. Then one has the unique, orthogonal decomposition
	\begin{align*}
	Q=\mpart{Q}+Q'
	\end{align*}
	where the smooth $n$-covariant $\qm$-tensor $\mpart{Q}$ is supported only on $l=0,1$ and $Q'\in\LM$.
\end{proposition}

\subsubsection{The $l=0,1$ spherical harmonics and $\qmsm$ 1-forms}\label{Thel=0,1sphericalharmonicsandqmsm1forms}

Next we consider $\qmsm$ 1-forms. 

First we recall the classical Hodge decomposition of smooth 1-forms on $S^2$
\begin{align*}
\xi=\sdso\big(\xi_{\text{e}}, \xi_{\text{o}}\big).
\end{align*}
Here, $\xi_{\text{e}}$ and $\xi_{\text{o}}$ are smooth functions on $S^2$. This decomposition is unique if one assumes that the smooth functions $\xi_{\text{e}}$ and $\xi_{\text{o}}$ have vanishing mean over $S^2$ and can be readily extended to smooth $\qmsm$ 1-forms by using the frame $\{S,L\}$. This leads to the following definition, an analogue of which for smooth $\sm$ 1-forms was first given in \cite{D--H--R}.

\begin{definition}\label{defnqmsmoneformssupportedonlgeqi}
Let $\momega$ be a smooth $\qmsm$ 1-form.

Then we say that $\momega$ is supported only on $l=0,1$ iff the two smooth $\qm$ 1-forms $\momega_{\textnormal{e}}$ and $\momega_{\textnormal{o}}$ in the unique representation
	\begin{align*}
	\momega=\sdso\big(\momega_{\textnormal{e}}, \momega_{\textnormal{o}}\big)
	\end{align*}
	are supported only on $l=0,1$ and moreover have vanishing mean on every 2-sphere $S^2_{\taus, r}$.
	
	Conversely, we say that $\momega$ has vanishing projection to $l=0,1$ iff $\momega_{\textnormal{e}},\momega_{\textnormal{o}}\in\LM$.
\end{definition}
Linearity and Proposition \ref{propsphericalharmonicdecompositionofqmtensors} then yields the desired decomposition of smooth $\qmsm$ 1-forms.
\begin{proposition}\label{propqmoneformsmoneformsphericalharmonicdecomposition}
	Let $\stkout{\omega}$ be a smooth $\qm\otimes\sm$ 1-form. Then one has the unique, orthogonal decomposition
	\begin{align*}
	\text{\sout{\ensuremath{\omega}}}=\dpart{\momega}+\momega'
	\end{align*}
	where the smooth $\qmsm$ 1-form $\momega$ is supported only on $l=0,1$ and $\momega'$ has vanishing projection to $l=0,1$.
\end{proposition}

\subsubsection{The $l=0,1$ spherical harmonics and symmetric, traceless 2-covariant $\sm$-tensors} \label{Thel=0,1sphericalharmonicsandsm2tensors}

Now we consider symmetric, traceless 2-covariant $\sm$-tensors. 

One has the following proposition, proved in \cite{D--H--R}, which establishes the sense in which all smooth, symmetric, traceless, 2-covariant $S$-tensors have vanishing projection to $l=0,1$.
\begin{proposition}\label{propsmtwotensorsphericalharmonicdecomposition}
	Let $\theta$ be a smooth, symmetric, traceless 2-covariant $\sm$-tensor. Then $\theta$ can be uniquely represented as
	\begin{align*}
	\theta=\sdst\sdso\big(\theta_{\textnormal{e}}, \theta_{\textnormal{o}}\big)
	\end{align*}
	where $\theta_{\textnormal{e}}, \theta_{\textnormal{o}}\in\LM$ are smooth functions. 
\end{proposition}

\subsubsection{The $l=0,1$ spherical harmonics and $\smi$-tensors}\label{Thel=01sphericalharmonicsandqmsmitensors}

Our final considerations are $\smi$-tensors.

Analagously to the previous three sections one has the following.

In the sequel, we denote by $\LS$ the space of smooth functions on $\Sigma$ having vanishing projection to $l=0,1$.
\begin{proposition}\label{propsphericalharmonicdecompositionofsmitensors}
	Let $f$ be a smooth function on $\Sigma$. Then one has the unique, orthogonal decomposition
	\begin{align*}
	f=\mpart{f}+f'
	\end{align*}
	where $\mpart{f}$ is a smooth function on $\Sigma$ supported only on $l=0,1$ and $f'\in\LS$ is smooth.
	
	Moreover, if $\xi$ is a smooth $\smi$ 1-form then one has the unique, orthogonal decomposition
	\begin{align*}
	\xi=\dpart{\xi}+\xi'
	\end{align*}
	where $\dpart{\xi}$ is a smooth $\smi$ 1-form supported only on $l=0,1$ and $\xi'$ is a smooth $\sm$ 1-form with vanishing projection to $l=0,1$ (cf. Definition \ref{defnqmsmoneformssupportedonlgeqi}).
	
	Finally, if $\theta$ is a smooth, symmetric, traceless 2-covariant $\smi$-tensor then $\theta$ can be uniquely represented as
	\begin{align*}
	\theta=\sdst\sdso\big(\theta_{\textnormal{e}}, \theta_{\textnormal{o}}\big)
	\end{align*}
	where $\theta_{\textnormal{e}}, \theta_{\textnormal{o}}\in\LS$ are smooth. 
\end{proposition}

\subsubsection{The projection onto and away from $l=0,1$}\label{Theprojectionontoandwayfroml=0,1}

We end this section by defining a projection mapping onto and away from the space of tensor fields supported only on $l=0,1$.

In what follows, $n\geq0$ is an integer.
\begin{definition}\label{defnprojectionmap}
Let $Q$ be a smooth $n$-covariant $\qm$-tensor field with $\momega$ a smooth $\qmsm$ 1-form and let $f$ be a smooth function on $\Sigma$ with $\xi$ a smooth $\smi$ 1-form.

Then we respectively call the maps
\begin{align*}
Q&\rightarrow\mpart{Q},\\
\momega&\rightarrow\mpart{\momega}
\end{align*}
and
\begin{align*}
f&\rightarrow\mpart{f},\\
\xi&\rightarrow\mpart{\xi}
\end{align*}
the projection of $Q, \momega$ and $f, \xi$ onto $l=0,1$.

Conversely, we respectively call the maps
\begin{align*}
Q&\rightarrow Q',\\
\momega&\rightarrow\momega'
\end{align*}
\begin{align*}
Q&\rightarrow Q',\\
\momega&\rightarrow\momega'
\end{align*}
the projection of $Q, \momega$ and $f, \xi$ away from $l=0,1$.
\end{definition}

\subsection{Elliptic estimates on 2-spheres}\label{Ellipticoperatorson2-spheres}

In this section we introduce various elliptic operators on 2-spheres for which elliptic estimates are derived.

These operators will play an important role throughout the paper.

\subsubsection{Norms on spheres}\label{Normsonspheres}

First we define norms on 2-spheres through which the elliptic estimates are to be measured. Importantly, these norms shall respect the $\taus$-foliation of section \ref{Thetaustarfoliation}.

\begin{definition}\label{defninnerproduct}
Let $m\geq 0$ be an integer. Then we define the norm $||\cdot||_{H^m_{\taus, r}}$ acting on $\qm$-tensor fields, $\qmsm$ 1-forms and $\sm$-tensor tensor fields as follows.

If $Q,q$ and $f$ are respectively a symmetric $2$-covariant $\qm$-tensor, a $\qm$ 1-form and a function on $D^+(\Sigma)$ then
\begin{align*}
||Q||_{H^m_{\taus, r}}^2:&=\frac{1}{r^2}\int_{S^2_{\taus,r}}\sum_{i=0}^{m}\bigg(2|(r\sn)^iQ_{SS}\sdot(r\sn)^iQ_{LL}|+|(r\sn)^iQ_{SL}|_\sg^2+2|(r\sn)^iQ_{SL}\sdot(r\sn)^iQ_{LL}|+|(r\sn)^iQ_{LL}|_\sg^2\bigg)\sepsilon,\\
||q||_{H^m_{\taus, r}}^2:&=\frac{1}{r^2}\int_{S^2_{\taus,r}}\sum_{i=0}^{m}\bigg(2|(r\sn)^iq_{S}\sdot(r\sn)^iq_{L}|+|(r\sn)^iq_{L}|_\sg^2\bigg)\sepsilon,\\
||f||_{H^m_{\taus, r}}^2:&=\frac{1}{r^2}\int_{S^2_{\taus,r}}\sum_{i=0}^{m}|(r\sn)^if|_\sg^2\sepsilon.
\end{align*}
Conversely, if $\Theta$ is an $n$-covariant $\sm$-tensor then
\begin{align*}
||\Theta||_{H^m_{\taus, r}}^2:&=\frac{1}{r^2}\int_{S^2_{\taus,r}}\sum_{i=0}^{m}|(r\sn)^i\Theta|_\sg^2\sepsilon.
\end{align*}
Finally, if $\momega$ is a $\qmsm$ 1-form then
\begin{align*}
||\momega||_{H^m_{\taus, r}}^2:&=\frac{1}{r^2}\int_{S^2_{\taus,r}}\sum_{i=0}^{m}\bigg(2|(r\sn)^i\momega_{S}\sdot(r\sn)^i\momega_{L}|+|(r\sn)^i\momega_{L}|_\sg^2\bigg)\sepsilon.
\end{align*}
\end{definition}

We denote by $H^m_{\taus,r}$ the space of symmetric, 2-covariant $\qm$-tensors, $\qm$ 1-forms, functions on $\mcalm$, $\qmsm$ 1-forms and symmetric, traceless 2-covariant $\sm$-tensors which have $m$ weak derivatives in the sense defined above on any 2-sphere $S^2_{\taus,r}$. For the case $m=0$ we set $H^0_{\taus,r}=L^2_{\taus,r}$. We also denote by $H^m_{\mathsmaller{\Sigma},r}$ the appropriate restriction of $H^m_{\taus,r}$ to $\Sigma$.

Equipping either of these spaces with the canonical inner product yields a Hilbert space on every 2-sphere $S^2_{\taus,r}$. In particular, the closed subspaces $H^{m,\prime}_{\taus,r}$ and $H^{m,\prime}_{\mathsmaller{\Sigma},r}$ of $H^m_{\taus,r}$ and $H^m_{\mathsmaller{\Sigma},r}$ consisting of those tensor fields on $\mcalm$ and $\Sigma$ having vanishing projection to $l=0,1$ is a Hilbert space on every 2-sphere $S^2_{\taus,r}$.

\subsubsection{The family of operators $\slashed{\mathcal{A}}$}\label{ThefamilyofoperatorssA}

We continue by introducing a family of operators on $\mcalm$ which shall ultimately serve as a shorthand notation for controlling higher order angular derivatives of tensor fields on $\mcalm$ and $\Sigma$ when measured in the norms of the previous section.

First we follow \cite{D--H--R} by defining the family of $\slashed{\mcald}$ operators through their action on tensor fields on $\mcalm$, with their restriction to tensor fields on $\Sigma$ immediate:
\begin{itemize}
	\item the operator $\sdo$ maps $\sm$ 1-forms $\xi$ to the pair of functions on $\mcalm$
	\begin{align*}
	\sdo\xi=\big(\sdiv\xi, \scurl\xi\big)
	\end{align*}
	\item the $L^2$ (with respect to $\sg_M$) adjoint $\sdso$ of $\sdo$ maps pairs of functions on $\mcalm$ into the $S$ 1-form
	\begin{align*}
	\sdso(\rho, \sigma)=-\sn\rho-\slashed{\star}\sn\sigma
	\end{align*}
	\item the operator $\sdt$ maps symmetric, traceless, 2-covariant $\sm$-tensor fields into the $\sm$ 1-form
	\begin{align*}
	\sdt\theta=\sdiv\theta
	\end{align*}
	\item the $L^2$ (with respect to $\sg_M$) adjoint $\sdst$ of $\sdt$ maps $\sm$ 1-forms $\xi$ to the symmetric, traceless, 2-covariant $\sm$-tensor
	\begin{align*}
	\sdst\xi=-\frac{1}{2}\sn\otimeshat\xi
	\end{align*}
\end{itemize}
Moreover, by utilising the frame $\{S,L\}$ one generalises the operators $\sdo$ and $\sdso$ to acting on $\qmsm$ 1-forms:
\begin{itemize}
	\item the operator $\sdo$ maps $\qmsm$ 1-forms to the pair of $\qm$ 1-forms
	\begin{align*}
	\sdo\stkout{\omega}=\big(\sdiv\stkout{\omega}, \scurl\stkout{\omega}\big)
	\end{align*}
	\item the $L^2$ (with respect to $\sg_M$) adjoint $\sdso$ of $\sdo$ maps pairs of $\qm$ 1-forms into the $\qmsm$ 1-form
	\begin{align*}
	\sdso(q, q')=-\sn q-\slashed{\star}\sn q'
	\end{align*}
\end{itemize}
Consequently, proceeding again as in \cite{D--H--R}, the family of operators $\slashed{\mcala}$ are defined as follows, with their restriction to $\Sigma$ immediate:
\begin{itemize}
	\item the operators $\smcA{i}$ are defined inductively as
	\begin{align*}
	\smcA{2i+1}:=r\sn\smcA{2i},\qquad\smcA{2i+2}:=-r\sdiv\smcA{2i+1}
	\end{align*} 
	with $\smcA{1}=r\sn$
	\item the operators $\vmcA{i}$ are defined inductively as
	\begin{align*}
	\vmcA{2i+1}:=r\sdo\vmcA{2i},\qquad\vmcA{2i+2}:=r\sdso\vmcA{2i+1}
	\end{align*} 
	with $\vmcA{1}=r\sdo$
	\item the operators $\tmcA{i}$ are defined inductively as
	\begin{align*}
	\tmcA{2i+1}:=r\sdt\vmcA{2i},\qquad\tmcA{2i+2}:=r\sdst\vmcA{2i+1}
	\end{align*} 
	with $\tmcA{1}=r\sdt$
\end{itemize}

\begin{proposition}\label{propellipticestimatesonA}
Let $Q'\in H^{m,\prime}_{\taus, r}$ be a symmetric $2$-covariant $\qm$-tensor, $\momega'\in H^{m,\prime}_{\taus, r}$ be a $\qmsm$ 1-form and let $\theta'\in H^{m,\prime}_{\taus, r}$ be a symmetric, traceless 2-covariant $\sm$-tensor respectively. Then for any 2-sphere $S^2_{\taus,r}$ and any integer $m\geq 0$
	\begin{align*}
	\|Q'\|^2_{H^m_{\taus, r}}&\lesssim\snorm{\mcala_f^{[m]}Q'}[\taus][r],\\
	\|\momega'\|^2_{H^m_{\taus, r}}&\lesssim\snorm{\mcala_\xi^{[m]}\momega'}[\taus][r],\\
	\|\theta'\|^2_{H^m_{\taus, r}}&\lesssim\snorm{\mcala_\theta^{[m]}\theta'}[\taus][r].
	\end{align*}
\end{proposition}
\begin{proof}
	We prove the proposition first assuming smoothness of all quantities and then appeal to a density arguement.
	
	We first note the identities
	\begin{align*}
		\sdso\sdo&=-\slap+\frac{1}{r^2},\\
		\sdst\sdt&=-\frac{1}{2}\slap+\frac{1}{r^2}.
	\end{align*}
	Computing thus in the frame $\{S,L\}$ one finds that on every 2-sphere $S^2_{\taus,r}$
	\begin{align}
	\snorm{\smcA{1}Q'}[\taus][r]&=\snorm{r\sn Q'}[\taus][r],\label{identityAf1}\\
	\snorm{\vmcA{1}\momega'}[\taus][r]&=\snorm{r\sn \momega'}[\taus][r]+\snorm{\momega'}[\taus][r],\label{identityAv1}\\
	\snorm{\tmcA{1}\theta'}[\taus][r]&=\frac{1}{2}\snorm{r\sn \theta'}[\taus][r]+\snorm{\theta'}[\taus][r]\label{identityAt1}
	\end{align}
	and
	\begin{align*}
	\snorm{\smcA{2}Q'}[\taus][r]&=\snorm{r^2\slap Q'}[\taus][r],\\
	\snorm{\vmcA{2}\momega'}[\taus][r]&=\snorm{r^2\slap \momega'}[\taus][r]+2\snorm{r\sn \momega'}[\taus][r]+\snorm{\momega'}[\taus][r],\\
	\snorm{\tmcA{2}\theta'}[\taus][r]&=\frac{1}{4}\snorm{r^2\slap \theta'}[\taus][r]+\snorm{r\sn \theta'}[\taus][r]+\snorm{\theta'}[\taus][r].
	\end{align*}
	The former along with the Poincar\'e inequality immediately yields the $m=1$ case of the proposition whereas the latter combined with elliptic estimates on $\slap$ and the Poincar\'e inequality yields the $m=2$ case.
	
	The higher order cases then follow by an inductive procedure.
\end{proof}
Of course, replacing $Q'$ by a $\qm$ 1-form $q'\in H^{m,\prime}_{\taus,r}$ or a function $f'\in H^{m,\prime}_{\taus,r}$ yields an analagous result. Similarly, replacing $H^{m,\prime}_{\taus,r}$ with $H^{m,\prime}_{\mathsmaller{\Sigma},r}$ then an analagous result holds for $Q$ a function on $\Sigma$, $\momega$ a $\smi$ 1-form and $\theta'$ a symmetric, traceless 2-covariant $\smi$-tensor.

\subsubsection{The family of operators $\slashed{\Pi}$ and $\mcalt$}\label{ThefamilyofoperatorsPandT}

Next we introduce a family of operators on $\mcalm$ and $\Sigma$ which have the useful property of mapping all types of tensor fields into the spaces $\LM$ and $\LS$ respectively. We moreover define a family of elliptic operators that are naturally related to the former family and which shall appear frequently in the text.

To achieve that this family indeed have ranges contained within $\LM$ and $\LS$ we use the result of Lemma 4.4.1 of \cite{D--H--R} which shows that the $l=0,1$ spherical harmonics lie in the kernel of the operator $\sdst\sdso$. The desired operators are then defined as follows, with their restriction to $\Sigma$ immediate:
\begin{itemize}
	\item the operator $\sps$ maps $n$-covariant $\qm$-tensors $Q$ into the $n$-covariant $\qm$-tensor in $\LM$
	\begin{align*}
	\sps Q=r^4\sdiv\sdiv\sn\otimeshat\sn Q
	\end{align*}
	\item the operator $\spv$ maps $\qmsm$ 1-forms $\momega$ into the pair of $\qm$ 1-forms in $\LM$
	\begin{align*}
	\spv \momega=r^4\sdo\sdiv\sn\otimeshat\momega
	\end{align*}
	\item the operator $\spt$ maps symmetric, traceless 2-covariant $\sm$-tensors $\theta$ into the pair of functions in $\LM$
	\begin{align*}
	\spt \theta=r^4\sdo\sdt\theta
	\end{align*}
\end{itemize}
The family of elliptic operators, which now preserve the $\sm$-type of the tensor under consideration, are then defined as follows, with their restriction to $\Sigma$ immediate:
\begin{itemize}
	\item the operator $\mcaltf$ maps $n$-covariant $\qm$-tensors $Q$ into the $n$-covariant $\qm$-tensor in $\LM$
	\begin{align*}
	\mcaltf Q=\sps Q
	\end{align*}
	\item the operator $\mcaltv$ maps $\qmsm$ 1-form $\momega$ into the $\qmsm$ 1-forms in $\LM$
	\begin{align*}
	\mcaltv \momega=\sdso\spv\momega
	\end{align*}
	\item the operator $\mcaltt$ maps symmetric, traceless 2-covariant $\sm$-tensors $\theta$ into the symmetric, traceless 2-covariant $\sm$-tensor
	\begin{align*}
	\mcaltt \theta=\sdst\sdso\spt\theta
	\end{align*}
\end{itemize}

\begin{proposition}\label{propellipticestimatesonT}
Let $Q'\in H^{m,\prime}_{\taus, r}$ be a symmetric $2$-covariant $\qm$-tensor, $\momega'\in H^{m,\prime}_{\taus, r}$ be a $\qmsm$ 1-form and let $\theta'\in H^{m,\prime}_{\taus, r}$ be a symmetric, traceless 2-covariant $\sm$-tensor respectively. Then for any 2-sphere $S^2_{\taus,r}$ and any integer $m\geq 0$
	\begin{align*}
	\|Q'\|^2_{H^{m+4}_{\taus, r}}&\lesssim\snorm{\mcala_f^{[m]}\mcaltf Q'}[\taus][r],\\
	\|\momega'\|^2_{H^{m+4}_{\taus, r}}&\lesssim\snorm{\mcala_\xi^{[m]}\mcaltv \momega'}[\taus][r],\\
	\|\theta'\|^2_{H^{m+4}_{\taus, r}}&\lesssim\snorm{\mcala_\theta^{[m]}\mcaltt\theta'}[\taus][r].
	\end{align*}
\end{proposition}
\begin{proof}
We prove the proposition first assuming smoothness of all quantities and then appeal to a density arguement.

Using the identities
\begin{align*}
\sdiv\sn\otimeshat\momega&=\slap\momega+\frac{1}{r^2}\momega,\\
\sdiv\sn\otimeshat\sdiv\theta&=\slap\sdiv\theta+\frac{1}{r^2}\sdiv\theta
\end{align*}
we find
\begin{align*}
\mcaltf&=\smcA{4}-2\smcA{2},\\
\mcaltv&=-\vmcA{4}+2\vmcA{2},\\
\mcaltt&=2\tmcA{4}+2\tmcA{2}
\end{align*}
and so on every 2-sphere $S^2_{\taus,r}$
\begin{align*}
\snorm{\smcA{m}\mcaltf Q'}[\taus][r]&=\snorm{\smcA{m+4} Q'}[\taus][r]-4\snorm{\smcA{m+3} Q'}[\taus][r]+4\snorm{\smcA{m+2} Q'}[\taus][r],\\
\snorm{\vmcA{m}\mcaltv\momega'}[\taus][r]&=\snorm{\vmcA{m+4}\momega'}[\taus][r]-4\snorm{\vmcA{m+3}\momega'}[\taus][r]+4\snorm{\vmcA{m+2}\momega'}[\taus][r],\\
\snorm{\tmcA{m}\mcaltt\theta'}[\taus][r]&=4\snorm{\mcala^{[m+4]}\theta'}[\taus][r]+8\snorm{\mcala^{[m+3]}\theta'}[\taus][r]+4\snorm{\mcala^{[m+2]}\theta'}[\taus][r].
\end{align*}
Applying the Poincar\'e inequality in the first and second lines along with the identities \eqref{identityAf1}-\eqref{identityAt1} and then appealing to Proposition \ref{propellipticestimatesonA} thus yields the proposition.
\end{proof}
Of course, replacing $Q'$ by a $\qm$ 1-form $q'\in H^{m,\prime}_{\taus,r}$ or a function $f'\in H^{m,\prime}_{\taus,r}$ yields an analagous result. Similarly, replacing $H^{m,\prime}_{\taus,r}$ with $H^{m,\prime}_{\mathsmaller{\Sigma},r}$ then an analagous result holds for $Q$ a function on $\Sigma$, $\momega$ a $\smi$ 1-form and $\theta'$ a symmetric, traceless 2-covariant $\smi$-tensor.

An corollary of this proposition is that the family $\mcalt$ are bijections.
\begin{corollary}\label{corrinvertingToperators}
Let $Q'$ be a smooth symmetric $2$-covariant $\qm$-tensor with vanishing projection to $l=0,1$, $\momega'$ be a smooth $\qmsm$ 1-form and let $\theta'$ be a smooth symmetric, traceless 2-covariant $\sm$-tensor with vanishing projection to $l=0,1$. Conversely, let $Q_\prime$ be a smooth symmetric $2$-covariant $\qm$-tensor supported only on $l=0,1$ and let $\momega_\prime$ be a smooth $\qmsm$ supported only on $l=0,1$. Then on $D^+(\Sigma)$ there exists a smooth symmetric $2$-covariant $\qm$-tensor, a smooth $\qmsm$ 1-form $\momega$ and a smooth, symmetric, traceless 2-covariant $\sm$-tensor $\theta$ such that
\begin{align*}
\mcaltf Q&=Q',\\
\mcaltv\momega&=\momega',\\
\mcaltt\theta&=\theta'
\end{align*}
and
\begin{align*}
\mpart{Q}&=Q_\prime,\\
\mpart{\momega}&=\momega_\prime.
\end{align*}
\end{corollary}
\begin{proof}
We prove the corollary only for the operator $\mcaltf$ acting on smooth $n$-covariant $\qm$-tensors $Q$ since the other cases are similar.

Indeed, viewing $\mcaltf$ as a map from $H^{4,\prime}_{\taus,r}$ to $L^{2,\prime}_{\taus,r}$ we have from Proposition \ref{propellipticestimatesonT} and compactness of the embedding that $\mcaltf$ has closed range. Moreover, proceeding as in the proof of Lemma 4.4.1 in \cite{D--H--R} one finds that $\mcaltf$ and its formal adjoint have trivial kernels over $H^{4,\prime}_{\taus,r}$ and $L^{2,\prime}_{\taus,r}$ respectively and thus $\mcaltf$ is both injective and surjective. Finally,  the smoothness property results from applying standard elliptic-type estimates (in particular using the frame $\{S,L\}$).
\end{proof}
Of course, replacing $Q', Q_{\prime}$ by smooth $\qm$ 1-forms $q', q_{\prime}$ or smooth functions $f', f_{\prime}$ with the same assumptions regarding the their supports with respect to the $l=0,1$ spherical harmonics yields an analagous result. Similarly, an analagous result holds for $Q', Q_{\prime}$ smooth functions on $\Sigma$, $\momega', \omega_{\prime}$ smooth $\smi$ 1-forms and $\theta'$ a smooth, symmetric, traceless 2-covariant $\smi$-tensor.

\subsubsection{The operator $\szeta$}\label{Theellipticoperatorszslap}

The final operator we introduce is to appear in the definition of the Zerilli equation in section \ref{TheRWandZeqns}.

The operator in question is to be defined as the inverse of the operator $\zslap :=\slap +\frac{2}{r^2}\Big(1-\frac{3M}{r}\Big)\Id$.
\begin{lemma}
The operator $\zslap$ is a bijection from smooth functions in $\LM$ to smooth functions in $\LM$.
\end{lemma}
\begin{proof}
	Invertibility follows from applying the Lax--Millgram Theorem (see \cite{Evans}) combined with the Poincar\'e inequality. Elliptic regularity then yields smoothness.
\end{proof}
Of course, the above lemma holds replacing $\LM$ with $\LS$.

Consequently, we define the inverse operator $\szeta$, with its restriction to $\Sigma$ immediate:
\begin{itemize}
	\item the operator $\szetap{p}$ maps functions in $\LM$ to the function in $\LM$
	\begin{align*}
	\szetap{p}f=r^{2p}\zslap^{-p}f.
	\end{align*}
\end{itemize}
Elliptic regularity then yields the following proposition.
\begin{proposition}\label{propellipticestimateszlsap}
Let $f\in\LM$ be smooth. Then for any integer $p\geq 0$ and any 2-sphere $S^2_{\taus,r}$
\begin{align*}
||\szetap{-p}f||^2_{H^{p+2}_{\taus,r}}\lesssim\snorm{f}[\taus][r].
\end{align*}
\begin{proof}
Integrating by parts on any 2-sphere $S^2_{\taus,r}$ one finds 
	\begin{align*}
	\snorm{\zslap f}[\taus][r]=\snorm{\sn\sn f}[a][b]-\frac{3}{r^3}(r-6M)\snorm
	{\sn f}[\taus][r]+\frac{4}{r^6}(r-3M)^2\snorm{f}[\taus][r].
	\end{align*}
	The Poincar\'e inequality thus yields 
	\begin{align*}
	||f||^2_{H^2_{\taus,r}}\lesssim \snorm{r^2\zslap f}[\taus][r]
	\end{align*}
from which the proposition follows.
\end{proof}
\end{proposition}
Of course, an analagous result holds replacing $\LM$ by $\LS$.

\subsubsection{Commutation formulae and useful identities}\label{Commutationformulaeandusefulidentities}

In this final section we collect via various lemmata certain commutation relations and identities on the operators introduced over the previous sections that will be used throughout the text.

The first such lemma concerns the family of operators $\slashed{\mcala}$ and follows from the presence of the $r$-weights in the definitions of said operators.
\begin{lemma}\label{lemmacommA}
Let $\slashed{\mcala}^{[k]}$ denote any of the operators $\smcA{k}, \vmcA{k}$ or $\tmcA{k}$ for any integer $k\geq 1$. Then one has the commutation relations
\begin{align*}
\big[\qn, \slashed{\mcala}^{[k]}\big]=\big[\snnu, \slashed{\mcala}^{[k]}\big]=0
\end{align*}
where in the latter we view the family of operators $\slashed{\mcala}$ as restricted to $\Sigma$.
\end{lemma}
The second such lemma concerns the $\spi$ family of operators and follows from dilligently commuting covariant derivatives.
\begin{lemma}\label{lemmacommutingspoperators}
Let $\spi$ denote any of the operators $\sps, \spv$ and $\spt$. Then one has the commutation relations
\begin{align*}
\big[\qn,\spi\big]=\big[\snnu,\spi\big]=0
\end{align*}
where in the latter we view the family of operators $\slashed{\Pi}$ as restricted to $\Sigma$.

Moroever
\begin{align*}
\big[\slap,\sps\big]&=0,\\
\big[\slap,\spv\big]&=-\frac{1}{r^2}\spv,\\
\big[\slap,\spt\big]&=-\frac{2}{r^2}\spt.
\end{align*}
Finally, one has the identities
\begin{align*}
\sps\sdo&=\slap\spv,\\
\spv\sdt&=\slap\spt+\frac{2}{r^2}\spt.
\end{align*}
\end{lemma}
The third such lemma concerns the $\mcalt$ family of operators and follows from their very definitions.
\begin{lemma}\label{lemmacommformToperators}
Let $\mcalt$ denote any of the operators $\mcaltf, \mcaltv$ and $\mcaltt$. Then one has the commutation relations
\begin{align*}
\big[\qn,\mcalt\big]=\big[\snnu,\mcalt\big]=0
\end{align*}
where in the latter we view the family of operators $\mcalt$ as restricted to $\Sigma$.

Moroever one has the identities
\begin{align}
\sdso\Big(\mcaltf,0\Big)&=-\mcaltv\sdso\\
\sdst\mcaltv&=-2\mcaltt\sdst.
\end{align}
\end{lemma}
The final such lemma concerns the operator $\szeta$.
\begin{lemma}\label{lemmacommformszeta}
For any non-negative integer $k$ one has the commutation relation
\begin{align*}
\big[\qn, \szetap{k}\big]=-3k\frac{\mu}{r}\qn r\,\szetap{k+1},\qquad\big[\snnu, \szetap{k}\big]=-3\frac{k}{\bh}\frac{\mu}{r}\,\szetap{k+1}
\end{align*}
where in the latter we view the operator $\szetap{k}$ as restricted to $\Sigma$.
\end{lemma}
\begin{proof}
We prove only the former as the latter follows in a similar fashion.

We have
\begin{align*}
\big[\qn,r^2\zslap\big]=\frac{3\mu}{r}\qn r\,\textnormal{Id}.
\end{align*}
The $k=1$ case thus follows from the formula
\begin{align*}
[\qn, \szetap{1}]=-\szetap{1}\big[\qn, r^2\zslap\big]\szetap{1}
\end{align*}
For general $k$, one applies the induction formulae
\begin{align*}
\big[\qn, \szetap{n}\big]=\big[\qn, \szetap{n-1}\big]\szetap{1}+\szetap{n-1}\big[\qn, \szetap{1}\big].
\end{align*}
\end{proof}

\section{The equations of linearised gravity around Schwarzschild}\label{TheequationsoflinearisedgravityaroundSchwarzschild}

In this section we derive the equations of linearised gravity -- the system of equations that result from formally linearising the equations of section \ref{TheEinsteinequations} about the Schwarzschild exterior solution $\Mgs$.

\subsection{The formal linearisation of the equations of section \ref{TheEinsteinequations}}\label{The formal linearisation of the equations of section}

We begin by establishing a formal linearisation theory for the Einstein equations in a generalised wave gauge. 

This section of the paper proceeds in a similar fashion to that of section 5.1 in \cite{D--H--R}.

\subsubsection{Preliminaries}\label{Preliminaries}

We identify the manifold $\boldsymbol{\mcalm}$ in section \ref{Thegeneralisedwavegauge} with the Schwarzschild exterior background $\mathcal{M}$.

On $\mathcal{M}$ we consider a smooth 1-parameter family of Lorentzian metrics $\bold{g}(\epsilon)$ with $\bold{g}(0)=g_M$ and a smooth map $\boldsymbol{f}:T^2(\mcalm)\times T^2(\mcalm)\rightarrow T\mcalm$ such that $\boldsymbol{f}(g_M, g_M)=0$ and demand that each $\boldsymbol{g}(\epsilon)$ is in a generalised $\boldsymbol{f}$-wave gauge with respect to $g_M$.

We moreover assume that each pair $\big(\mcalm, \boldsymbol{g}(\epsilon)\big)$ is a solution to the Einstein vacuum equations.
 
\subsubsection{The linearisation procedure}\label{Thelinearisationprocedure}

Let us first immediately dispense with the $\epsilon$ notation and use the convention that bold quantities are with respect to the family of perturbed metrics and unbolded quantities are given by their background Schwarzschild value.

In particular, identifying $\boldsymbol{\overline{g}}$ with $g_M$ in equations \eqref{eqndefnofageneralisedwavegauge} and \eqref{einstein equations in a generalised wave gauge} of section \ref{ThevacuumEinsteinequationsinageneralisedwavegauge}, the assumptions of section \ref{Preliminaries} imply that each member $\boldsymbol{g}$ of the 1-parameter family must satisfy
\begin{align}
\bold{\big(g^{-1}\big)}^{\gamma\delta}\unbold{\nabla}[\gamma]\unbold{\nabla}[\delta]\bold{g}_{\alpha\beta}+2\boldsymbol{C}^{\gamma}_{\delta\epsilon}\cdot\bold{g}_{\gamma(\alpha}\unbold{\nabla}[\beta)]\bold{\big(g^{-1}\big)}^{\delta\varepsilon}-4\bold{g}_{\delta\epsilon}\boldsymbol{C}^{\epsilon}_{\beta[\alpha}\unbold{\nabla}[{\gamma]}]\bold{\big(g^{-1}\big)}^{\gamma\delta}-4\boldsymbol{C}^\delta_{\beta[\alpha}\boldsymbol{C}^\gamma_{\gamma]\delta}+2\bold{\big(g^{-1}\big)}^{\gamma\delta}\bold{g}_{\epsilon(\alpha}\unbold{\text{Riem}}[\beta)\gamma\delta][\varepsilon]\nonumber\\=2\bold{g}_{\gamma(\alpha}\unbold{\nabla}[\beta)]\boldsymbol{f}^\gamma(\boldsymbol{g}, g_M),\label{einstein equations about schwarzschild}\\
\boldsymbol{\big(g^{-1}\big)}^{\beta\gamma}\boldsymbol{C}^\alpha_{\beta\gamma}=\boldsymbol{f}^\alpha(\boldsymbol{g}, g_M)\label{wave gauge about schwarzschild},
\end{align}
Here, Riem and $\nabla$ are the Riemann curvature tensor and Levi-Civita connection of $\Mgs$ respectively. Observe therefore that $\boldsymbol{g}=g_M$ is indeed a solution to this system of equations since by assumption $\boldsymbol{f}(g_M, g_M)=0$.

To formally linearise, we write
\begin{align*}
\boldsymbol{g}-g_M&\equiv \glin
\end{align*}
where $\equiv$ means equivalent to first order in $\epsilon$. Thus, in keeping with the notation of \cite{D--H--R}, quantities with a superscript ``(1)'' denote linear perturbations of bolded quantities about their background Schwarzschild value. In particular,
\begin{itemize}
	\item the smooth, symmetric two tensor field $\glin$ on $\mcalm$ denotes the linearised metric
\end{itemize} 
Moreover, we write
\begin{align*}
\boldsymbol{f}(\boldsymbol{g}, g_M)\equiv Df\big|_{g_M}(\glin)
\end{align*}
where $Df\big|_{g_M}:T^2(\mcalm)\rightarrow T\mcalm$ is a smooth linear map. In particular,
\begin{itemize}
	\item the smooth linear map $Df\big|_{g_M}$ denotes the linearisation of the map $\boldsymbol{f}(\cdot, g_M)$ at $g_M$
\end{itemize}

Subsequently, to derive the linearised equations one simply expands the terms appearing in equations \eqref{einstein equations about schwarzschild} and \eqref{wave gauge about schwarzschild} in powers of $\epsilon$, keeping only those terms that enter to first order. 

\subsubsection{The Einstein equations in a generalised wave gauge linearised about the Schwarzschild metric}\label{TheEinsteinequationsinageneralisedwavegaugelinearisedabouttheSchwarzschildmetric}

Proceeding in this manner one arrives at the following system of equations:
\begin{align}
\Box \lin{g}[\alpha\beta]-2\tensor{\text{Riem}}{^\gamma_{\alpha\beta}^\delta}\lin{g}[\gamma\delta]&=2\nabla_{(\alpha}\lin{f}_{\beta)},\label{eqnlinearisedeinsteinequations}\\
\nabla^\beta\glin_{\alpha\beta}-\frac{1}{2}\n_\alpha\glin&=\lin{f}_{\alpha}\label{eqnlorentzgauge}.
\end{align}
Here, $\Box$ is the wave operator on $\big(\mcalm, g_M\big)$ and we have set
\begin{align*}
\tflin:=Df\big|_{g_M}(\glin).
\end{align*}

The above system of equations thus describe the linearisation of the Einstein vacuum equations, as expressed in a generalised $\boldsymbol{f}$-wave gauge with respect to $g_M$, about the Schwarzschild exterior solution $g_M$.

\subsection{The system of gravitational perturbations}\label{Thesystemofgravitationalperturbations}

In this section we apply the $2+2$ formalism developed in section \ref{The2+2formalism} to the linearised metric of the previous section.\newline

Indeed, given the symmetric two tensor $\glin$ on $\mcalm$ we have:
\begin{itemize}
	\item the projection onto the symmetric 2-covariant $\qm$-tensor field $\lin{\tilde{g}}$
	\item the projection onto the $\qm\otimes\sm$ 1-form $\mling$
	\item the projection onto the symmetric 2-covariant $\sm$-tensor field $\sglin$
\end{itemize}
Furthermore, given the 1-form $\lin{f}$ on $\mcalm$ we have:
\begin{itemize}
	\item the projection onto the $\qm$ 1-form $\tflin$ 
	\item the projection onto the $\sm$ 1-form $\sflin$
\end{itemize}
Moreover, decomposing $\lin{\tilde{g}}$ and $\sglin$ into their trace and trace-free parts with respect to the $\qm$-metric $\qg_M$ and the $\sm$-metric $\sg_M$
\begin{align*}
\lin{\tilde{g}}=\gabhatlin+\frac{1}{2}\qg_M\cdot\trgablin
\end{align*}
and
\begin{align*}
\sglin=\sghatlin+\frac{1}{2}\sg_M\cdot\trsglin
\end{align*}
yields the collection
\begin{align*}
\gabhatlin, \trgablin, \mling, \sghatlin, \trsglin, \tflin, \sflin
\end{align*}
where now
\begin{itemize}
	\item $\gabhatlin$ is a symmetric, traceless 2-covariant $\qm$-tensor field
	\item $\trgablin$ is a function on $\mcalm$
	\item $\mling$ is a $\qm\otimes\sm$ 1-form
	\item $\sghatlin$ is a symmetric, traceless 2-covariant $\sm$-tensor field
	\item $\trsglin$ is a function on $\mcalm$
	\item $\tflin$ is a $\qm$ 1-form
	\item $\sflin$ is an $\sm$ 1-form
\end{itemize}
each of which must satisfy the following system of gravitational perturbations.\newline

The equations for the linearised $\mcalq$-metric:

\begin{align}
\qbox\gabhatlin+\slap\gabhatlin+\frac{2}{r}\qn_{P}\gabhatlin-\frac{2}{r^2}\ud r\qotimeshat\gabhatlin_P-\frac{2}{r}\frac{\mu}{r}\gabhatlin&=\frac{1}{r^2}\ud r\qotimeshat\ud r\,\trgablin+\frac{2}{r}\ud r\qotimeshat\sdiv\mling-\frac{1}{r^2}\ud r\otimeshat\ud r\,\trsglin+\qn\otimeshat\tflin,\label{eqnforgabhatfgauge}\\
\qbox\trgablin+\slap\trgablin+\frac{2}{r}\qn_{P}\trgablin-\frac{2}{r^2}(1-2\mu)\trgablin&=\frac{4}{r^2}\gabhatlin_{P P}+\frac{4}{r}\sdiv\mling_P-\frac{2}{r^2}(1-2\mu)\trsglin-2\qd\tflin\label{eqnfortrgabfgauge}.
\end{align}

The equations for the linearised $\qmsm$-metric:

\begin{align}
\qbox\mling+\slap\mling-\frac{1}{r^2}\ommu\mling=\frac{2}{r}\ud r\,{\otimes}\bigg(\frac{2}{r}\mling_P+\sdiv\sghatlin+\frac{1}{2}\sn\trsglin-\sflin\bigg)-\frac{2}{r}\sn{\otimes}\gabhatlin_P-\frac{1}{r}\ud r\,\sn{\otimes}\trgablin+\qn{\otimes}\sflin+\sn{\otimes}\tflin\label{eqnforgAbfgauge}.
\end{align}

The equations for the linearised $\sm$-metric:

\begin{align}
\qbox\sghatlin+\slap\sghatlin-\frac{2}{r}\qn_{P}\sghatlin-\frac{4}{r}\frac{\mu}{r}\sghatlin&=-\frac{2}{r}\sn\otimeshat\mling_P+\sn\otimeshat\sflin\label{eqnforgABfgauge},\\
\qbox\trsglin+\slap\trsglin+\frac{2}{r}\qn_{P}\trsglin-\frac{2}{r^2}(1-2\mu)\trsglin&=-\frac{4}{r^2}\gabhatlin_{P P}-\frac{4}{r}\sdiv\mling_P-\frac{2}{r^2}(1-2\mu)\trgablin+2\sdiv\sflin+\frac{4}{r}\tflin_P\label{eqnfortrgABfgauge}.
\end{align}

Finally, the linearised generalised wave gauge conditions:

\begin{align}
-\qd\gabhatlin+\sdiv\mling-\frac{1}{2}\qexd\trsglin+\frac{2}{r}\gabhatlin_P+\frac{1}{r}\ud r\,\trgablin-\frac{1}{r}\ud r\,\trsglin&=\tflin\label{eqnfwavegauge1},\\
-\qd\mling-\frac{1}{2}\sn\trgablin+\sdiv\sghatlin+\frac{2}{r}\mling_P&=\sflin\label{eqnfwavegauge2}.
\end{align}

\begin{remark}\label{rmkdecomposingboldmetric}
	In view of the fact that in section \ref{Preliminaries} we fixed the differential structure of $\mcalm$, one could instead decompose the perturbed metric $\boldsymbol{g}$ as in section \ref{The2+2decompositionoftensorfieldsonM}. A further decomposition of \eqref{eqndefnofageneralisedwavegauge} and \eqref{einstein equations in a generalised wave gauge} would lead to a system of equations for the quantities
	\begin{align*}
	\boldsymbol{\tilde{g}}, \boldsymbol{\stkout{g}}, \boldsymbol{\slashed{g}}, \boldsymbol{\tilde{f}}, \boldsymbol{\slashed{f}}.
	\end{align*}
	Upon linearisation, this yields the system \eqref{eqnforgabhatfgauge}-\eqref{eqnfwavegauge2} for
	\begin{align*}
	\lin{\tilde{g}}, \mling, \sglin, \tflin, \sflin
	\end{align*}
	where we now recognise the tensor fields $\lin{\tilde{g}}, \mling, \sglin, \tflin, \sflin$ as the linearisation of the quantities $\boldsymbol{\tilde{g}}, \boldsymbol{\stkout{g}}, \boldsymbol{\slashed{g}}, \boldsymbol{\tilde{f}}$ and $\boldsymbol{\slashed{f}}$:
	\begin{align*}
	\boldsymbol{\tilde{g}}-\qg_M&\equiv\lin{\tilde{g}},\\
	\boldsymbol{\stkout{g}}&\equiv\mling,\\
	\boldsymbol{\slashed{g}}-\sg_M&\equiv\sglin,\\
	\boldsymbol{\tilde{f}}&\equiv\tflin,\\
	\boldsymbol{\slashed{f}}&\equiv\sflin.
	\end{align*}
	Since, however, the aforementioned decompositon of \eqref{eqndefnofageneralisedwavegauge} and \eqref{einstein equations in a generalised wave gauge} is rather cumbersome, we prefer our more direct approach.
\end{remark}

\subsection{The equations of linearised gravity}\label{Theequationsoflinearisedgravity}

In this section we make an explicit choice of the linear map $Df\big|_{g_M}$ and then present the resulting system of equations which we shall refer to in this paper as the equations of linearised gravity.

The remainder of the paper is then concerned with solutions to this system of equations.\newline

Indeed, viewing now $Df\big|_{g_M}$ as a smooth map from\footnote{Here we recall the notation $\mscrT^k(\mcalm)$ for the space of smooth $k$-covariant tensor fields on $\mcalm$.} $\mscrT^2(\mcalm)\rightarrow\mscrT(\mcalm)$, we set
\begin{align*}
Df\big|_{g_M}(X):=\frac{2}{r}\tilde{X}_P-\frac{1}{r}\big(\hat{\tilde{X}}_{l=0}\big)_P+\frac{2}{r}\stkout{X}_P-\frac{1}{r}\ud r\,\str\slashed{X}.
\end{align*}
Here, $\tilde{X}_{l=0}$ denotes the spherically symmetric part of the 2-covariant $\qm$-tensor $\tilde{X}$ (cf. Definition \ref{defnqmtensorssupportedonlgeq2}).

Recalling the definition of the 1-form $\flin$ from section \ref{TheEinsteinequationsinageneralisedwavegaugelinearisedabouttheSchwarzschildmetric} along with the subsequent definitions of the $\qm$ 1-form $\tflin$ and the $\sm$ 1-form $\sflin$ from section \ref{Thesystemofgravitationalperturbations}, it follows that
\begin{align}
\tflin:&=\frac{2}{r}\gabhatlin_P-\frac{1}{r}\big(\gabhatlin_{l=0}\big)_P+\frac{1}{r}\ud r\,\trgablin-\frac{1}{r}\ud r\,\trsglin,\label{eqndefnoftfliningoodgauge} \\
\sflin:&=\frac{2}{r}\mling_P \label{eqndefnofsfliningoodgauge}.
\end{align}
Consequently, 
with this choice of $\tflin$ and $\sflin$ the system of gravitational perturbations presented in section \ref{Thesystemofgravitationalperturbations} reduce to the following system of equations.\newline

The equations for the linearised $\mcalq$-metric:
\begin{align}
\qbox\gabhatlin+\slap\gabhatlin+\frac{2}{r}\qn_{P}\gabhatlin-\frac{2}{r}\big(\qn\otimeshat\gabhatlin\big)_P-\frac{4}{r}\frac{\mu}{r}\gabhatlin&=\frac{1}{r}\ud r\qotimeshat\qexd\trgablin+\frac{2}{r}\ud r\qotimeshat\sdiv\mling-\frac{1}{r}\ud r\qotimeshat\qexd\trsglin-\qn\otimeshat\Big(r^{-1}\big({\gabhatlin_{l=0}}\big)_P\Big),\label{eqnlingrav1}\\
\qbox\trgablin+\slap\trgablin&=-\frac{2}{r}\qd\Big({\gabhatlin_{l=0}}\big)_P\Big)+\frac{2}{r^2}\big(\gabhatlin_{l=0}\big)_P\label{eqnlingrav2}.
\end{align}

The equations for the linearised $\qmsm$-metric:

\begin{align}
\qbox\mling+\slap\mling-\frac{2}{r}\big(\qn{\otimes}\mling\big)_P-\frac{1}{r^2}\mling+\frac{2}{r^2}\ud r{\otimes}\mling_P=\frac{2}{r}\ud r\,{\otimes}\sdiv\sghatlin\label{eqnlingrav3}.
\end{align}

The equations for the linearised $\sm$-metric:

\begin{align}
\qbox\sghatlin+\slap\sghatlin-\frac{2}{r}\qn_{P}\sghatlin-\frac{4}{r}\frac{\mu}{r}\sghatlin&=0,\label{eqnlingrav4}\\
\qbox\trsglin+\slap\trsglin+\frac{2}{r}\qn_{P}\trsglin+\frac{2}{r^2}\trsglin&=\frac{4}{r^2}\gabhatlin_{P P}+\frac{2}{r^2}\trgablin-\frac{4}{r^2}\big(\gabhatlin_{l=0}\big)_{PP}\label{eqnlingrav5}.
\end{align}

Finally, the linearised generalised wave gauge conditions:

\begin{align}
-\qd\gabhatlin+\sdiv\mling-\frac{1}{2}\qexd\trsglin&=-\frac{1}{r}\big(\gabhatlin_{l=0}\big)_P,\label{eqnlingrav6}\\
-\qd\mling-\frac{1}{2}\sn\trgablin+\sdiv\sghatlin&=0\label{eqnlingrav7}.
\end{align}

Following \cite{D--H--R}, we shall collectively denote a collection $\gabhatlin, \trgablin, \mling, \sghatlin, \trsglin$ of tensor fields on $\mcalm$ which satisfying the above system of equations according
\begin{align*}
\mathscr{S}=\bigg(\gabhatlin, \trgablin, \mling, \sghatlin, \trsglin\bigg)
\end{align*} 
and we shall refer to the ensemble $\mathscr{S}$ as a solution to \textbf{the equations of linearised gravity.}\newline

We make the following remarks.
\begin{remark}\label{rmkl=0modes}
The `$l=0$ modification' in the expression $\tflin$ is motivated by the linearised Schwarzschild solutions we present in section \ref{Linearisingaone-parameterfamilyofSchwarzschildsolutions}.
\end{remark}
\begin{remark}\label{rmkmotivationforchoosing}
Observe that, modulo the `$l=0$ modes', the choice of $\tflin$ and $\sflin$ as in \eqref{eqndefnoftfliningoodgauge}-\eqref{eqndefnofsfliningoodgauge} removes all zero'th order terms in the generalised wave gauge conditions \eqref{eqnfwavegauge1}-\eqref{eqnfwavegauge2}.
\end{remark}
\begin{remark}\label{rmknonlinearmap}
We note that the linear map $\boldsymbol{{f}}:\mathscr{T}^2(\mcalm)\rightarrow \mscrT\mcalm$ defined according to
\begin{align*}
\boldsymbol{f}(X)&=\frac{2}{r}\tilde{X}_P-\frac{1}{r}\big(\hat{\tilde{X}}_{l=0}\big)_P+\frac{2}{r}\stkout{X}_P-\frac{1}{r}\ud r\,\str\slashed{X}
\end{align*}
trivially gives rise to the expressions $\flin$ and $\sflin$ of  \eqref{eqndefnoftfliningoodgauge}-\eqref{eqndefnofsfliningoodgauge}. Moreover, we further observe that
\begin{align*}
\boldsymbol{{f}}(g_M)=0.
\end{align*}
\end{remark}

\section{Special solutions to the equations of linearised gravity}\label{Specialsolutionstotheequationsoflinearisedgravity}
 
In this section we introduce two special classes of solutions to the equations of linearised gravity, namely the linearised Kerr and pure gauge solutions.

Knowledge of these special solutions will prove fundamental in the formulation of the main Theorem of this paper, to be found in section \ref{Precisestatementsofthemaintheorems}.

\subsection{Special solutions I: The 4-dimensional linearised Kerr family $\Ke$}\label{SpecialsolutionsI:The4-dimensionallinearisedKerrfamily}

The first class of solutions we introduce are to be the linearised Kerr family which arise as the linearisation of the Kerr exterior family about $\Mgs$.

This section of the paper proceeds in a similar fashion to that of section 6.2 in \cite{D--H--R}.

\subsubsection{Linearising a 1-parameter family of Schwarzschild solutions}\label{Linearisingaone-parameterfamilyofSchwarzschildsolutions}

We begin by linearising, in the mass parameter, the 1-parameter subfamily of the Kerr exterior family corresponding to the Schwarzschild exterior family.\newline

In order to formally linearise we consider a smooth 1-parameter family of functions $\boldsymbol{M}:(-\epsilon, \epsilon)\rightarrow\mathbb{R}$ with $\boldsymbol{M}(0)=M$. This family subsequently generates the smooth 1-parameter family of Schwarzschild exterior solutions, with mass $\boldsymbol{M}$, to the Einstein vacuum equations on $\mcalm$:
\begin{align*}
g_{\boldsymbol{M}}:=-\bigg(1-\frac{2\boldsymbol{M}}{r}\bigg){\ud t^*}^2+\frac{4\boldsymbol{M}}{{r}}\ud t^*\ud r+\bigg(1+\frac{2\boldsymbol{M}}{{r}}\bigg){\ud {r}}^2+{r}^2\mathring{g}.
\end{align*}
Here, we have dispensed with the $\epsilon$ notation.

Now, projecting the symmetric, 2-covariant tensor fields $g_{\boldsymbol{M}}$ onto $\qmm, \qmm\times\smm$ and $\smm$ respectively (recall Definition \ref{defnprojection}), we compute that
\begin{align*}
-\qd\hat{\tilde{g}}_{\boldsymbol{M}}-\frac{1}{2}\qexd\str{\slashed{g}_{\boldsymbol{M}}}+\frac{1}{r}\big((\hat{\tilde{g}}_{\boldsymbol{M}})_{l=0}\big)_P=-\frac{1}{2}\sn\qtr{g}_{\boldsymbol{M}}=0,\qquad \stkout{g}_{\boldsymbol{M}}=\hat{\slashed{g}}_{\boldsymbol{M}}=0.
\end{align*}
Here, $\hat{\tilde{g}}_{\boldsymbol{M}}$ is the traceless part of ${\tilde{g}}_{\boldsymbol{M}}$ with respect to $\qg_M$ and $\hat{\slashed{g}}_{\boldsymbol{M}}$ is the traceless part of ${\slashed{g}}_{\boldsymbol{M}}$ with respect to $\sg_M$.

It therefore follows (cf. Remark \ref{rmknonlinearmap}) that upon linearisation in the mass parameter $\boldsymbol{M}$ the smooth 1-parameter family of solutions to the Einstein vacuum equations $g_{\boldsymbol{M}}$ give rise to a smooth 1-parameter family of \emph{explicit} solutions to the equations of linearised gravity presented in section \ref{Theequationsoflinearisedgravity}.

Consequently, applying the formal linearisation theory of section \ref{The formal linearisation of the equations of section} to the 1-parameter family $g_{\boldsymbol{M}}$ results in:
\begin{align*}
\gabhatlin&=-\frac{1}{(1-\mu)^2}\frac{\lin{M}}{r}\Big(\qhd P\qotimeshat\qhd P-2\qhd P\qotimeshat\ud r-\ud r\qotimeshat\ud r\Big),\\
\trgablin=\trsglin=\mling=\sghatlin&=0.
\end{align*}
Here, $\boldsymbol{M}\equiv M+\epsilon\cdot\lin{M}$.

We have therefore shown the following.
\begin{proposition}\label{proplinssoln}
	For any $\mfM\in\mathbb{R}$ the following is a smooth (spherically symmetric and stationary) solution to the equations of linearised gravity:
	\begin{align*}
	\gabhatlin&=-\frac{1}{(1-\mu)^2}\frac{\mfM}{r}\Big(\qhd P\qotimeshat\qhd P-2\qhd P\qotimeshat\ud r-\ud r\qotimeshat\ud r\Big),\\
	\trgablin=\trsglin=\mling=\sghatlin&=0.
	\end{align*}
\end{proposition}

\subsubsection{Linearised Kerr solutions leaving the mass unchanged}\label{LinearisedKerrsolutionsleavingthemassunchanged}

We continue by linearising, now in the rotation parameter $\boldsymbol{a}$, the 1-parameter subfamily of the Kerr exterior metrics which have fixed mass $M$.\newline

In order to formally linearise we consider a smooth 1-parameter family of functions $\boldsymbol{a}:(-\epsilon, \epsilon)\rightarrow\mathbb{R}$ with $\boldsymbol{a}(0)=0$. This family subsequently generates the smooth 1-parameter family of Kerr exterior metrics on $\mcalm$
\begin{align*}
	g_{M, \boldsymbol{a}}=g_M-\frac{4}{r}\boldsymbol{a}M\sin^2\theta \ud t^*\ud\varphi-2\bigg(1+\frac{2M}{r}\bigg)\boldsymbol{a}\sin^2\theta \ud r\ud\varphi+\mathcal{O}\big(\boldsymbol{a}^2\big).
\end{align*}
Here, we have dispensed with the $\epsilon$ notation and neglected to explicitly state terms that are higher than linear order in $\boldsymbol{a}$.

Now, projecting the symmetric, 2-covariant tensor fields $g_{M,\boldsymbol{a}}$ onto $\qmm, \qmm\times\smm$ and $\smm$ respectively, we observe that
\begin{align*}
\sdiv\stkout{g}_{M,\boldsymbol{a}}-\frac{1}{2}\qexd\str{\slashed{g}_{M,\boldsymbol{a}}}=
-\qd\stkout{g}_{M,\boldsymbol{a}}-\frac{1}{2}\sn\qtr{g}_{M,\boldsymbol{a}}+\sdiv\hat{\slashed{g}}_{M,\boldsymbol{a}}=\mathcal{O}\big(\boldsymbol{a}^2\big),\qquad \hat{\tilde{g}}_{M,\boldsymbol{a}}=\hat{\slashed{g}}_{M,\boldsymbol{a}}=\mathcal{O}\big(\boldsymbol{a}^2\big).
\end{align*}
Here, $\hat{\tilde{g}}_{M,\boldsymbol{a}}$ is the traceless part of ${\tilde{g}}_{M,\boldsymbol{a}}$ with respect to $\qg_M$ and $\hat{\slashed{g}}_{M,\boldsymbol{a}}$ is the traceless part of ${\slashed{g}}_{M,\boldsymbol{a}}$ with respect to $\sg_M$.

It therefore follows that upon linearisation in the parameter $\boldsymbol{a}$ the smooth 1-parameter family of solutions to the Einstein vacuum equations $g_{M,\boldsymbol{a}}$ give rise to a smooth 1-parameter family of \emph{explicit} solutions to the equations of linearised gravity.

Consequently, applying the formal linearisation theory of section \ref{The formal linearisation of the equations of section} to the 1-parameter family $g_{\boldsymbol{M}}$ results in:
\begin{align*}
\mling&=-\sqrt{\frac{32\pi}{3}}\frac{1}{1-\mu}\lin{a}\Big(\mu\qhd P-\ud r\Big){\otimes}\shd\sn Y^1_0,\\
\gabhatlin=\trgablin=\sghatlin=\trgablin&=0.
\end{align*}
Here, $\boldsymbol{a}\equiv \epsilon\cdot\alin$ and $Y^1_0$ is the $l=1, m=0$ spherical harmonic of section \ref{Thel=0,1sphericalharmonicsandthesphericalharmonicdecomposition}.

We have therefore shown the ($i=0$ case of the) following.
\begin{proposition}\label{proplinkerrsoln}
	Let $\mathfrak{a}_i\in\mathbb{R}$ with $i\in\{-1,0,1\}.$ Then the following is a smooth (stationary) solution to the equations of linearised gravity:
	\begin{align*}
	\mling&=-\frac{1}{1-\mu}\mathfrak{a}_i\Big(\mu\qhd P-\ud r\Big){\otimes}\shd\sn Y^1_i,\\
	\sflin&=\frac{2}{r}\mathfrak{a}_i\shd\sn Y^1_i,\\
	\gabhatlin=\trgablin=\sghatlin=\trgablin=\tflin&=0.
	\end{align*}
\end{proposition}

\subsubsection{The linearised Kerr family}\label{ThelinearisedKerrfamily}

The full 4-parameter family of linearised Kerr solutions to the system of gravitational perturbations are then determined by combining the family of solutions introduced in the previous two sections.\newline

Indeed, by linearity, one has the following.
\begin{proposition}\label{propfullkerrfamily}
Let $\mfM\in\mathbb{R}$ and let $\mathfrak{a}$ be a smooth function on $S^2$ that is given as a linear combination of the $l=1$ spherical harmonics. Then the collection
\begin{align*}
\Ke_{\mathsmaller{\mfM,\mathfrak{a}}}=\bigg(\gabhatlin, \trgablin, \mling, \sghatlin, \trsglin\bigg)
\end{align*}
defined by
\begin{align*}
\gabhatlin&=-\frac{1}{(1-\mu)^2}\frac{\mfM}{r}\Big(\qhd P\qotimeshat\qhd P-2\qhd P\qotimeshat\ud r-\ud r\qotimeshat\ud r\Big),\\
\trgablin&=0,\\
\mling&=-\frac{1}{1-\mu}\Big(\mu\qhd P-\ud r\Big)\shd\sn \mathfrak{a},\\
\sghatlin&=0,\\
\trgablin=\mling=\trsglin&=0
\end{align*}
is a smooth (stationary) solution to the equations of linearised gravity. We call such a solution a \textnormal{\textbf{linearised Kerr solution}} with parameters $\mfM$ and $\mathfrak{a}$.
\end{proposition}

We make the following remarks.
\begin{remark}
The above indeed defines a 4-parameter family of solutions to the equations of linearised gravity as the function $\mathfrak{a}$ is parametrised by three real numbers $\mfa_i$ according to
\begin{align*}
\mfa=\sum_{i=-1}^{i=+1}\mfa_i Y^1_i.
\end{align*}
\end{remark}
\begin{remark}\label{rmkkerrkernal}
Observe that for any $\mfM$ and $\mathfrak{a}$ each of the above collection lie in the kernel of an appropriate member of the $\spi$ family of operators defined in section \ref{ThefamilyofoperatorsPandT}.
\end{remark}

\subsection{Special solutions II: Pure gauge solutions $\Ga$}\label{Specialsolutions2:Puregaugesolutions}

The second class of solutions we introduce are the pure gauge solutions which arise as the linearisation of a 1-parameter family of Lorentzian metrics on $\mcalm$ given as the pullback of $g_M$ under a 1-parameter family of diffeomorphisms on $\mcalm$ which preserve the generalised wave gauge to first order.\newline

Let $f:\mathscr{T}^2(\mcalm)\rightarrow \mathscr{T}\mcalm$ be the smooth map 
\begin{align*}
f(X)&=\frac{2}{r}\tilde{X}_P-\frac{1}{r}\big(\hat{\tilde{X}}_{l=0}\big)_P+\frac{2}{r}\stkout{X}_P-\frac{1}{r}\ud r\,\str\slashed{X}
\end{align*}
and let
$\qv$, $\sv$ be a pair of smooth smooth $\qm$ 1-forms and $\sm$ 1-forms respectively which satisfy
\begin{align}
\qbox\qv+\slap\qv-\frac{2}{r}\big(\qn\qv\big)_P+\frac{2}{r^2}\ud r\,\qv_P&=-\frac{1}{r}\big(\qn\otimes\qv_{l=0}\big)_P,\label{puregauge1}\\
\qbox\sv+\slap\sv-\frac{2}{r}\qn_{P}\sv+\frac{1}{r^2}(3-4\mu)\sv&=0\label{puregauge2}.
\end{align}
Denoting by $\boldsymbol{\phi}_\epsilon$ the smooth 1-parameter family of diffeomorphisms on $\mcalm$ generated by the smooth vector field
\begin{align*}
v:=\qv+\sv.
\end{align*}
it follows that the corresponding smooth 1-parameter family of Lorentzian metrics $(\boldsymbol{\phi}_\epsilon)^*g_M$ define a smooth 1-parameter family of solutions to the Einstein vacuum equations. Moreover, the conditions \eqref{puregauge1}-\eqref{puregauge2} ensure that, to first order in $\epsilon$, the identity map
\begin{align}\label{dick}
\text{Id}:\big(\mcalm, (\boldsymbol{\phi}_\epsilon)^*g_M\big)\rightarrow\big(\mcalm, g_M\big)
\end{align}
is an $f((\boldsymbol{\phi}_\epsilon)^*g_M,g_M)$-wave map. Indeed, to first order in $\epsilon$
\begin{align*}
(\boldsymbol{\phi}_\epsilon)^*g_M-g_M\equiv\mcall_vg_M
\end{align*}
and so
\begin{align*}
\big((\boldsymbol{\phi}_\epsilon)^*g_M\big)^{-1}\cdot C_{(\boldsymbol{\phi}_\epsilon)^*g_M,g_M}\equiv\qbox\qv+\slap\qv-\frac{2}{r}\big(\qn\qv\big)_P+\frac{2}{r^2}\ud r\,\qv_P+\frac{1}{r}\big(\qn\otimes\qv_{l=0}\big)_P+
\qbox\sv+\slap\sv-\frac{2}{r}\qn_{P}\sv+\frac{1}{r^2}(3-4\mu)\sv.
\end{align*}
Here, $C_{(\boldsymbol{\phi}_\epsilon)^*g_M,g_M}$ is the connection tensor between $(\boldsymbol{\phi}_\epsilon)^*g_M$ and $g_M$ and recall, from section \ref{Thegeneralisedwavegauge}, that the above is equivalent to the identity map \eqref{dick} being an $f$-wave map to first order.

We have therefore shown the following.
\begin{proposition}\label{proppuregaugesolution}
	Let $\qv$ be a smooth $\qm$ 1-form and let $\sv$ be a smooth $\sm$ 1-forms which satisfy
	\begin{align}
	\qbox\qv+\slap\qv-\frac{2}{r}\big(\qn\qv\big)_P+\frac{2}{r^2}\ud r\,\qv_P&=-\frac{1}{r}\big(\qn\otimes\qv_{l=0}\big)_P,\label{feelgrim1}\\
	\qbox\sv+\slap\sv-\frac{2}{r}\qn_{P}\sv+\frac{1}{r^2}(3-4\mu)\sv&=0\label{feelgrim2}.
	\end{align}
	Then the collection\footnote{This notation is borrowed from \cite{D--H--R}.}
	\begin{align*}
	\Ga=\bigg(\gabhatlin, \trgablin, \mling, \sghatlin, \trsglin\bigg)
	\end{align*}
	defined by 
	\begin{align*}
	\gabhatlin&=\qn\otimeshat\qv,\\
	\trgablin&=-2\qd\qv,\\
	\mling&=\sn\qv+\qexd\sv-\frac{2}{r}\ud r\otimes\sv,\\
	\sghatlin&=-2\sdst\sv,\\
	\trsglin&=2\sdiv\sv+\frac{4}{r}\qv_P,\\
	\tflin&=\qf,\\
	\sflin&=\slf
	\end{align*}
	is a smooth solution to the equations of linearised gravity. We such a solution \emph{\textbf{a pure gauge solution}} that is generated by $\qv$ and $\sv$.
\end{proposition}
One can of course verify the above explicitly from the equations \eqref{eqnlingrav1}-\eqref{eqnlingrav7}.

\section{The Regge--Wheeler and Zerilli equations and the gauge-invariant hierarchy}\label{TheRegge-WheelerandZerilliequationsandthegaugeinvariantheirarchy}

In this section we demonstrate the well-known (\cite{RW}, \cite{Z}) but nevertheless remarkable phenomena which sees certain \emph{gauge-invariant} quantities decouple from the full system of gravitational perturbations into the celebrated Regge--Wheeler and Zerilli equations.

This decoupling will serve as the catalyst to unlocking the tensorial structure of the equations of linearised gravity.

\subsection{The Regge--Wheeler and Zerilli equations}\label{TheRWandZeqns}

We begin by first introducing the pair of scalar wave equations on $\Mgs$ which describe the Regge--Wheeler and Zerilli equations respectively.

As these equations admit a formulation independent of the equations of linearised gravity we shall denote solutions to such equations without the superscript $(1)$.\newline

First to be defined is the Regge--Wheeler equation. 

In what follows, wee remind the reader of the space $\LM$ introduced in section \ref{Thel=0,1spheriacalharmonicsandqmtensors}.
\begin{definition}
	Let $\Phi\in\LM$ be a smooth function. Then we say that $\Phi$ is a smooth solution to the Regge--Wheeler equation on $\Mgs$ iff
	\begin{align}\label{eqnRWeqn}
	\qbox\Phi+\slap\Phi=-\frac{6}{r^2}\frac{M}{r}\Phi.
	\end{align}
\end{definition}

Next to be defined is the Zerilli equations. 

In what follows, we remind the reader of the operator $\szetap{p}$ introduced in section \ref{Theellipticoperatorszslap}.
\begin{definition}
Let $\Psi\in\LM$ be a smooth function. Then we say that $\Psi$ is a smooth solution to the Zerilli equation on $\mcalm$ iff
	\begin{align}\label{eqnZereqn}
	\qbox\Psi+\slap\Psi=-\frac{6}{r^2}\frac{M}{r}\Psi+\frac{24}{r^3}\frac{M}{r}(r-3M)\szetap{1}\Psi+\frac{72}{r^5}\frac{M}{r}\frac{M}{r}(r-2M)\szetap{2}\Psi.
	\end{align}
\end{definition}

\subsection{The connection to the system of gravitational perturbations}\label{Theconnectiontothesystemofgravitationalperturbations}

In this section we reveal the remarkable connection, as first discovered by Regge--Wheeler \cite{RW} and Zerilli \cite{Z}, between the Regge--Wheeler and Zerilli equations and the system of gravitational perturbations.\newline

We would like to acknowledge that the derivations throughout this section rely heavily on those found in the paper \cite{C--O--S} of Chaverra, Ortiz and Sarbach. See also section \ref{OVTheRegge--WheelerandZerilliequationsandthegauge-invarianthierarchy} of the overview for a further review of the literature.

\subsubsection{Gauge-invariant quantities}\label{Gauge-invariantquantities}

The correct way to reveal this connection is to consider certain quantities which \emph{vanish} for all the pure gauge solutions of section \ref{Specialsolutions2:Puregaugesolutions}. Isolating these \emph{gauge-invariant} quantities is thus the content of the following proposition. 

In what follows we recall the family of operators $\sps, \spv$ and $\spt$ introduced in section \ref{ThefamilyofoperatorsPandT}. 

\begin{proposition}\label{propdefngaugeinvariantquatities}
	Let $\So$ be a smooth solution to the equations of linearised gravity. We define
	\begin{align}\label{defnofhandeta}
	\hlin:&=r^4\sdiv\sdiv\sn\otimeshat\mling-r^2\qexd\bigg(r^2\sdiv\sdiv\sghatlin\bigg).
	\end{align}
	Then the following quantities are gauge-invariant:\newline
		\begin{align}
		\tauhatlin:&=\sps\gabhatlin-\qn\otimeshat\hlin\label{eqndefntauhat},\\
		\trtaulin:&=\sps\trgablin+2\qd\hlin\label{eqndefnqtrtau},\\
		\etalin:&=r^4\scurl\sdiv\sn\otimeshat\mling-r^2\qexd\bigg(r^2\scurl\sdiv\sghatlin\bigg),\\
		\sigmalin:&=\sps\trsglin-\frac{4}{r}\hlin_P-2r^4\slap\sdiv\sdiv\sghatlin.\label{eqndefnsigma}
		\end{align}
\end{proposition}

\begin{proof}
	Let $\Ga$ be the pure gauge solution, described as in Proposition \ref{proppuregaugesolution}, that is generated by the smooth $\qm$ 1-form $\qv$ and the smooth $\sm$ 1-form $\sv$. Then
	\begin{align*}
	\sdiv\sn\otimeshat\mling-
	\qn\sdiv\sghatlin=\sdiv\sn\otimeshat\sn\qv.
	\end{align*}
	Thus, for the solution $\Ga$
	\begin{align*}
	\hlin=\sps\qv,\qquad\etalin=0.
	\end{align*}
	Moreover, one computes:
	\begin{align*}
	\sps\gabhatlin&=\qn\otimeshat\hlin,\\
	\sps\trgablin&=-2\qd\hlin,\\
	\sps\trsglin&=2\sps\sdiv\sv+\frac{4}{r}\hlin_P,\qquad,\\
	\slap\sdiv\sdiv\sghatlin&=\slap\sdiv\sdiv\sn\otimeshat\sv.
	\end{align*}
	Applying the commutation formulae of Lemma \ref{lemmacommutingspoperators}, the proposition follows.
\end{proof}

\begin{remark}
	Observe that each of $\tauhatlin, \trtaulin, \etalin$ and $\sigmalin$ lie in the space $\LM$ and thus moreover vanish for all members of the linearised Kerr family (cf. Remark \ref{rmkkerrkernal}).
\end{remark}

\subsubsection{The gauge-invariant hierarchy}\label{Thegauge-invarianthierarchy}

To finally extract the Regge--Wheeler and Zerilli equations from the equations of linearised gravity requires embedding the gauge-invariant quantities of Proposition \ref{propdefngaugeinvariantquatities} into a hierarchy of such quantities. This (ultimately) leads to the following theorem.

\begin{theorem}\label{thmgaugeinvariantquantintermsofRWandZ}
	Let $\So$ be a smooth solution to the equations of linearised gravity. Then the following relations hold:
	\begin{align*}
	\tauhatlin&=\qn\otimeshat\qexd\Big(r\Psilin\Big)+6\mu\ud r\otimeshat\szetap{1}\qexd\Psilin,\\
	\trtaulin&=0,\\
	\etalin&=-\qhd\qexd\Big(r\Philin\Big),\\
	\sigmalin&=-2r\slap\Psilin+4\qn_{P}\Psilin+12\mu r^{-1}\ommu\szetap{1}\Psilin
	\end{align*}
	where $\Philin,\Psilin\in\LM$ are unique and satisfy the Regge--Wheeler and Zerilli equations respectively.
\end{theorem}

\begin{proof}
Exploiting the commutation formulae of Lemma \ref{lemmacommutingspoperators}	and imposing the system of gravitational perturbations  \eqref{eqnforgabhatfgauge}-\eqref{eqnfwavegauge2} yields\footnote{One could of course derive this by hand but the much more efficient approach is to perform the computations in the Regge--Wheeler gauge of section \ref{TheKerradaptedRWgauge} and then appeal to the gauge-invariance of all quantities under consideration.} the following system of wave equations
	\begin{align}
	\qbox\tauhatlin+\slap\tauhatlin+\frac{2}{r}\qn_{P}\tauhatlin-\frac{2}{r}P\cdot\qn\otimeshat\tauhatlin-\frac{4}{r}\frac{\mu}{r}\tauhatlin&=-\frac{1}{r}\ud r\otimeshat\qexd \sigmalin,\label{eqnwaveeqnfortauhat}\\
	\qbox\trtaulin+\slap\trtaulin&=0,\\
	\qbox\etalin+\slap\etalin-\frac{2}{r}\Big(\qn\etalin\Big)_P+\frac{2}{r^2}\ud r\text{ }\etalin_P&=0,\label{eqnwaveeqnforeta}\\
	\qbox \sigmalin+\slap \sigmalin+\frac{2}{r}\qn_{P}\sigmalin+\frac{2}{r^2}\sigmalin&=\frac{4}{r^2}\tauhatlin_{PP}\label{eqnwaveeqnforsigma}
	\end{align}
	along with the relations
	\begin{align}
	\qd\tauhatlin+\frac{1}{2}\qexd \sigmalin&=0,\label{eqndivtauhat}\\
	\trtaulin&=0\label{eqnfortrtau},\\
	\qd\etalin&=0\label{eqnforeta}.
	\end{align}
	In particular, $\etalin$ is co-closed. Arguing as in the Poincar\'e lemma for the simply connected manifold $\mcalq$, it thus  follows that there exists a unique, smooth function $\Philin\in\LM$ and a smooth function $f$ on $S^2$, with vanishing projection to $l=0,=1$, such that
	\begin{align*}
	\etalin=-\qhd\qexd\Big(r\Philin+f\Big).
	\end{align*}
	Inserting this ansatz into \eqref{eqnwaveeqnforeta} and integrating yields
	\begin{align}\label{eqnrwslaprwintermsofeta}
	\qbox\Philin+\slap\Philin+\frac{3}{r}\frac{\mu}{r}\Philin=\frac{1}{r}F-r\slap f-\frac{2}{r}f.
	\end{align}
	Here, $F$ is a smooth function on $S^2$ with vanishing projection to $l=0,1$. This uniquely fixes $f$ by demanding $
	r^2\slap f+2=F$.
	
	Conversely, contracting \eqref{eqndivtauhat} successively with $P$ and $\qhd\ud r$ results in the relations
	\begin{align}\label{wee}
	-2\qd \zetalin+r\qbox \sigmalin=0,\qquad
	\qexd \zetalin&=0
	\end{align}
	where $\zetalin:=\tauhatlin_P-\frac{r}{2}\qexd\sigmalin$. Consequently, contracting now \eqref{eqnwaveeqnfortauhat} with $P$ and utilising \eqref{eqnwaveeqnforsigma} yields
	\begin{align}\label{allow}
	-\frac{1}{r^2}\qexd\bigg(4r \zetalin_P-r^3\zslap\sigmalin\bigg)+2\slap \zetalin=0,\qquad-\frac{1}{r^2}\qexd\bigg(r^2\qd \zetalin+3M\sigmalin\bigg)+\slap \zetalin=0.
	\end{align}
	Now as $\zetalin$ is closed there exists a unique, smooth $\varphilin\in\LM$ and a smooth function $g$ on $S^2$ with vanishing projection to $l=0,1$ such that
	\begin{align}\label{eqnzetaintermsofvarphi}
	\zetalin=\qexd \big(\varphilin+g\big)
	\end{align}
	for which, by relations \eqref{allow}, the following decoupled equation must hold
	\begin{align}
	\zslap\qbox \varphilin+\slap\slap\varphilin+\frac{2}{r^2}\slap\varphilin-\frac{6\mu}{r^3}\qn_{P}\varphilin=0.\label{eqndecoupledeqnforvarphi}
	\end{align}
	Here, we have once more used an implicit integration constant to uniquely specify $g$.
	
	Thus, recalling that the operator $\szeta$ is a bijection on the space $\LM$, the unique function $\Psilin\in\LM$ defined by
	\begin{align}\label{mum}
	\szetap{1}\varphilin:=\Psilin
	\end{align}
	satisfies (using the commutation formulae of Lemma \ref{lemmacommformszeta})
	\begin{align*}
	\qbox\Psilin+\slap\Psilin+\frac{3}{r}\frac{\mu}{r}\Psilin-\frac{6}{r}\frac{\mu}{r}(2-3\mu)\szetap{1}\Psilin-18\frac{\mu}{r}\frac{\mu}{r}\ommu\szetap{2}\Psilin=0.
	\end{align*}
	
	To conclude the theorem we simply reverse our steps. Indeed, from \eqref{allow} we find
	\begin{align*}
	\sigmalin=-2r\slap\Psilin+4\qn_{P}\Psilin+12\mu r^{-1}\ommu\szetap{1}\Psilin.
	\end{align*}
	This combined with the definition of $\zetalin$ finally yields
	\begin{align*}
	\tauhatlin=\qn\otimeshat\qexd\Big(r\Psilin\Big)+6\mu\qexd r\otimeshat\szetap{1}\qexd\Psilin.
	\end{align*}
\end{proof}

We make the following remark.
\begin{remark}
	As observed in \cite{C--O--S}, one can exploit relations \eqref{wee}-\eqref{allow} to construct a solution to the Regge--Wheeler equation by considering the quantity
	\begin{align*}
	\lin{\phi}=-r^2\qd\zetalin.
	\end{align*}
One might then be tempted to avoid the use of the Zerilli equation altogether. The problem however, is that it is not clear how to construct the gauge invariant quantities $\tauhatlin$ and $\sigmalin$ from $\philin$ in a manner that is useful for the later analysis. This is manifestly not true if one instead introduces the Zerilli quantity $\Psilin$.
\end{remark}

We end this section with the following corollary, the proof of which is due to Moncrief \cite{Moncrief}.
\begin{corollary}\label{corrphipsigaugeinvariant}
	Let $\So$ be a smooth solution to the system of gravitational perturbations. Then $\Philin$ and $\Psilin$ are gauge-invariant.
\end{corollary}
\begin{proof}
We claim that:
\begin{align}\label{egg}
\slap\Philin+\frac{2}{r^2}\Philin=-r\qhd\qexd\bigg(\frac{1}{r^2}\etalin\bigg),\qquad \slap\Psilin=\frac{2}{r}\szetap{1}\tauhatlin_{{PP}}-\frac{1}{2}\frac{1}{r}\sigmalin-\qn_{P}\szetap{1}\sigmalin
\end{align}
from which the corollary would follow owing to the invertibility of the Laplacian over the space $\LM$.

Indeed, for the former one simply applies the curl operator $\qhd\qexd$ to $\etalin$. For the latter, we note the relations
\begin{align*}
\zetalin_P=\tauhatlin_{PP}-\frac{r}{2}\qn_{P}\sigmalin,\qquad4r\zetalin_P=r^3\zslap\sigmalin+2r^2\slap\varphilin.
\end{align*}
\end{proof}

\begin{remark}
Of course, one could use the relations \eqref{egg} to define the quantities $\Philin$ and $\Psilin$ thus forgoing the introduction of the quantities $\tauhatlin, \trtaulin, \etalin$ and $\sigmalin$. However the latter, now viewed simply as tensorial expressions in $\Philin$ and $\Psilin$, will in fact appear frequently throughout the text.
\end{remark}

\section{Initial data and well-posedness of linearised gravity}\label{Initialdataandwell-posednessoflinearisedgravity}

In this section we establish a well-posedness theorem for the equations of linearised gravity as a Cauchy initial value problem, thereby initiating their formal analysis.

\subsection{Initial data for linearised gravity}\label{Initialdataforlinearisedgravity}

We begin with a specification of Cauchy initial data for the equations of linearised gravity, a process which is complicated by the existence of \emph{constraints}.

\subsubsection{Admissible initial data sets for the equations of linearised gravity}\label{Admissibleinitialdatasetsforthesystemofgravitationalperturbations}

First we provide a definition of \emph{admissible} initial data for the equations of linearised gravity. Such data is to be defined on the initial Cauchy hypersurface $\Sigma$ of section \ref{TheCauchyhypersurfaceSigma}.
 
In what follows, we remind the reader of the contents of section \ref{TheCauchyhypersurfaceSigmaandsmitensoranalysis} on $\smi$-tensor analysis, in particular the decomposition of the second fundamental form of $\Sigma$, along with the definition of the lapse quantity $\bhlin$ and the future-pointing unit normal to $\Sigma$. 
\begin{definition}\label{defnadmissibleinitialdatasets}
We consider the following quantities:
\begin{itemize}
	\item the smooth functions $\llin, \bblin, \bhlin, \strhlin, \bklin$ and $\trklin$ on $\Sigma$
	\item the smooth $\smi$ 1-forms $\mblin, \mhlin$ and $\mklin$
	\item the smooth, symmetric, traceless 2-covariant $\smi$-tensor fields $\shhatlin$ and $\skhatlin$
\end{itemize}
Then we say that the collection
\begin{align*}
\mathscr{A}:=\bigg(\llin, \bblin, \mblin, \bhlin, \mhlin, \shhatlin, \strhlin, \bklin, \mklin, \skhatlin, \trklin\bigg)
\end{align*}
is an admissible initial data set for the equations of linearised gravity if the following identities hold:
\begin{align}
0&=\mathscr{H}\strhlin-\frac{2}{r^3}\snnu\bigg(\frac{1}{\bh}r^3\sdiv\mhlin\bigg)+\frac{2}{\bh}\slap\bhlin-\sdiv\sdiv\shhatlin-\frac{4}{r^2}\bh^2\snnu\bigg(\frac{r}{\bh^4}\bhlin\bigg)-3\bk\,\bklin+\frac{4}{3}\trk\trklin,\label{eqnconstraint1}\\
0&=\sdiv\mklin-\frac{2}{3}\snnu\trklin+\frac{1}{r^3}\snnu\bigg(r^3\bklin\bigg)+\frac{3}{4}\bk\snnu\strhlin-\frac{3}{2}\frac{\bk}{\bh}\sdiv\mhlin,\label{eqnconstraint2}\\
0&=\sdiv\skhatlin+\frac{1}{r^3}\snnu\Big(r^3\mklin\Big)-\frac{3}{2}\sn\bklin-\frac{3}{2}\frac{\bk}{\bh}\sn\bhlin-\frac{1}{2}\bigg(\frac{2}{3}\trk-\bk\bigg)\sdiv\shhatlin\label{eqnconstraint3}.
\end{align}
Here, $\mscrH$ is the operator defined by
\begin{align}\label{eqndefnH}
\mathscr{H}\cdot=\frac{1}{r^3}\snnu\Big(r^3\snnu \cdot\Big)+\frac{1}{2}\slap \cdot+\frac{1}{r^2}\bigg(1-\frac{r}{1+\mu}\snnu\bh\bigg)\cdot.
\end{align}
Given such an admissible data set, initial data on the hypersurface $\Sigma$ for the equations of linearised gravity is defined as follows.

\noindent For the symmetric, traceless 2-covariant $\qm$-tensor $\gabhatlin$
\begin{align*}
\gabhatlin_{nn}\Big|_\Sigma:&=-\frac{1}{N}\llin+\frac{1}{\bh}\bhlin,\\
\Big(\mcall_n\gabhatlin\Big)_{nn}\Big|_\Sigma:&=\uwlin_1,
\end{align*}
\begin{align*}
\gabhatlin_{n\nu}\Big|_\Sigma:&=\bblin-\frac{2\mu}{\bh}\bhlin,\\
:&=\bslin,\\
\Big(\mcall_n\gabhatlin\Big)_{n\nu}\Big|_\Sigma:&=\owlin
\end{align*}
and
\begin{align*}
\gabhatlin_{\nu\nu}\Big|_\Sigma:&=-\frac{1}{N}\llin+\frac{1}{\bh}\bhlin,\\
\Big(\mcall_n\gabhatlin\Big)_{\nu\nu}\Big|_\Sigma:&=2\bklin+\frac{2}{3}\trklin+\snnu\bslin-2\mu\bh\snnu\bigg(\frac{1}{\bh}\bhlin\bigg)+\frac{2}{\bh}\bigg(\bk+\frac{1}{3}\trk\bigg)\bhlin-\frac{1}{2}\uwlin_2.
\end{align*}
For the function $\trgablin$
\begin{align*}
\trgablin\Big|_\Sigma:&=\frac{2}{N}\llin+\frac{2}{\bh}\bhlin,\\
\Big(\mcall_n\trgablin\Big)\Big|_\Sigma:&=\uwlin_2.
\end{align*}
For the $\qmsm$ 1-form $\mling$
\begin{align*}
\mling_{n}\Big|_\Sigma:&=\bh\cdot\mblin-\frac{\mu}{\bh}\mhlin,\\
:&=\mslin,\\
\Big(\mcall_n\mling\Big)_{n}\Big|_\Sigma:&=\swlin
\end{align*}
and
\begin{align*}
\mling_{\nu}\Big|_\Sigma:&=\frac{1}{\bh}\mhlin,\\
\Big(\mcall_n\mling\Big)_\nu\Big|_\Sigma:&=2\mklin+\snnu\mslin+\sn\bslin-\frac{1}{r}\mslin-\frac{1}{2}\frac{1}{\bh}\bigg(\bk-\frac{2}{3}\trk\bigg)\mhlin.
\end{align*}
For the symmetric, traceless 2-covariant $\sm$-tensor $\sghatlin$
\begin{align*}
\sghatlin\Big|_\Sigma:&=\shhatlin,\\
\Big(\mcall_n\sghatlin\Big)\Big|_\Sigma:&=2\skhatlin-2\sdst\mslin.
\end{align*}
Finally, for the function $\trsglin$
\begin{align*}
\trsglin\Big|_\Sigma:&=\strhlin,\\
\Big(\mcall_n\trsglin\Big)\Big|_\Sigma:&=\frac{4}{3}\trklin-2\bklin+2\sdiv\mslin+\frac{1}{r}\frac{4}{\bh}\bslin-\frac{2}{N}\bigg(\bk-\frac{2}{3}\trk\bigg)\llin.
\end{align*}
Here, $\uwlin_1, \uwlin_2, \owlin$ and $\swlin$ are chosen such that the wave gauge conditions \eqref{eqnlingrav6} and \eqref{eqnlingrav7} are satisfied on $\Sigma$.
\end{definition}
\begin{remark}
The collection $\Soa$ correspond to linearised versions of the lapse, shift and various decompositions of the induced metric and second fundamental form of $\Sigma$. In particular, the coupled system of equations \eqref{eqnconstraint1}-\eqref{eqnconstraint3} arise as the linearisation of the Gauss--Codazzi equations.
Conversely, the four remaining constraints arise as a consequence of the divergence relation \eqref{eqnlorentzgauge}.
\end{remark}

\subsubsection{Seed data for linearised gravity}\label{Seeddataforlinearisedgravity}

In view of these constraints on initial data we provide in this section a notion of freely prescribed \emph{seed} data for the equations of linearised gravity. 

In what follows, we remind the reader of the space $\LS$ defined in section \ref{Thel=01sphericalharmonicsandqmsmitensors}.

\begin{definition}\label{defnseeddata}
A smooth seed data set for the equations of linearised gravity consists of prescribing:
\begin{itemize}
	\item four smooth functions $\oPhilin, \uPhilin, \oPsilin, \uPsilin\in\LS$
	\item four smooth functions $\fvulin, \fvblin, \fwulin$ and $\fwblin$ on $\Sigma$
	\item two smooth $\smi$ 1-forms $\fvslin$ and $\fwslin$
	\item a smooth function $\mathfrak{a}$ on the horizon sphere $\hplusi$ given as a linear combination of the $l=1$ spherical harmonics (cf. section \ref{Thel=0,1sphericalharmonicsandthesphericalharmonicdecomposition})
	\item a constant $\mfM$
\end{itemize}
We will denote such a smooth seed initial data set by the collection
\begin{align*}
\Seed.
\end{align*}
\end{definition}

\subsection{The well-posedness theorem}\label{Thewell-posednesstheorem}

In this section we establish the foundational well-posedness result for the equations of linearised gravity.

\begin{theorem}\label{thmwellposedness}
	Let $\mathscr{D}$ be a smooth seed data set on $\Sigma$. Then
	\begin{enumerate}[i)]
		\item there exists a unique extension of $\Sos$ to a smooth admissible initial data set $\Soa$ on $\Sigma$
		\item the initial data set $\Soa$ gives rise to a unique, smooth solution $\So$ to the equations of linearised gravity on $\mcalm\cap D^+(\Sigma)$.
	\end{enumerate}
\end{theorem}

\begin{proof}
The proof of the first part of the theorem is to be found in section \ref{AppendixConstructingadmissibleinitialdatafromseeddata} of the appendix.

To prove the second part of the theorem we begin by uniquely constructing the collection of tensors fields
\begin{align*}
\gabhatlin, \trgablin, \mling, \sghatlin, \trsglin
\end{align*}
on $D^+(\Sigma)$ where
\begin{itemize}
	\item $\gabhatlin$ is a smooth, symmetric, traceless 2-covariant $\qm$-tensor field
	\item $\trgablin$ is a smooth function
	\item $\mling$ is a smooth $\qm\otimes\sm$ 1-form
	\item $\sghatlin$ is a smooth, symmetric, traceless 2-covariant $\sm$-tensor field
	\item $\trsglin$ is a smooth function
\end{itemize}
by solving the following coupled system of (tensorial) wave equations with Cauchy data given by the initial data set $\Soa$ according to Definition \ref{defnadmissibleinitialdatasets}:
\begin{align}
\qbox\gabhatlin+\slap\gabhatlin+\frac{2}{r}\qn_{P}\gabhatlin-\frac{2}{r}\big(\qn\otimeshat\gabhatlin\big)_P-\frac{4}{r}\frac{\mu}{r}\gabhatlin&=\frac{1}{r}\ud r\qotimeshat\qexd\trgablin+\frac{2}{r}\ud r\qotimeshat\sdiv\mling-\frac{1}{r}\ud r\qotimeshat\qexd\trsglin-\qn\otimeshat\Big(r^{-1}\big({\gabhatlin_{l=0}}\big)_P\Big),\label{pork1}\\
\qbox\trgablin+\slap\trgablin&=-\frac{2}{r}\qd\Big(\big({\gabhatlin_{l=0}}\big)_P\Big)+\frac{2}{r^2}\big(\gabhatlin_{l=0}\big)_P\label{pork2}.
\end{align}
\begin{align}
\qbox\mling+\slap\mling-\frac{2}{r}\big(\qn{\otimes}\mling\big)_P-\frac{1}{r^2}\mling+\frac{2}{r^2}\ud r{\otimes}\mling_P=\frac{2}{r}\ud r\,{\otimes}\sdiv\sghatlin\label{pork3}.
\end{align}
\begin{align}
\qbox\sghatlin+\slap\sghatlin-\frac{2}{r}\qn_{P}\sghatlin-\frac{4}{r}\frac{\mu}{r}\sghatlin&=0,\label{pork4}\\
\qbox\trsglin+\slap\trsglin+\frac{2}{r}\qn_{P}\trsglin+\frac{2}{r^2}\trsglin&=\frac{4}{r^2}\gabhatlin_{P P}+\frac{2}{r^2}\trgablin-\frac{4}{r^2}\big(\gabhatlin_{l=0}\big)_{PP}\label{pork5}.
\end{align}
We continue by introducing the smooth $\qm$ 1-form $\qwlin$ and the smooth $\sm$ 1-form $\swlin$ defined as
\begin{align*}
\qwlin:&=-\qd\gabhatlin+\sdiv\mling-\frac{1}{2}\qexd\trsglin+\frac{1}{r}\big(\gabhatlin_{l=0}\big)_P,\\
\swlin:&=-\qd\mling-\frac{1}{2}\sn\trgablin+\sdiv\sghatlin.
\end{align*} 
It then follows from equations \eqref{pork1}-\eqref{pork5} that $\qwlin$ and $\swlin$ satisfy the tensorial system of wave equations 
\begin{align*}
\qbox\qwlin+\slap\qwlin+\frac{2}{r}\qn_{P}\qwlin-\frac{2}{r^2}\ud r\text{ }\qwlin_P&=\frac{2}{r}\ud r\text{ }\sdiv\swlin,\\
\qbox\swlin+\slap\swlin-\frac{1}{r^2}\swlin&=-\frac{2}{r}\sn\qwlin_P
\end{align*}
with \emph{trivial} Cauchy data. Indeed, that this data vanishes is a consequence of the initial data set $\Soa$ being admissible, and is by now a classical result. See for instance the book of Choquet-Bruhat \cite{C-Bbook}.

Consequently, by uniqueness it must be that $\qwlin$ and $\swlin$ vanish globally on $D^+(\Sigma)$. Thus, since the vanishing of these quantities is, by definition, equivalent to the validity of the wave gauge conditions \eqref{eqnlingrav6}-\eqref{eqnlingrav7},  it follows that the collection
\begin{align*}
\Solution
\end{align*}
defines the unique solution to the system of gravitational perturbations which agrees with the admissible initial data set $\Soa$.
\end{proof}

We make the following remarks.
\begin{remark}\label{rmkseedRW}
The construction of $\Soa$ in Theorem \ref{appendixthmconstraints} reveals that the seed quantities $\oPhilin, \uPhilin, \oPsilin$ and $\uPsilin$ determine Cauchy data for the gauge-invariant quantities $\Philin$ and $\Psilin$ respectively.
\end{remark}
\begin{remark}
See section \ref{OVTheCauchyproblemfortheequationsoflinearisedgravity} of the overview for comments regarding the parametrisation of the space of solutions to the constraint equations \eqref{eqnconstraint1}-\eqref{eqnconstraint3} by the seed data $\Sos$.
\end{remark}

\subsection{Pointwise strong asymptotic flatness}\label{Pointwisestrongasymptoticflatness}

In this final section we provide a notion of asymptotic flatness on the seed data for which the main Theorem of section \ref{Precisestatementsofthemaintheorems} will be most naturally formulated.\newline

In order to present this definition in the most concise manner possible requires a concise notation for collections of seed quantities along with higher order derivatives thereof. 

Indeed
\begin{itemize}
	\item the quantity $\opsilin$ denotes either of the pair $\Big(\oPhilin, \oPsilin\Big)$
	\item the quantity $\upsilin$ denotes either of the pair $\Big(\uPhilin, \uPsilin\Big)$
	\item the quantity $\bvlin$ denotes any of the triple $\Big(\fvulin, \fvblin, \fvslin\Big)$
	\item the quantity $\uvlin$ denotes any of the triple $\Big(\fwulin, \fwblin, \fwslin\Big)$
\end{itemize}

Moreover, we have the shorthand notation for higher order derivatives
\begin{align*}
|\mathfrak{D}^n S|:=\sum_{i+j=0}^{n}|(r\snnu)^i(r\sn)^jS|.
\end{align*}
Here, $S$ is any smooth $\smi$-tensor.

The regularity we are to impose on the seed is then defined below. 

In what follows, we remind the reader of the pointwise norm $\sg$ defined in section \ref{Smitensoranalysis}.

\begin{definition}\label{defnstrongasymptoticflatness}
Let $0<s\leq 1$ and $n\geq 2$. Then we say that a smooth seed data set $\Sos$ is \textbf{\emph{strongly asymptotically flat}} with weight $s$ to order $n$ provided that the seed quantities satisfy the following pointwise estimates on $\Sigma$:
\begin{align*}
\big|\mathfrak{D}^{n-2}\big(r^{\frac{1}{2}+s}\opsilin\big)\big|_{\sg}+\big|\mathfrak{D}^{n-3}\big(r^{\frac{3}{2}+s}\upsilin\big)\big|_{\sg}&\lesssim C_{n},\\
\big|\mathfrak{D}^{n+1}\big(r^{\frac{1}{2}+s}\bvlin\big)\big|_{\sg}+\big|\mathfrak{D}^{n}\big(r^{\frac{3}{2}+s}\uvlin\big)\big|_{\sg}&\lesssim C_{n}.
\end{align*}
Here, the $C_n$ are constants depending only on $n$. 
\end{definition}

We note the important property of this notion of pointwise asymptotic flatness in that it propagates under part i) of Theorem \ref{thmwellposedness}. See Proposition \ref{Apppropasymptoticflatness} in the Appendix. In regards to propagating under evolution, see Remark \ref{rmkdontpropagateasymptoticflatness}.

\section{Gauge-normalised solutions and identification of the Kerr parameters}\label{GaugenormalisedsolutionsandidentificationoftheKerrparameters}

In this section we consider certain solutions to the equations of linearised that are realised as the addition of a pure gauge solution of section \ref{Specialsolutions2:Puregaugesolutions} to a solution to the equations of linearised gravity arising under Theorem \ref{thmwellposedness}.

It is for these \emph{gauge-normalised} solutions that the quantitative boundedness and decay statements of section \ref{Precisestatementsofthemaintheorems} shall be formulated.

\subsection{The $\Ke_{\mathsmaller{\mfM, \mathfrak{a}}}$-adapted Regge--Wheeler gauge}\label{TheKerradaptedRWgauge}

The ``gauge'' we are to introduce is actually a family of gauges, a member of which is said\footnote{See section \ref{OVGaugenormalisedsolutionstotheequationsoflinearisedgravity} of the overview for a comment on the nomenclature.} to be a $\Ke_{\mathsmaller{\mfM, \mfa}}$-adapted Regge--Wheeler gauge.\newline

This gauge is so named as it is a modification of the Regge--Wheeler gauge employed by Regge and Wheeler in \cite{RW} adapted now to the equations of linearised gravity. See section \ref{OVGaugenormalisedsolutionstotheequationsoflinearisedgravity} of the overview for further discussion. 

\subsubsection{The projection onto and away from $l=0,1$}\label{Theprojectionontol=0andl=1}

We begin by defining the projection of a solution to the equations of linearised gravity onto and away from $l=0,1$, the former of which shall be used in the definition of the $\Ke_{\mathsmaller{\mfM, \mfa}}$-Regge--Wheeler gauge.

This section of the paper is lifed almost verbatim from section 9.5.1 in \cite{D--H--R}.\newline

First we provide a definition of a solution $\So$ having support only on $l=0,1$.
\begin{definition}\label{defnsolnstolingravsphericalharmdecomp}
	Let $\So$ be a smooth solution to the equations of linearised gravity.
	
	Then we say that $\So$ is supported only on $l=0,1$ iff all quantities in $\So$ are supported only on $l=0,1$ (cf. Definitions \ref{defnqmtensorssupportedonlgeq2} and \ref{defnqmsmoneformssupportedonlgeqi}). In particular, 
	\begin{align*}
	\sghatlin=0.
	\end{align*}
	
Conversely, we say that $\So$ has vanishing projection to $l=0,1$ iff all quantities in $\So$ have vanishing projection to $l=0,1$ (cf. Definitions \ref{defnqmtensorssupportedonlgeq2} and \ref{defnqmsmoneformssupportedonlgeqi}). 
\end{definition}
One then has the following proposition, which follows easily from linearity along with Propositions \ref{propsphericalharmonicdecompositionofqmtensors} and \ref{propqmoneformsmoneformsphericalharmonicdecomposition}.
\begin{proposition}\label{propdecompositionofsolution}
	Let $\So$ be a smooth solution to the equations of linearised gravity. Then one has the unique, orthogonal decomposition
	\begin{align*}
	\So=\So_{l=0,1}+\So'
	\end{align*}
	where $\So_{l=0,1}$
	is a  smooth solution to the equations of linearised gravity supported only on $l=0,1$ and $\So'$ is a smooth solution to the equations of linearised gravity that has vanishing projection to $l=0,1$.
\end{proposition}

Finally, the projection of $\So$ onto and away from $l=0,1$ is defined below.

\begin{definition}\label{defnprojectionl=0,1solution}
Let $\So$ be a smooth solution to the equations of linearised gravity. 

Then we call the map
\begin{align*}
\So\rightarrow\So_{l=0,1}
\end{align*}
the projection of $\So$ onto $l=0,1$.

Conversely, we call the map
\begin{align*}
\So\rightarrow\So'
\end{align*}
the projection of $\So$ away from $l=0,1$.
\end{definition}

\begin{remark}
Observe that the linearised Kerr solutions $\Ke_{\mathsmaller{\mfM, \mfa}}$ of Proposition \ref{propfullkerrfamily} are supported only on $l=0,1$. Thus $\Ke_{\mathsmaller{\mfM, \mfa}}'=0$.
\end{remark}

\subsubsection{The $\Ke_{\mathsmaller{\mfM, \mfa}}$-adapted Regge--Wheeler gauge}\label{TheKerradaptaedtracelessgauge}

We are now in a position to define (a member) of the $\Ke_{\mathsmaller{\mfM, \mfa}}$-adapted Regge--Wheeler gauge.

\begin{definition}\label{defnRWgauge}
	Let $\mfM\in\mathbb{R}$ and let $\mfa$ be a smooth function on $S^2$ supported only on $l=1$. Then a solution $\So$ to the equations of linearised gravity is said to be in the $\Ke_{\mathsmaller{\mfM, \mfa}}$-adapted Regge--Wheeler gauge iff the following conditions hold on $\Sigma$:
	\begin{itemize}
		\item the \textnormal{\textbf{propagation gauge conditions}:}
		\begin{align*}
		\sdiv\mling\Big|_{\Sigma}=\mcall_n\sdiv\mling\Big|_{\Sigma}&=0,\\
		\sghatlin\Big|_{\Sigma}=\mcall_n\sghatlin\Big|_{\Sigma}&=0
		\end{align*}
		\item the \textbf{{Kerr gauge condition}}:
		\begin{align*}
		\So_{l=0,1}\Big|_{\Sigma}=\Ke_{\mathsmaller{\mfM, \mathfrak{a}}}\Big|_{\Sigma}
		\end{align*}
	\end{itemize}
\end{definition}

We make the following remark.
\begin{remark}\label{rmkinitialdatagauge}
	Observe that, for a solution $\So$ arising from a seed data set $\Sos$ under Theorem \ref{thmwellposedness}, the conditions of Definition \ref{defnRWgauge} can be verified either from the seed or from the admissible initial data set constructed from it according to i) of said theorem. In particular, the parameters $\mfM$ and $\mfa$ are determined explicitly from the seed (see Definition \ref{defnseeddata}).
\end{remark}

\subsection{Achieving the $\Ke_{\mathsmaller{\mfM, \mathfrak{a}}}$-adapted Regge--Wheeler gauge: the initial-data-normalised solution $\SfMa$}\label{AchievingtheinitialdatanormalisationforageneralsolutionS}

In this section we show that, given a solution to the equations of linearised gravity arising under Theorem \ref{thmwellposedness}, one can indeed add to it a pure gauge solution for which the resulting solution is in a $\Ke_{\mathsmaller{\mfM, \mathfrak{a}}}$-adapted Regge--Wheeler gauge.\newline

In the following theorem statement
we shall employ the notation $\SosMa$ for a seed data set to explicitly reference the constant $\mfM$ and function $\mfa$ it contains as per Definition \ref{defnseeddata}.

\begin{theorem}\label{thminitialdatagauge}
	Let $\So$ be the smooth solution to the equations of linearised gravity arising from the smooth seed data set $\SosMa$ that is strongly asymptotically flat with weight $s$ to order $n\geq 2$. Then there exists a unique pure gauge solution $\GfMa$ such that the resulting (unique) solution
	\begin{align*}
	\SfMa:=\So+\GfMa
	\end{align*}
	is in the $\Ke_{\mathsmaller{\mfM, \mathfrak{a}}}$-adapted Regge--Wheeler gauge. Moreover, the seed data of $\Gf$ is strongly asymptotically flat with weight $s$ to order $n\geq 2$.
\end{theorem}

In view of Remark \ref{rmkinitialdatagauge} the solution $\SiMa$ constructed above is thus said to be \textnormal{\textbf{initial-data-normalised}}.\newline

The remainder of this section is concerned with the proof of Theorem \ref{thminitialdatagauge}.

\subsubsection{Parametrising the space of pure gauge solutions}\label{Kadapatedpuregaugesolutions}

In order to prove Theorem \ref{thminitialdatagauge} we first parametrise the class of pure gauge solutions to the equations of linearised gravity certain quantities $\Sigma$. 

Indeed the subsequent lemma, which we only state, when combined with Proposition \ref{proppuregaugesolution} establishes a bijection between the space of pure gauge solutions and the freely prescribed quantities appearing in the statement of the lemma. This bijection will thus ensure uniqueness of the pure gauge solutions constructed in the proof of Theorem \ref{thminitialdatagauge}.

\begin{lemma}\label{lemmaparametrisingpuregaugesolutions}
Let $\fwu, \fwo, \fvu, \fvo$ be four smooth functions on $\Sigma$ and let $\fws, \fvs$ be two smooth $\smi$ 1-forms. Then there exists a unique, smooth $\qm$ 1-form $\qv$ and a unique, smooth $\sm$ 1-form $\sv$ satisfying
\begin{align*}
\qbox\qv+\slap\qv-\frac{2}{r}\big(\qn\qv\big)_P+\frac{2}{r^2}\ud r\,\qv_P&=-\frac{1}{r}\big(\qn\otimes\qv_{l=0}\big)_P,\\
\qbox\sv+\slap\sv-\frac{2}{r}\qn_{P}\sv+\frac{1}{r^2}(3-4\mu)\sv&=0
\end{align*}
with
\begin{align*}
\qv_{n}\big|_{\Sigma}=\fvu,\qquad\qv_\nu\big|_{\Sigma}&=\fvo,\qquad\sv\big|_{\Sigma}=\fvs,\\
\Big(\qn_n\qv_{n}\Big)\Big|_{\Sigma}=\fwu,\qquad\Big(\qn_n\qv_\nu\Big)\Big|_{\Sigma}&=\fwo,\qquad\Big(\qn_n\sv\Big)\Big|_{\Sigma}=\fws.
\end{align*}
\end{lemma}

\subsubsection{Proof of Theorem \ref{thminitialdatagauge}}\label{Proofofthminitialdatagauge}

We now prove  Theorem \ref{Proofofthminitialdatagauge}.
\begin{proof}
We consider the seed data set for the solution $\So$:
\begin{align*}
\Seed.
\end{align*}
We make the following observations. 

First, from both the construction of the admissible initial data set $\Soa$ from $\Sos$ in accordance with Theorem \ref{appendixthmconstraints} of the Appendix combined with the construction of initial data for the collection $\So$ from $\Soa$ according to Definition \ref{defnadmissibleinitialdatasets} it follows that if the seed quantities $\fvulin, \fvblin, \fvslin, \fwulin, \fwblin, \fwslin$ all vanish then the solution $\So$ will satisfy the conditions of Definition \ref{defnRWgauge} with the parameters $\mfM$ and $\mfa$ determined from $\SosMa$.

Second, it follows once more from the proof of Theorem \ref{appendixthmconstraints} that the seed quantitites $\fvulin, \fvblin, \fvslin, \fwulin, \fwblin, \fwslin$ correspond to initial data for the generators $\qv$ and $\sv$ of a pure gauge solution $\Ga$ in accordance with Lemma \ref{lemmaparametrisingpuregaugesolutions}.

Combining these two observations, it thus suffices to define the pure gauge solution $\GfMa$ as the pure gauge solution that arises from Proposition \ref{proppuregaugesolution} combined with Lemma \ref{lemmaparametrisingpuregaugesolutions} in which the quantities $\fwu, \fwo, \fvu, \fvo, \fws, \fvs$ are defined according to
\begin{align*}
\fvu&=-\fvulin,\\
\fvo&=-\fvblin,\\
\fvs&=-\fvslin,\\
\fwu&=-\fwulin,\\
\fwo&=-\fwblin,\\
\fws&=-\fwslin.
\end{align*}
Finally, since by construction the seed data of $\GfMa$ is given by
\begin{align*}
\Sos_{\GfMa}=\bigg(0,0,0,0,-\fvulin,-\fvblin,-\fvslin,-\fwulin,-\fwblin,-\fwslin,0,0\bigg)
\end{align*}
the seed data set $\Sos_{\GfMa}$ thus inherits the regularity of $\Sos$.
\end{proof}

\subsection{Global properties of the solution $\SfMa$}\label{GlobalpropertiesofSf}

In this final section we prove certain global properties of the solution $\SfMa$ which will be fundamental in establishing the boundedness and decay statements of Theorem 1.

\begin{proposition}\label{propglobalproperties}
	Let $\SfMa$ be the unique initial-data-normalised solution to the equations of linearised gravity that arises in accordance with Theorem \ref{thminitialdatagauge}. Then the following hold on $D^+(\Sigma)$:
	\begin{itemize}
		\item the \textnormal{\textbf{propagation conditions propagate}:}
		\begin{align*}
		\sdiv\mling&=0,\\
		\sghatlin&=0
		\end{align*}
		\item the \textnormal{\textbf{Kerr condition propagates}:}
		\begin{align*}
		\mpart{\big(\SfMa\big)}=\Ke_{\mathsmaller{\mfM, \mathfrak{a}}}
		\end{align*}
	\end{itemize}
	Moreover, the linearised $\qm$-trace $\trgablin$ vanishes on $D^+(\Sigma)$:
	\begin{align*}
	\trgablin=0.
	\end{align*}
\end{proposition}
\begin{proof}
We have by the uniqueness criterion of Theorem \ref{thmwellposedness} that the only smooth solution to the system of gravitational perturbations for which the initial data is exactly that of the linearised Kerr solution $\Ke_{\mathsmaller{\mfM, \mfa}}$ is the linearised Kerr solution $\Ke_{\mathsmaller{\mfM, \mfa}}$. This yields the propagation of the Kerr gauge condition.

Moreover, we have from equation \eqref{eqnlingrav4} that the linearised metric quantity $\sghatlin$ associated to the solution $\SfMa$ satisfies the wave equation
\begin{align*}
\qbox\sghatlin+\slap\sghatlin-\frac{2}{r}\qn_{P}\sghatlin-\frac{4}{r}\frac{\mu}{r}\sghatlin=0
\end{align*}
with trivial Cauchy data. It therefore follows from a standard Gr\"onwall estimate that $\sghatlin$ vanishes globally. 

Consequently, turning now to equation \eqref{eqnlingrav3} we have that the linearised metric quantity $\sdiv\mling$ associated to the solution $\SfMa$ satisfies the wave equation
\begin{align*}
\qbox\sdiv\mling+\slap\sdiv\mling-\frac{2}{r}\big(\qn\sdiv\mling\big)_P+\frac{2}{r^2}\ud r\,\sdiv\mling_P=0
\end{align*}
with trivial Cauchy data and thus must vanish globally by Gr\"onwall.

Finally, combining the vanishing of $\sghatlin$ and $\sdiv\mling$ for the solution $\SfMa$ with the gauge condition \eqref{eqnlingrav7} yields
\begin{align*}
\slap\trgablin=0.
\end{align*}
It follows that $\trglin$ is spherically symmetric. However, since the Kerr gauge condition propagates, the spherically symmetric part of the solution $\SfMa$ is the linearised Kerr solution $\Ke_{\mathsmaller{\mfM, 0}}$ for which $\trglin=0$ by Proposition \ref{proplinkerrsoln}. This completes the proof of the proposition.
\end{proof}

We make the following remarks.
\begin{remark}
It is immediate from the proof that any solution that is in a $\Ke_{\mathsmaller{\mfM, \mfa}}$-adapted Regge--Wheeler gauge will also satisfy the conclusions of Proposition \ref{propglobalproperties}.
\end{remark}
\begin{remark}\label{rmkcurlggaugeinvariant}
We note that although the quantity $\scurl\mling$ associated to the solution $\SfMa$ satisifes the same decoupled equation as $\sdiv\mling$ one cannot set it to vanish since in the $\Ke_{\mathsmaller{\mfM, \mfa}}$-adapted Regge--Wheeler gauge $\scurl\mling$ is gauge-invariant.
\end{remark}

We also have the following Corollary to Proposition \ref{propglobalproperties} which we shall exploit in the proof of Theorem 1 of section \ref{Precisestatementsofthemaintheorems}.

In what follows, we remind the reader of the gauge-invariant quantities $\Philin$ and $\Psilin$ constructed in Theorem \ref{thmgaugeinvariantquantintermsofRWandZ} and the operator $\mcaltf$ defined in section \ref{ThefamilyofoperatorsPandT}.
\begin{corollary}\label{corrRWgauge}
Let $\SfMa$ be as in Proposition \ref{propglobalproperties}. Then the following relations hold on $D^+(\Sigma)$:
\begin{align*}
\mcaltf\gabhatlin&=\qn\otimeshat\qexd\Big(r\Psilin\Big)+6\mu\ud r\otimeshat\szetap{1}\qexd\Psilin,\\
r^4\scurl\sdiv\sn\otimeshat\mling&=-\qhd\qexd\Big(r\Philin\Big),\\
\mcaltf\trsglin&=-2r\slap\Psilin+4\qn_{P}\Psilin+12\mu r^{-1}\ommu\szetap{1}\Psilin.
\end{align*}
\end{corollary}
\begin{proof}
We recall from Theorem \ref{thmgaugeinvariantquantintermsofRWandZ} that associated to any solution $\So$ to the equations of linearised gravity the gauge-invariant quantities $\tauhatlin, \etalin$ and $\sigmalin$ of Proposition \ref{propdefngaugeinvariantquatities} satisfy
\begin{align*}
\tauhatlin&=\qn\otimeshat\qexd\Big(r\Psilin\Big)+6\mu\ud r\otimeshat\szetap{1}\qexd\Psilin,\\
\etalin&=-\qhd\qexd\Big(r\Philin\Big),\\
\sigmalin&=-2r\slap\Psilin+4\qn_{P}\Psilin+12\mu r^{-1}\ommu\szetap{1}\Psilin.
\end{align*}
However, since by Proposition \ref{propglobalproperties} the linearised metric quantities $\sghatlin$ and $\sdiv\mling$ vanish globally for the solution $\SfMa$ it follows by definition that $\SfMa$ 
\begin{align*}
\mcaltf\gabhatlin&=\tauhatlin,\\
r^4\scurl\sdiv\sn\otimeshat\mling&=\etalin,\\
\mcaltf\trsglin&=\sigmalin.
\end{align*}
\end{proof}
We end this section with the following interesting result, the truth of which is immediate from the above.
\begin{corollary}\label{corrPsiPhivanishimpliespuregauge}
	Let $\So$ be a smooth solution to the equations of linearised gravity arising from the smooth seed data $\SosMa$. Suppose further that the gauge-invariant quantities $\Psilin$ and $\Philin$ associated to $\So$ vanish. Then $\So$ is the sum of the pure gauge solution $\GfMa$ and the linearised Kerr solution $\Ke_{\mathsmaller{\mfM, \mfa}}$:
	\begin{align*}
	\Psilin=\Philin=0\implies\So=\GfMa+\Ke_{\mathsmaller{\mfM, \mathfrak{a}}}.
	\end{align*}
\end{corollary}

\section{Precise statements of the main theorems}\label{Precisestatementsofthemaintheorems}

In this section we finally give the precise statements regarding quantitative boundedness and decay for the initial-data-normalised solution $\SfMa$, to be measured in certain flux and integrated decay norms on $\mcalm$ as in Theorem 1. The remainder of the paper then concerns the proof of said theorem.\newline

We also include in this section a statement regarding quantitative boundedness and decay for solutions to the Regge--Wheeler and Zerilli equations on $\mcalm$, to be proved in our upcoming \cite{J}, the conclusions of which shall be fundamental for the proof of Theorem 1. 

\subsection{Flux and integrated decay norms}\label{Fluxandintegrateddecaynorms}

We begin by defining these flux and integrated decay norms by their action on smooth tensor fields on $\mcalm$.\newline

First we consider smooth functions $\psi$ on $\mcalm$. 

In what follows, we remind the reader of the function $\taus$ defined in section \ref{Thetaustarfoliation} and the norms on spheres defined in section \ref{Normsonspheres}.

We associate to $\psi$ the energy norm
\begin{align*}
\mathbb{E}[\psi](\taus)&:=\int_{2M}^R\Big(\snorm{S\psi}[\taus][r]+\snorm{L\psi}[\taus][r]+\snorm{\sn\psi}[\taus][r]\Big)\ud r+\int_{R}^\infty\Big(\snorm{L(r\psi)}[\taus][r]+\snorm{\sn(r\psi)}[\taus][r]\Big)\ud r
\end{align*}
and the (weighted) flux norms
\begin{align*}
\mathbb{F}[\psi]:&=\sup_{\taus\in[\taus_0, \infty)}\mathbb{E}[\psi](\taus)\\
+&\sup_{\taus\in(-\infty, \infty)}\int_{R}^\infty \Big(r^2\snorm{L(r\psi)}[\taus][r]+\snorm{\sn(r\psi)}[\taus][r]\Big)\ud r+\int_{-\infty}^{+\infty} \Big(\nisnorm{r(L-2S)(r\psi)}+\nisnorm{r\sn(r\psi)}\Big)\ud \taus
\end{align*}
where we define, for any $\taus$,
\begin{align*}
\|\cdot\|^2_{\mathsmaller{\iplus_{\taus}}}:=\lim_{r\rightarrow\infty}\snorm{\cdot}[\taus][r].
\end{align*}
We morever define the initial flux norms along the initial Cauchy hypersurface $\Sigma$,
\begin{align*}
\mathbb{D}[\psi]:=\int_{2M}^\infty r^2\Big(\isnorm{\snnu\psi}+\isnorm{\sn\psi}+\isnorm{n\psi}\Big)\ud r,
\end{align*}
We further associate to $\psi$ the integrated decay norm
\begin{align*}
\mathbb{M}[\psi]:=\int_{\taus_0}^\infty\int_{2M}^{\infty}\frac{1}{r^3}\Big(\snorm{S(r\psi)}[\taus][r]+\snorm{L(r\psi)}[\taus][r]+\snorm{\sn(r\psi)}[\taus][r]+\snorm{r\psi}[\taus][r]\Big)\ud \taus\ud r
\end{align*}
and the (weighted, degenerate) integrated decay norms
\begin{align*}
\mathbb{I}[\psi]:=&\int_{\taus_0}^\infty\int_{2M}^{R}(r-3M)^2\Big(\snorm{S\psi}[\taus][r]+\snorm{L\psi}[\taus][r]+\snorm{\sn\psi}[\taus][r]+\snorm{\psi}[\taus][r]\Big)\ud \taus\ud r\\
+&\int_{\tau^\star_0}^\infty\int_{R}^\infty\Big(r\snorm{L(r\psi)}[\taus][r]+r^{1+\beta_0-2}\snorm{\sn(r\psi)}[\taus][r]\Big)\ud \taus\ud r.
\end{align*}
Here, $\beta_0>0$ is a fixed constant such that $1-\beta_0<<1$. 

Higher order (weighted) flux norms are defined as follows:
\begin{align*}
\mathbb{F}^n[{\psi}]:=\sum_{i+j+k=0}^n\mathbb{F}[S^i(rL)^k(r\sn)^k\psi]
\end{align*}
with analagous definitions for the higher order energy and integrated decay norms. Conversely, for higher order versions of the initial flux norms we define
\begin{align*}
\mathbb{D}^n[{\psi}]:=\sum_{i+j=0}^n\mathbb{D}[(r\snnu)^i(r\sn)^k\psi].
\end{align*}
One also has the above when restricted to angular derivatives:
\begin{align*}
\mathbb{F}^{n, \sn}[{\psi}]:=\sum_{i=0}^n\mathbb{F}[(r\sn)^k\psi]
\end{align*}
with analagous definitions for the higher order energy and integrated decay norms.

Next we consider smooth tensor fields on $\mcalm$.

Indeed, with the following definition the norms\footnote{Generalising the $\mathbb{D}$ norms will be unecessary.} $\mathbb{E},\mathbb{F}, \mathbb{M}$ and $\mathbb{I}$, along with their higher order counter parts, easily generalise to smooth $\qm$-tensors, smooth $\qmsm$ 1-forms and smooth $\sm$-tensors.
\begin{definition}\label{defnnormderivativesofqmtensorfields}
  Let $n$ be a positive integer and suppose that $X\in\text{span}\{S,L\}$.
  
   Then if $Q$ is a smooth $2$-covariant $\qm$-tensor we define the action of $X$ on $Q$ in the norm $|\cdot|_\sg$ as
  \begin{align*}
  \big|X^nQ\big|_\sg^2:=\big|\big(\qn_X^nQ\big)_{TT}\big(\qn_X^nQ\big)_{TL}\big|+\big|\big(\qn_X^nQ\big)_{TL}\big|^2+\big|\big(\qn_X^nQ\big)_{TL}\big(\qn_X^nQ\big)_{LL}\big|+\big|\big(\qn_X^nQ\big)_{LL}\big|^2.
  \end{align*}
  Conversely, if $\theta$ is a smooth $2$-covariant $\sm$-tensor then we define the action of $X$ on $Q$ in the norm $|\cdot|_\sg$ as
  	\begin{align*}
  \big|X^n\theta\big|_\sg^2:=\big|\qn^n_X\theta\big|^2.
  \end{align*}
  Finally, if $\momega$ is a smooth $\qmsm$ 1-form then we define the action of $X$ on $\momega$ in the norm $|\cdot|_\sg$ as
  \begin{align*}
  \big|X^n\momega\big|_\sg^2:=\big|\big(\qn_X^n\momega\big)_{T}\big(\qn_X^n\momega\big)_{L}\big|+\big|\big(\qn_X^n\momega\big)_{L}\big|^2.
  \end{align*}
\end{definition}

Lastly, we introduce a concise notation to describe the action of the above defined norms now on (collections) of linearised quantities.

Indeed, if $\psilin=\{\psilin_1,...,\psilin_n\}$ denotes a collection of linearised quantities we define
\begin{align*}
\lin{\mathbb{F}}^n_{p}[\psi]:=\sum_{i=1}^{n}\mathbb{F}^n_{p}[\psilin_i]
\end{align*}
with analagous definitions for the norms $\mathbb{E}, \mathbb{M}, \mathbb{I}$ and $\mathbb{D}$. 

Finally, if $\So$ is a smooth solution to the equations of linearised gravity we define, for $p\in\mathbb{R}$,
\begin{align*}
\lin{\mathbb{F}}^n[r^p\So]:=\lin{\mathbb{F}}^n[r^p\gabhat]+\lin{\mathbb{F}}^n[r^p\trgab]+\lin{\mathbb{F}}^n[r^p\mg]+\lin{\mathbb{F}}^n[r^p\sghat]+\lin{\mathbb{F}}^n[r^p\strg]
\end{align*}
with analagous definitions for the remaining norms.

\subsection{Theorem \ref{mainthm1}: Boundedness and decay for the initial-data-normalised solution $\SfMa$}\label{BoundednessforsolutionsintheKerradaptedtracelessgauge}
 
In this section we state Theorem 1 concerning boundedness and decay for the solution $\SfMa$ in the (appropriate) norms of the previous section, the proof of which is the content of section \ref{Proofoftheorem1}. We also give several remarks and a statement of pointwise decay which follows as a corollary from Theorem 1.

\subsubsection{Statement of Theorem 1}\label{Statementoftheorem1}

The theorem statement is as given below.

In what follows, we remind the reader of the (gauge-invariant) quantities $\Philin$ and $\Psilin$ defined in section \ref{TheRegge-WheelerandZerilliequationsandthegaugeinvariantheirarchy}.
\begin{customthm}{1}\label{mainthm1}
	Let $\So$ be the smooth solution to the equations of linearised gravity arising from the smooth seed data set $\SosMa$. We assume that $\So$ is in the $\Ke_{\mathsmaller{\mfM, \mathfrak{a}}}$-adapted Regge--Wheeler gauge. 
	
	We further assume finiteness of the following initial flux norm on $\Sigma$:
	\begin{align*}
	\lin{\Delta}^5[\So]:=\linnorm{D}{r^{-1}\Phi, r^{-1}\Psi}[3].
	\end{align*}
	
	Then the norm propagates under evolution:
	\begin{align}\label{robbie}
	\linnorm{F}{r^{-1}\Phi, r^{-1}\Psi}[3]\lesssim \lin{\Delta}^5[\So].
	\end{align}
	Moreover, for the projection of $\So$ away from $l=0,1$,
	\begin{align*}
	\So'=\So-\Ke_{\mathsmaller{\mfM, \mathfrak{a}}}
	\end{align*}
 the initial flux norms in addition control weighted energies and integrated decay norms for up to 6 angular derivatives of the linearised metric:
	\begin{enumerate}[i)]
		\item the flux estimates
		\begin{align*}
		\linnorm{F}{r^{-\frac{3}{2}}\So'}[5,\sn]&\lesssim\lin{\Delta}^5[\So]
		\end{align*}
		\item the integrated decay estimates
		\begin{align*}
		\linnorm{M}{r^{-\frac{3}{2}}\So'}[4,\sn]+\linnorm{I}{r^{-\frac{3}{2}}\So'}[5,\sn]&\lesssim\lin{\Delta}^5[\So].
		\end{align*}
	\end{enumerate}
In addition, valid for any $\taus\geq \tau_0$, one has the decay estimates
\begin{align*}
\linnorm{E}{r^{-\frac{3}{2}}\Sf'}[3,\sn]\lesssim\frac{1}{{\taus}^2}\cdot\lin{\Delta}^5[\So'].
\end{align*}
\end{customthm}

\subsubsection{Remarks on Theorem 1 and uniform pointwise decay}\label{Remarksanduniformpointwiseboundedness}

We make the following remarks regarding Theorem 1.
\begin{remark}\label{reminitialdatanormalised}
If one assumes that the seed data $\SosMa$ is strongly asymptotically flat with weight $s$ to order $n\geq 8$ then the initial-data-normalised solution $\SfMa$ constructed from $\SosMa$ according to Theorem \ref{thminitialdatagauge} satisfies the assumptions and hence the conclusions of the theorem. In particular, the initial data norm is finite which can in fact be verified explicitly from the seed data. Moreover, the linearised Kerr solution which arises as the projection of $\So$ to $l=0,1$ is determined explicitly from the seed data $\SosMa$.
\end{remark}
\begin{remark}
With further work one can recover the missing derivatives, although the loss of a full derivative in the integrated decay norm $\mathbb{M}$ is unavoidable.  See section \ref{OVTheorem1:boundednessofthesolutionSi} and \ref{OVAside:thescalarwaveequationontheSchwarzschildexteriorspacetime} of the overview.
\end{remark}
\begin{remark}\label{rmkwastefulin}
We are quite wasteful with the $r$-weights in the Theorem statement. In fact, for any quantity other than $\mling$ associated to the solution $\So$ one can replace the weight $r^{-\frac{3}{2}}$ to $r^{-1}$. Of course, by Proposition \ref{propglobalproperties} $\trgablin=\sghatlin=0$.
\end{remark}
\begin{remark}\label{rmakTEa}
Observe that the quantities $\Philin$ and $\Psilin$ coincide for the solution $\So$ and $\So'$. One can thus replace $\So$ by $\So'$ on the right hand side of the estimates found in i) and ii).
\end{remark}

Moreover, a standard application of Sobolev embedding on 2-spheres yields as an immediate corollary of Theorem 1 the following uniform pointwise decay bounds on the solution $\So$.

In what follows, we remind the reader of the family of hypersurfaces $\Xi_{\taus}$ defined in section \ref{AfoliationofMthatpenetratesbothhplusandiplus} and the pointwise norm $\big|\cdot\big|_{\mathsmaller{\mcalm}}$ defined in section \ref{TheqmframeSL}. Moreover, for a solution $\So$ to the equations of linearised gravity we employ the shorthand notation 
\begin{align*}
\big|r^p\So\big|_{\mathsmaller{\mcalm}}:=\big|r^p\gabhatlin\big|_{\mathsmaller{\mcalm}}+\big|r^p\trgablin\big|_{\mathsmaller{\mcalm}}+\big|r^p\mling\big|_{\mathsmaller{\mcalm}}+\big|r^p\sghatlin\big|_{\mathsmaller{\mcalm}}+\big|r^p\trsglin\big|_{\mathsmaller{\mcalm}}.
\end{align*}
\\
\begin{corollary}\label{corrpointwisedecay}
	Let $\So$ and $\So'$ be as in the statement of Theorem 1. Then on the spacetime region $ D^+(\Sigma)$ one has the pointwise bounds
	\begin{align*}
	\big|r^{-\frac{1}{2}}\So'\big|_{\mathsmaller{\mcalm}}\lesssim\lin{\Delta}^5[\So'].
	\end{align*}
	Conversely, on the spacetime region $D^+\big(\Xi_{\taus_0}\big)$ one has the pointwise decay
	\begin{align*}
	\big|r^{-\frac{1}{2}}\Sf'\big|_{\mathsmaller{\mcalm}}\lesssim\frac{1}{\sqrt{\taus}}\cdot\lin{\Delta}^5[\So'].
	\end{align*}
	In particular, the solution $\SfMa$ decays at an inverse polynomial rate to the linearised Kerr solution $\Ke_{\mathsmaller{\mfM, \mfa}}$.
\end{corollary}
We make an additional remarks regarding Corollary \ref{corrpointwisedecay}.
\begin{remark}\label{rmkdontpropagateasymptoticflatness}
	Although the initial data norm propagates under evolution we do not propagate pointwise asymptotic flatness for the solution $\SfMa$. However, this can be rectified by modifying the choice of the linear map $Df\big|_{g_M}$ in section \ref{Theequationsoflinearisedgravity} and this modification shall be performed in our upcoming \cite{J}. 
	
	In fact, in this modified generalised wave gauge one can remove all $r$-weights in the norms of the Theorem statement for an appropriately identified initial-data-normalised solution. Since this procedure is slightly cumbersome however, for the purposes of this paper we prefer the simpler choice of gauge which we have utilised throughout.
\end{remark}

\subsection{Theorem 2: Boundedness and decay for solutions to the Regge--Wheeler and Zerilli equations on the Schwarzschild exterior spacetime}\label{BoundednessanddecayforsolutionstoequationsofRWtypeontheSchwarzschildexteriorspacetime}

In this final section we state Theorem 2 concerning a bounds for solutions to the Regge--Wheeler and Zerilli equations on $\mcalm$, the conclusions of which shall play a key role in the proof of Theorem 1. 

A proof of Theorem 2 can be found in the independent works \cite{Me} and \cite{H--K--W}. Nevertheless, we shall reprove it in our upcoming \cite{J}.\newline

The theorem statement is as follows.

We remind the reader of the space $\LM$ defined in section \ref{Thel=0,1spheriacalharmonicsandqmtensors} and the operator $\szetap{p}$ of section \ref{Theellipticoperatorszslap}.
\begin{customthm}{2}\label{mainthm3}
	Let $\Phi, \Psi\in\LM$ be smooth solutions to the Regge--Wheeler and Zerilli equations respectively:
		\begin{align*}
		\qbox\Phi+\slap\Phi=-\frac{6}{r^2}\frac{M}{r}\Phi,\qquad\qbox\Psi+\slap\Psi=-\frac{6}{r^2}\frac{M}{r}\Psi+\frac{24}{r^3}\frac{M}{r}(r-3M)\szetap{1}\Psi+\frac{72}{r^5}\frac{M}{r}\frac{M}{r}(r-2M)\szetap{2}\Psi.
		\end{align*}
	We assume finiteness of the initial flux norms
	\begin{align*}
	\norm{D}{r^{-1}\Phi, r^{-1}\Psi}[5].
	\end{align*}
	Then the following estimates hold.
	\begin{enumerate}[i)]
		\item the flux estimates 
		\begin{align*}
		\norm{F}{r^{-1}\Phi, r^{-1}\Psi}[5]&\lesssim\norm{D}{r^{-1}\Phi, r^{-1}\Psi}[5]
		\end{align*}
		\item the integrated decay estimates
		\begin{align*}
		\norm{M}{r^{-1}\Phi, r^{-1}\Psi}[4]+\norm{I}{r^{-1}\Phi, r^{-1}\Psi}[5]&\lesssim\norm{D}{r^{-1}\Phi, r^{-1}\Psi}[5]
		\end{align*}
	\end{enumerate}
	In addition, for any $\taus\geq\taus_0$, one has the decay estimates
	\begin{align*}
	\norm{E}{r^{-1}\Phi, r^{-1}\Psi}[3](\tau^\star)\lesssim\frac{1}{{\tau^\star}^2}\cdot\norm{D}{r^{-1}\Phi, r^{-1}\Psi}[5].
	\end{align*}
\end{customthm}

\section{Proof of Theorem 1}\label{Proofoftheorem1}

In this section we prove Theorem 1.

The proof in fact essentially follows from Theorem 2 combined with Proposition \ref{propglobalproperties} and Corollary \ref{corrRWgauge}.\newline

Throughout this section all linearised quantities are implicitly understood to be derived from the solution $\So'$ of the equations of linearised gravity in Theorem 1.

\subsection{Decay for the gauge-invariant quantities $\Philin$ and $\Psilin$}\label{DecayforthegaugeinvariantquantitiesPhiandPSi}

In this section we apply Theorem 2 to the gauge-invariant quantities $\Philin$ and $\Psilin$, which is applicable courtesy of of Theorem \ref{thmgaugeinvariantquantintermsofRWandZ}. This immediately yields the following proposition.
\begin{proposition}\label{propascendingthehierarchy}
Let $\So'$ be as in the statement of Theorem 1. 
Then the quantities $\Philin$ and $\Psilin$ satisfy the assumptions and hence the conclusions of Theorem 2.
\end{proposition}

\subsection{Completing the proof of Theorem 1}\label{CompletingtheproofofTheorem1}
 
 In this section we complete the proof of Theorem 1.
 
 \begin{proof}[Proof of Theorem 1]
According to Proposition \ref{propglobalproperties} and Corollary \ref{corrRWgauge} the solution $\So'$ to the equations of linearised gravity satisfies
\begin{align}
\mcaltf\gabhatlin&=\qn\otimeshat\qexd\Big(r\Psilin\Big)+6\mu\ud r\otimeshat\szetap{1}\qexd\Psilin,\label{x}\\
\mcaltv\mling&=-\sdso\Big(0,\qhd\qexd\Big(r\Philin\Big)\Big),\label{y}\\
\mcaltf\trsglin&=-2r\slap\Psilin+4\qn_{P}\Psilin+12\mu r^{-1}\ommu\szetap{1}\Psilin\label{z}
\end{align}
with $\trgablin=\sghatlin=0$. The estimates in the statement of Theorem 1 then follow from dilligently commuting and evaluating the tensorial expressions on the right hand sides of \eqref{x}-\eqref{z} in the frame $\{S,L\}$ of section \ref{TheqmframeSL}, keeping careful track of $r$-weights, and then applying Proposition \ref{propascendingthehierarchy} to control flux and integrated decay estimates of sufficiently many derivatives of $\Philin$ and $\Psilin$. Since this is rather cumbersome in practice we only note the key points (with the collective notation scheme under use being self-evident):
\begin{itemize}
	\item Proposition \ref{propSLframe} allows one to perfom all necessary computations in the $\{S,L\}$ frame. In particular, noting that we require commuting with $r^2\qn_L$ at most twice and that on $D^+(\Sigma)$ one has the bounds $|\kappa|\leq 1, |\frac{\kappa}{1-\mu}|\lesssim C$ for some constant $C$, the (smooth) connection coefficients in this frame are of order $O(r^{-2})$ and hence play no role when evaluating the (commuted) tensorial expressions 
	\item by definition of the flux and integrated decay norms the derivatives $\qn_L$ and $\sn$ always appear with an additional $r$-weight, thus gaining in regularity towards $\iplus$
	\item to bound the (commuted) terms involving the operator $\szetap{1}$ one applies the commutation relations of Lemma \ref{lemmacommformszeta} along with the estimates of Proposition \ref{propellipticestimateszlsap}
	\item to translate the bounds on the collection $\mcalt\So'$ to the solution $\So'$ one uses the elliptic estimates of Proposition \ref{propellipticestimatesonT} (recalling that the solution $\So'$ has vanishing projection to $l=0,1$)
	\item to control higher order derivatives one commutes the collection $\mcalt\So'$ with $\qn$ and $\slashed{\mcala}$, using the commutation relations of section \ref{Commutationformulaeandusefulidentities} for the former, and then applies the elliptic estimates of Proposition \ref{propellipticestimatesonT}
	\item by definition of the flux and integrated decay norms on $\qm$-tensors and $\qmsm$ 1-forms a contraction with $S$ always appears alongside a contraction with $L$ thus ensuring\footnote{Observe that this fact was made redundant in the proof of Proposition \ref{propboundsongammani} by the presence of the operator $\qdL$.} the $\qn$-derivatives appearing on the right had side of the above family of expression always contain at least one $\qn_L$ derivative
\end{itemize}
\end{proof}

\appendix

\section{Initial data for the equations of linearised gravity}\label{Appendixinitialdataforthesystemofgravitationalperturbations}

In this section of the Appendix we consider the construction and regularity of initial data for the equations of linearised gravity.

\subsection{Constructing admissible initial data from seed data }\label{AppendixConstructingadmissibleinitialdatafromseeddata}

We begin in this section by constructing admissible initial data for the equations of linearised gravity from freely prescribed seed data, thereby establishing part i) of Theorem \ref{thmwellposedness} in section \ref{Thewell-posednesstheorem}.\newline

The proof of the following theorem exploits the existence of three explicit classes of solutions to the equations of linearised gravity each of which must necessarily generate three classes of admissible initial data.

\begin{theorem}\label{appendixthmconstraints}
	Let $\Sos$ be a smooth seed data set for the equations of linearised gravity. Then there exists a unique extension of $\Sos$ to a smooth admissible initial data set $\Soa$.
\end{theorem}

\begin{proof}
Let $\Sos$ be the seed data set in question,
\begin{align*}
\Seed.
\end{align*}
We proceed in three steps.
\begin{enumerate}
	\item From the subset of seed $\Big(\oPsilin, \oPhilin, \uPsilin, \uPhilin\Big)$ lying in the space $\LS$ we construct the unique, smooth functions $\Philin,\Psilin\in\LM$ according to
	\begin{align}\label{tired1}
	\qbox\Phi+\slap\Phi=-\frac{6}{r^2}\frac{M}{r}\Phi,\qquad\qbox\Psi+\slap\Psi=-\frac{6}{r^2}\frac{M}{r}\Psi+\frac{24}{r^3}\frac{M}{r}(r-3M)\szetap{1}\Psi+\frac{72}{r^5}\frac{M}{r}\frac{M}{r}(r-2M)\szetap{2}\Psi.
	\end{align}
	with
	\begin{align}\label{tired2}
	\Big(\Philin, n\Philin\Big)\Big|_{\Sigma}=\Big(\oPhilin, \uPhilin\Big),\qquad\qquad \Big(\Psilin, n\Psilin\Big)\Big|_{\Sigma}=\Big(\oPsilin, \uPsilin\Big).
	\end{align}
	This in turn, when combined with Corollary \ref{corrinvertingToperators}, (uniquely) determines the smooth solution 
	$\So$ to the equations of linearised gravity, which has vanishing projection to $l=0,1$, according to:
	\begin{align*}
	\mcaltf\gablin&=\tauhatlin,\\
	\mcaltv\mling&=\sdso\Big(0, \etalin\Big),\\
	\mcaltt\sghatlin&=0,\\
	\mcaltf\trsglin&=\sigmalin,\\
	\mcaltf\tflin&=\frac{2}{r}\tauhatlin_P-\frac{1}{r}\ud r\text{ }\sigmalin,\\
	\mcaltv\sflin&=\frac{2}{r}\sdso\Big(0,\etalin_P\Big).
	\end{align*}
	Indeed, that the above collection determines a solution to the equations of linearised gravity is immediate from the existence of the Regge--Wheeler gauge in section \ref{TheKerradaptedRWgauge} and the invertibility of the family of operators $\mcalt$ over the space of tensor fields on $\mcalm$ having vanishing projection to $l=0,1$.
	
	We subsequently denote by $\Soa_{\mathsmaller{\oPhilin, \oPsilin, \uPhilin, \uPsilin}}$ the corresponding admissible initial data set for the equations of linearised gravity determined from the collection $\So$ according to Definition \ref{defnadmissibleinitialdatasets} and the rules \eqref{tired1}-\eqref{tired2} for projecting normal derivatives onto $\Sigma$.
	\item From the subset of seed $\Big(\mfM, \mfa\Big)$ we (uniquely) construct the linearised Kerr solution $\Ke_{\mathsmaller{\mfM, \mathfrak{a}}}$ according to Proposition \ref{propfullkerrfamily}. 
	
	We subsequently denote by $\Soa_{\mathsmaller{\mfM, \mfa}}$ the corresponding admissible initial data set for the equations of linearised gravity determined from $\Ke_{\mathsmaller{\mfM, \mathfrak{a}}}$ according to Definition \ref{defnadmissibleinitialdatasets}.
	\item From the subset of seed $\Big(\fvulin, \fvblin, \fvslin, \fwulin, \fwblin, \fwslin\Big)$ we (uniquely) construct the smooth $\qm$ 1-form $\qv$ and the smooth $\sm$ 1-form $\sv$ according to
	\begin{align}\label{fish1}
	\qbox\qv+\slap\qv-\frac{2}{r}\big(\qn\qv\big)_P+\frac{2}{r^2}\ud r\,\qv_P&=-\frac{1}{r}\big(\qn\otimes\qv_{l=0}\big)_P,\qquad\qquad
	\qbox\sv+\slap\sv-\frac{2}{r}\qn_{P}\sv+\frac{1}{r^2}(3-4\mu)\sv&=0
	\end{align}
	with
	\begin{align}\label{fish2}
	\Big(\qv_n, \qv_\nu, (\qn_n\qv)_n, (\qn_n\qv)_\nu\Big)\Big|_{\Sigma}=\Big(\fvulin, \fvblin, \fwulin, \fwblin\Big),\qquad\qquad \Big(\sv, \qn_n\sv\Big)\Big|_{\Sigma}=\Big(\fvslin, \fwslin\Big)
	\end{align}
	This in turn (uniquely) determines the pure gauge solution 
	$\Ga$ generated by $\qv$ and $\sv$ according to Proposition \ref{proppuregaugesolution}.
	
	We subsequently denote by $\Soa_{\mathsmaller{\fvulin, \fvblin, \fvslin, \fwulin, \fwblin, \fwslin}}$ the corresponding admissible initial data for the equations of linearised gravity determined from $\Ga$ according to Definition \ref{defnadmissibleinitialdatasets} and the rules \eqref{fish1}-\eqref{fish2} for projecting normal derivatives onto $\Sigma$.
\end{enumerate}
Hence, from the full seed $\Sos$, we now determine the collection
	\begin{align*}
	\Soa:=\Soa_{\mathsmaller{\oPhilin, \oPsilin, \uPhilin, \uPsilin}}+\Soa_{\mathsmaller{\mfM, \mfa}}+\Soa_{\mathsmaller{\fvulin, \fvblin, \fvslin, \fwulin, \fwblin, \fwslin}}.
	\end{align*}
	which, by linearity and steps 1.-3., thus determines an admissible initial data set for the system of gravitational perturbations constructed uniquely from $\Sos$. 
\end{proof}
We make the following remarks.
\begin{remark}
We emphasize that the admissible initial data set constructed above is determined from expressions in the seed quantities which can be written down explicitly. However, as this is quite cumbersome, we prefer the more implicit approach presented above.
\end{remark}
\begin{remark}
Observe that imposing compact support on the seed quantities of steps 1. and 3. generates an admissible initial data set $\Soa$ for which $\Soa-\Soa_{\mathsmaller{\mfM, \mfa}}$ is compactly supported. Here, $\Soa_{\mathsmaller{\mfM, \mfa}}$ is the Kerr initial data set constructed in step 2.
\end{remark}

\subsection{Propagation of strong asymptotic flatness}\label{AppPropagationofasymptoticflatness}

We continue in this section by proving that the assumption of strong pointwise asymptotic flatness on the seed data actually propagates under Theorem \ref{appendixthmconstraints}.\newline

To state the result concisely will require introducing a collective notation for quantities on $\Sigma$ which are computable explicitly from the initial data set $\Soa$.

Indeed, given the admissible initial data set $\Soa$ arising from the seed data set $\Sos$ we define
\begin{align*}
\Soa':=\Soa-\Soa_{\mathsmaller{\mfM, \mfa}}
\end{align*}
where $\Soa_{\mathsmaller{\mfM, \mfa}}$ is the admissible data set for the linearised Kerr solution as in the proof of Theorem \ref{appendixthmconstraints}. We then employ the collective notation:
\begin{itemize}
	\item the quantity $\bglin'$ denotes any of the tuple $\Big(\Nlin', \bblin', \mblin', \bhlin', \mhlin', \shhatlin', \strhlin'\Big)$
	\item the quantity $\uglin'$ denotes any of the tuple \\ $\Big(\bklin', \mklin', \skhatlin', \trklin'\Big)$
\end{itemize}
The proposition statement is then as follows.
\begin{proposition}\label{Apppropasymptoticflatness}
Let $\Sos$ be a smooth seed data set for the equations of linearised gravity which is strongly asymptotically flat with weight $s$ to order $n\geq 7$. Then the collection $\Soa'$ defined above
satisfy the following pointwise bounds on $\Sigma$:
\begin{align*}
\big|\mathfrak{D}^{n-4}\big(r^{\frac{3}{2}+s}\bgprimelin\big)\big|_{\sg}+\big|\mathfrak{D}^{n-5}\big(r^{\frac{5}{2}+s}\ugprimelin\big)\big|_{\sg}&\lesssim \mathcal{C}_{n}.
\end{align*}
Here, the $\mcalc_{n}$ are constants depending only on the constants $C_n$ appearing in Definition \ref{defnstrongasymptoticflatness}.
\end{proposition}

\begin{proof}
The proof is more cumbersome than difficult and involves dilligently commuting and evaluating the expressions for the metric quantities as given from Theorem \ref{appendixthmconstraints} which we observe, modulo the family of operators $\mcalt$, are explicit expressions in the seed quantities. For this reason we will only describe the main steps. Moreover, we focus only on the bounds for $\bglin$ since for the other quantities the proof follows in a similar fashion.

Indeed, from the construction of the admissible initial data in the proof of Theorem \ref{appendixthmconstraints} one derives the estimates, for any $i,j\geq 0$ and any $r\in[2M, \infty)$: 
\begin{align*}
\isnorm{(r\snnu)^{i}\mcala^{[j]}\mcalt\big(r^{\frac{1}{2}+s}\bglin'\big)}\lesssim&\isnorm{\mathfrak{D}^{\max\{i,j\}+2}\big(r^{\frac{1}{2}+s}\opsilin\big)}+\isnorm{\mathfrak{D}^{\max\{i,j\}+1}\big(r^{\frac{3}{2}+s}\upsilin\big)}.
\end{align*}
Here, we employ the shorthand notation
\begin{align*}
\isnorm{\mathfrak{D}^nS}=\sum_{i+j=0}^{n}\isnorm{(r\snnu)^i(r\sn)^jS}.
\end{align*}
Hence, since $\snnu$ commutes trivially with $\mcala^{[i]}$ and $\mcalt$ one has by elliptic estimates
\begin{align*}
\mathsmaller{\sum_{i=0}^{4}}\isnorm{(r\sn)^i\mathfrak{D}^n\big(r^{\frac{1}{2}+s}\bglin'\big)}\lesssim&\isnorm{\mathfrak{D}^{n+2}\big(r^{\frac{1}{2}+s}\opsilin\big)}+\isnorm{\mathfrak{D}^{n+1}\big(r^{\frac{3}{2}+s}\upsilin\big)}.
\end{align*}
Applying Sobolev embedding on 2-spheres then yields
\begin{align*}
\big|\mathfrak{D}^n\big(r^{\frac{1}{2}+s}\bglin'\big)\big|_{\sg}\lesssim C\cdot C_{n+4}
\end{align*}
where the constant $C$ is computable explicitly from the right hand side of the above and $C_{n+4}$ is the constant appearing in Definition \ref{defnadmissibleinitialdatasets}.
\end{proof}

We make a remark on the derivative loss.
\begin{remark}
Observe thus that in the construction one loses four derivatives at the top order in these supremum norms when passing from the seed data to the full initial data. This loss occurs as a result of the construction in step 1. of Theorem \ref{appendixthmconstraints} which requires inverting the family of angular operators $\mcalt$ on the 2-spheres and can potentially be improved -- see the discussion at the end of section \ref{OVTheCauchyproblemfortheequationsoflinearisedgravity} in the overview. In particular, observe that at the top order in angular derivatives one only loses two derivatives, with this loss a consequence of the Sobolev embedding.
\end{remark}

A corollary of Proposition \ref{Apppropasymptoticflatness} is now pointwise bounds on the projection of all quantities in the collections $\So'$ to $\Sigma$. 

To state the corollary concisely we utilise one final time a collective notation defined as follows:
\begin{itemize}
	\item the quantity $\mfglin'$ denotes the projection of any linearised metric quantity associated to the solution $\So'$ onto $\Sigma$
	\item the quantity $n(\mfglin')$ denotes the projection of the normal derivative $\qn_n$ of any linearised metric quantity associated to the solution $\So'$ onto $\Sigma$  
\end{itemize}
The collective notation we are to employ in the statement of the corollary is at this stage self-evident.

\begin{corollary}\label{Appcorrfinitenessinitialdatanorm}
Let $\So$ be a smooth solution to the equations of linearised gravity arising from a smooth, seed data set $\Sos$ that is strongly asymptotically flat with weight $s$ to order $n\geq 7$. Then the collections $\So'$ satisfy the following pointwise bounds on $\Sigma$:
\begin{align}
\big|\mathfrak{D}^{n-4}\big(r^{\frac{3}{2}+s}\mfglin'\big)\big|_{\sg}+\big|\mathfrak{D}^{n-5}\big(r^{\frac{5}{2}+s}n(\mfglin')\big)\big|_{\sg}&\lesssim \mathcal{C}'_{n},\label{fuck1}.
\end{align}
Here, the constants $\mcalc_{n}'$ are computable explicitly from the constants $C_n$ and $\mcalc_{n}$.
\end{corollary}
\begin{proof}
The proof, which we leave implicit, proceeds by carefully applying the estimates of Proposition \ref{Apppropasymptoticflatness} in conjuction with the construction of initial data from the admissible initial data sets according to Definition \ref{defnadmissibleinitialdatasets}, noting in particular that the pointwise bounds on the forcing gauge terms are prescribed explicitly in the definition of strong pointwise asymptotic flatness.
\end{proof}


\begin{thebibliography}{9}
\bibitem{D--H--R}
M. Dafermos, G. Holzegel and I. Rodnianski,
\emph{The linear stability of the {S}chwarzschild solution to gravitational perturbations}, arXiv:1601.06467 (2016).
\bibitem{D--H--Rscattering}
-----, \emph{A scattering theory construction of dynamical black hole spacetimes}, arXiv:1306.5534 (2013), to appear in J. Diff. Geom.
\bibitem{L--R}
H. Lindblad and I. Rodnianksi, \emph{The global stability of Minkowski space-time in the harmonic gauge}, Ann. of Math., 171 (2010), pp. 1401-1477.
\bibitem{C--K}
D. Christodoulou and S. Klainerman,
\emph{The Global Nonlinear Stability of the Minkowski Space}, Princeton University Press.
\bibitem{J}
T. Johnson, \emph{The linear stability of the Schwarzschild solution to gravitational perturbations in the generalised wave gauge}, Imperial College London PhD Thesis, In Preparation.
\bibitem{K}
R. Kerr, \emph{Gravitational field of a spinning mass as an example of algebraically special metrics}, Phys. Rev. Lett, 11 (5) (1963), pp 237-238.
\bibitem{Aretakis}
S. Aretakis, \emph{Horizon instability of extremal black holes}, Adv. Theor. Math. Phys., 19 (2015), pp. 507-530.
\bibitem{D--L}
M Dafermos and J. Luk, \emph{The interior of dynamical vacuum black holes I: The $C^0$-stability of the Kerr Cauchy horizon}, arXiv:1710.01722 (2017).
\bibitem{C-B}
Y. Choquet-Bruhat, \emph{Th\'eor\`eme d'existence pour certains syst\`emes d'\'equations aux d\'eriv\'ees partielles non lin\'eaires}, Acta math., 88 (1952), pp. 141-225.
\bibitem{C-Bbook}
-----, \emph{General Relativity and the Einstein Equations}, Oxford University Press, 2008.
\bibitem{L--T}
H. Lindblad and M. Taylor, \emph{Global stability of Minkowski space for the Einstein--Vlasov system in the harmonic gauge}, arXiv:1707.06079 (2017).
\bibitem{Sp}
J. Speck, \emph{The global stability of the Minkowski spacetime solution to the Einstein-nonlinear electromagnetic system in wave coordinates},
Anal. PDE, 7 (4) (2014), pp. 771-901.
\bibitem{LeF--M}
P. LeFloch and Y. Ma, \emph{The global nonlinear stability of Minkowski space for self-gravitating massive fields}, arXiv:1511.03324 (2015).
\bibitem{F--J--S}
D. Fajman, J. Joudioux and J. Smulevici, \emph{The stability of the Minkowski space for the Einstein-Vlasov system}, arXiv:1707.06141 (2017).
\bibitem{L}
H. Lindblad, \emph{On the asymptotic behaviour of solutions to Einstein's vacuum equations in wave coordinates}, Comm. Math. Phys., 353 (1) (2017), pp. 135-184.
\bibitem{S}
K. Schwarzschild, \emph{{\"U}ber das Gravitationsfeld eines Massenpunktes nach der Einsteinschen Theorie}, Berl. Akad. Ber., 1 (1916), pp. 189-196.
\bibitem{Wald}
R. Wald, \emph{General Relativity}, The University of Chicago Press, 1984.
\bibitem{RW}
T. Regge and J. Wheeler, \emph{Stability of a {S}chwarzschild Singularity}, Phys. Rev., 108 (4) (1957).
\bibitem{Z}
F. Zerilli, \emph{Effective potential for even-parity {R}egge-{W}heeler gravitational perturbation equations}, Phys. Rev. Lett., 24 (13) (1970), pp. 737-738.
\bibitem{Moncrief}
V. Moncrief, 
\emph{Gravitational perturbations of spherically symmetric systems. I. The exterior problem}, Ann. of Phys., 88 (1974), pp. 323-342.
\bibitem{Moncrief2}
-----, \emph{Spacetime symmetries and linearization stability of the Einstein equations. I}, J. Math. Phys., 16 (1975), pp. 493-498.
\bibitem{Holz}
G. Holzegel, \emph{Ultimately Schwarzschildean spacetimes and the black hole stability problem}, arXiv:1010.3216, (2010).
\bibitem{Holzconv}
-----, \emph{Conservation laws and flux bounds for gravitational perturbations of the Schwarzschild metric}, Class. Quant. Grav., 30 (22) (2016).
\bibitem{B--S}
P. Blue and A. Soffer, \emph{The wave equation on the Schwarzschild metric. II. Local decay for the spin-2 Regge--Wheeler equation}, J. Math. Phys., 46 (2005).
\bibitem{Me}
T. Johnson, \emph{The Regge--Wheeler and Zerilli equations}, Report for Imperial College London (2015), Unpublished.
\bibitem{H--K--W}
P.-K. Hung, J. Keller, and M.-T. Wang, \emph{Linear stability of Schwarzschild spacetime: The Cauchy
problem of metric coefficients}, arXiv:1702.02843 (2017).
\bibitem{A--B--W}
L. Andersson, P. Blue and J. Wang, \emph{Morawetz estimate for linearised gravity in Schwarzschild}, arXiv: 1708.06943 (2017).
\bibitem{K--W}
B. Kay and R. Wald, \emph{Linear stability of Schwarzschild under perturbations which are nonvanishinh on the bifurcation sphere}, Class. Quant. Grav., 4 (1987), pp. 893-898.
\bibitem{D--R}
M. Dafermos and I. Rodnianski, \emph{The red-shift effect and radiation decay on black hole spacetimes}, Comm. Pure Appl. Math., 62 (2009), pp. 859-919.
\bibitem{newmethod}
-----, \emph{A new physical-space approach to decay for the wave equation with applications to black hole spacetimes}, arXiv:0910.4957 (2009).
\bibitem{B--St}
P. Blue and J. Sterbenz, \emph{Uniform decay of local energy and the semi-linear wave equation on
Schwarzschild space}, Commun. Math. Phys., 268 (2) (2006), pp. 481-504.
\bibitem{H--V}
P. Hintz and A. Vasy, \emph{The global non-linear stability of the Kerr--De Sitter family of black holes}, 
arXiv:1606.04014 (2016).
\bibitem{H--Vmink}
-----, \emph{A global analysis proof of the stability of Minkowski space and the polyhomogeneity of the metric}, arXiv:1711.00195 (2017).
\bibitem{C}
B. Carter, \emph{Hamilton--Jacobi and Schr\"odinger separable solutions of Einstein's
equations}, Comm. Math. Phys., 10 (4) (1968), pp. 280-310.
\bibitem{H}
P. Hintz, \emph{Non-linear stability of the Kerr--Newman--De Sitter family of charged black holes}, arXiv:1612.04489 (2016).
\bibitem{K--N--DS}
E. Newman, E. Couch, K. Chinnapared, A. Exton, A. Prakash, and R. Torrence, \emph{Metric of a rotating, charged mass}, J. Math. Phys., 6 (6) (1965), pp. 918-
919.
\bibitem{S--T}
O. Sarbach and M. Tiglio, \emph{Gauge-invariant perturbations of Schwarzschild black holes in horizon-penetrating coordinates}, Phys. Rev. D, 64 (8) (2001).
\bibitem{C--O--S}
E. Chaverra, N. Ortiz and O. Sarbach,
\emph{Linear perturbations of self-gravitating spherically symmetric configurations}, Phys. Rev. D, 87 (4) (2013).
\bibitem{P--M}
K. Martel and E. Poisson, \emph{Gravitational perturbations of the Schwarzschild spacetime: A practical
covariant and gauge-invariant formalism}, Phys. Rev. D, 71 (2005).
\bibitem{Sphd}
O. Sarbach, \emph{On the Generalization of the Regge--Wheeler Equation for Self-Gravitating Matter Fields}, PhD thesis at ETH, 2000.
\bibitem{Mphd}
K. Martel, \emph{Particles and black holes: time-domain integration of the equations of black hole perturbation theory}, PhD thesis at University of Guelph.
\bibitem{G--S}
U. Gerlach and U. Sengupta, \emph{Gauge-invariant coupled gravitational, acoustical and electromagnetic on most general spherical space-times}, Phys. Rev. D, 22 (6) (1980), pp. 1300-1312.
\bibitem{Jez}
J. Jezierski, \emph{Energy and angular momentum of the weak gravitational waves on the Schwarzschild background - quasilocal gauge-invariant formulation}, Gen. Rel. Grav., 31 (12) (1999), pp. 1855-1890.
\bibitem{V}
C. Vishveshwara, \emph{Stability of the Schwarzschild metric}, Phys. Rev. D, 1 (1970), pp. 2870-2879.
\bibitem{B--P}
J. Bardeen and W. Press, \emph{Radiation fields in the Schwarzschild background}, J. Math. Phys., 14 (1973), pp. 7-19.
\bibitem{Chand}
S. Chandrasekhar, \emph{On the equations governing the perturbations of the Schwarzschild black hole}, P. Roy. Soc. Lond. A Mat., 343 (1975), pp. 289-298.
\bibitem{Chandbook}
-----, \emph{The Mathematical Theory of Black Holes, Oxford University Press}, Oxford, 3 ed., 1992. 
\bibitem{D}
G. Dotti, \emph{Nonmodal linear stability of the Schwarzschild black hole}, Phys. Rev. Lett., 112 (2014).
\bibitem{F--S}
F. Finster and J. Smoller, \emph{Decay of solutions of the Teukolsky equation for higher spin in the Schwarzschild geometry}, Adv. Theor. Math. Phys., 13 (2009), pp. 71-110.
\bibitem{F--M}
J. Friedman and M. Morris, \emph{Schwarzschild perturbations die in time}, J. Math. Phys., 41 (2000), pp. 7529-7534.
\bibitem{T}
S. Teukolsky, \emph{Perturbations of a rotating black hole. I. Fundamental equations for gravitational,
electromagnetic, and neutrino-field perturbations}, Astrophys. J., 185 (1973), pp. 635-648.
\bibitem{K--S}
S. Klainerman and J. Szeftel, \emph{Global Nonlinear Stability of Schwarzschild Spacetime
under Polarized Perturbations}, arXiv:1711.07597 (2017).
\bibitem{R--T}
I. Robinson and A. Trautman, \emph{Spherical gravitational waves}, Phys. Rev. Lett., 4 (1960), pp. 431– 432.
\bibitem{C--S}
J. Corvino and R. Schoen, \emph{On the Asymptotics of the Vacuum Einstein Constraint Equations}, Jour. Diff. Geom., 73 (2) (2006), pp. 185-217.
\bibitem{Cz}
S. Czimek, \emph{An extension procedure for the constraint equations}, arXiv:1609.08814 (2016), to appear in Ann. of PDE.
\bibitem{Sbierski}
J. Sbierski, \emph{Characterisation of the energy of Gaussian beams on Lorentzian manifolds: with applications to black hole spacetimes}, Anal. PDE, 8 (2015), pp. 1379-1420.
\bibitem{H--V2}
P. Hintz and A. Vasy, \emph{Global analysis of quasilinear wave equations on asymptotically Kerr-de Sitter spaces}, Int. Math. Res. Not., 17 (2016), pp. 5355-5426.
\bibitem{Marz--Met--Tat--Toh}
J. Marzuola, J. Metcalfe, D. Tataru, M. Tohaneanu, \emph{Strichartz estimates on Schwarzschild
black hole backgrounds}, Comm. Math. Phys., 293 (2010), pp. 37-83.
\bibitem{Christ}
D. Christodoulou, \emph{The formation of black holes and singularities in spherically symmetric
gravitational collapse}, Comm. Pure Appl. Math., XLIV (1991),
pp. 339-373.
\bibitem{Friedrich}
H. Friedrich,  \emph{On the hyperbolicity of Einstein's and other gauge field equations},
Comm. Math. Phys., 100 (4) (1985), pp. 525-543.
\bibitem{Huneau}
C. Huneau, \emph{Stability of Minkowski Space-time with a translation space-like
Killing field}, PhD Thesis, UPMC (2016).
\end{thebibliography}
\end{document}